\documentclass{ituthesis}

\usepackage[ruled, vlined, linesnumbered]{algorithm2e}
\usepackage{tikz}
\usepackage{csvsimple}
\usepackage{subcaption}
\usepackage{wrapfig}

\usetikzlibrary{arrows}

\graphicspath{{figures/} {images/} {fn/figures/}}
\settitle{Similarity Problems in High Dimensions\footnote{The research leading to these results has received funding from the European Research Council under the EU 7th Framework Programme, ERC grant agreement no. 614331}}
\setauthor{Johan Sivertsen}
\setsupervisor{Rasmus Pagh}
\setdate{August 2017}

\newcommand{\E}{\mathop{\text{E}\/}}

\newcommand{\modl}{\mathop{\text{ mod$_*$ }\/}}    

\usepackage{xifthen}
\newcommand{\B}[3][]{
  \ifthenelse{\isempty{#2}}{|B_{#1}(#3)|}{B_{#1}(#2,#3)}%
}

\usepackage{enumitem}
\usepackage{xspace}

\newcommand{\simjoin}{{\;\bowtie_\lambda\;}}
\newcommand{\simil}{\ensuremath{\textnormal{sim}}}
\newcommand{\cp}{\textsc{Chosen Path}\xspace}
\newcommand{\cpsj}{\textsc{CPSJoin}\xspace}
\newcommand{\mh}{\textsc{MinHash}\xspace}
\newcommand{\blsh}{\textsc{BayesLSH}\xspace}
\newcommand{\all}{\textsc{AllPairs}\xspace}

\begin{document}
\frontmatter


\thetitlepage
\clearpage
\begin{otherlanguage}{danish}
  \begin{abstract}
    Beregning baseret p\aa~f\ae llestr\ae k mellem datapunkter fra store m\ae ngder h\o jdimensionel data er en hj\o rnesten i mange dele af moderne datalogi, fra kunsting intelligens til informationss\o gning.
    Den store m\ae ngde og kompleksitet af data g\o r, at vi almindeligvis forventer, at der ikke kan findes pr\ae cise svar p\aa~mange udregninger uden uoverstigelige krav til forbruget af enten tid eller plads.
    I denne afhandling bidrager vi med nye eller forbedrede approksimationsalgoritmer og datastrukturer til en r\ae kke problemer der omhandler f\ae llestr\ae k mellem datapunkter.
    Specifikt:
      \begin{itemize}
      \item Vi pr\ae senterer en algoritme der finder en \emph{tiln\ae rmelsesvis fjerneste nabo} hurtigere end med den tidligere hurtigste metode.
      \item Vi kombinerer denne algoritme med de bedste kendte teknikker til \emph{tiln\ae rmelsesvis n\ae rmeste nabo} for at finde en \emph{nabo i tiln\ae rmet ring}.
      \item Vi introducerer den f\o rste ikke-trivielle algoritme til \emph{tiln\ae rmet afstandsf\o lsomt medlemskab} uden falske negativer.
      \item Vi p\aa viser at \emph{indlejringer der bevarer n\ae rmeste nabo} kan udf\o res hurtigere ved at anvende id\'{e}er fra rammev\ae rket udviklet til \emph{hurtige afstandsbevarende indlejringer}.
      \item Vi pr\ae senterer en hurtig ny randomiseret algoritme til \emph{m\ae ngde sammenf\o jning med sammefaldskrav}, flere gange hurtigere end tidligere algorithmer.
  \end{itemize}
  \end{abstract}
\end{otherlanguage}
\clearpage
\begin{abstract}
  Similarity computations on large amounts of high-dimensional data has become the backbone of many of the tasks today at the frontier of computer science, from machine learning to information retrieval.
  With this volume and complexity of input we commonly accept that finding exact results for a given query will entail prohibitively large storage or time requirements, so we pursue approximate results.
  The main contribution of this dissertation is the introduction of new or improved approximation algorithms and data structures for several similarity search problems.
  We examine the furthest neighbor query, the annulus query, distance sensitive membership, nearest neighbor preserving embeddings and set similarity queries in the large-scale, high-dimensional setting.
  In particular:
  \begin{itemize}
  \item We present an algorithm for \emph{approximate furthest neighbor} improving on the query time of the previous state-of-the-art.
  \item We combine this algorithm with state-of-the-art \emph{approximate nearest neighbor} algorithms to address the \emph{approximate annulus query}.
  \item We introduce the first non-trivial algorithm for \emph{approximate distance sensitive membership} without false negatives.
  \item We show that \emph{nearest neighbor preserving embeddings} can be performed faster by applying ideas from the framework of \emph{Fast Distance Preserving Embeddings}.
  \item We introduce and analyse a new randomized algorithm for \emph{set similarity join}, several times faster than previous algorithms.
  \end{itemize}
\end{abstract}
\clearpage
\begin{acknowledgements}
  First and foremost I would like to thank my supervisor Rasmus Pagh.
  My interest in randomized algorithms was first sparked by the lectures he gave at the end of an algorithms course.
  I feel very fortunate to have been able to continue studying and researching with him since then.
  Asides from an astounding talent as a researcher, Rasmus represents a professionalism, kindness and patience that is very rare.

  Through writing the articles that this thesis is based on, I got to work with some great co-authors besides Rasmus.
  I would like to thank Francesco Silvestri, Mayank Goswami, Matthew Skala and Tobias Christiani for the hours of discussion and for sharing in the frustrations and joys of research with me.
  I also had the fortune of being able to spend the spring of 2016 at Carnegie Mellon University in Pittsburgh.
  This was possible thanks to the kindness of Professor Anupam Gupta.  
  I would like to thank Anupam for being an excellent academic host during those months and for many insightful discussions.
  Further I would like to thank the many other welcoming and incredibly gifted people at CMU who made my stay very memorable.

  During the past three years I have enjoyed the pleasant company of my colleagues in the 4b corridor of the ITU.
  I would like to thank everyone here for the seminars, technical discussions, non-technical discussion and lunch conversations over the years.
  Especially I would like to thank the other PhD students in my group, Thomas Ahle, Tobias Christiani and Matteo Dusefante for the adventures, both in and out of Hamming space.
  
  On a personal level, I would like to thank my family and friends for their unquestioning belief in times of doubt.
  Special thanks are due to my parents Michael and Vibeke for laying the foundation I stand on today.
  Finally, I am most grateful to my wife Agnieszka for encouraging me to pursue a PhD and for supporting me always.
  
\end{acknowledgements}
\clearpage
\setcounter{tocdepth}{3}
\tableofcontents*

\midsloppy
\sloppybottom
\mainmatter

\chapter{Introduction}

\section{Similarity Search}
Computers today are increasingly tasked with analyzing complex construct like music or images in ways that are sensible to humans.
However a computer has no more appreciation for a series of bits representing the Goldberg variations than for some representing the sound of repeatedly slamming a car door.
Barring a revolution in artificial intelligence computers have no inherent interpretation of the data they store.
This poses a barrier to the ways computers can help us.

At the same time the amount of digital data has exploded, both in complexity and volume.
Consider as an example the fact that early digital cameras like the 1990 Dycam Model 1 could capture and store 32 low resolution black and white images\footnote{The Dycam 1 featured a 375 x 240 pixel sensor, capturing 256 shades of grey.} and was too expensive for more than a few professional users.
Today a modern smart phone can captures and store thousands of high quality color images, and the number of smart phone users is counted in billions.
These images are of course not only captured, but shared, compared and searched in all manner of ways.
Similar explosive developments have taking place with almost any kind of digital data imaginable, from video and music, to sensor data and network traffic data.

This development means large amounts of data has become cheap and accessible, providing one way of circumventing the barrier:
Given large amounts of available data, computers can learn by example.
Computers are extremely well suited for quickly comparing large amounts of data and figuring out exactly how similar they are.
Consider the task from before: Classify a recording as either ``Bach'' or ``Car door''.
With no concept of music or sound this is a difficult task for a computer.
But if the computer has access to a database of examples from both categories we might simply ask which example is most similar to the recording and return the category of that example.
This is idea behind the \emph{k-Nearest Neighbors}~(k-NN) classifier, a simple but powerful machine learning algorithm.
At the heart of it sits the \emph{Nearest Neighbor}~(NN) problem:
Given a set of data points and a query point, return the point most similar to the query.
This is part of a larger family of problems that might generally be called similarity search problems: Answer questions about a set of points based on the similarity of the points to a query point.
The NN problem is probably the most fundamental similarity search problem, and we will often return to it as it encompasses many of the challenges in the field.
Similarity search problems are vital components of many machine learning techniques, and they are equally important in many other areas of computer science like information retrieval, compression, data mining and image analysis.
The main contribution of this thesis is a series of improvements in solving various similarity search problems, both in the speed and space necessary to solve them, and in the quality of the answers.
Before we can begin to study the problems, we must first address two obvious questions about the definition above:

How did the images, music, traffic data etc. above turn into \emph{points}, and what does it mean for two points to be similar?
Readers familiar with high-dimensional metric spaces and $\Osymbol$-notation can skip ahead to Section~\ref{sec:problems}.

\subsection{Representing data}

To process our data we first need to represent it digitally.
As an example think of a collection of text-only documents i.e. strings of letters and spaces.
In order to store them in a computer, a normal method is to agree to some standard of translating letters into numbers, then store the numbers representing the document on the computer.
Say the documents all contain only $d=2$ letters each.
If we map letters to their index in the alphabet, $a\rightarrow0,b\rightarrow1$ etc., we can represent the data as points in the set $\mathbb{N}^2$: the set of all pairs of natural numbers\footnote{Table~\ref{tab:sets} lists the standard notation for working with sets that we will be using.}.
\begin{figure}
  \centering
  \includegraphics{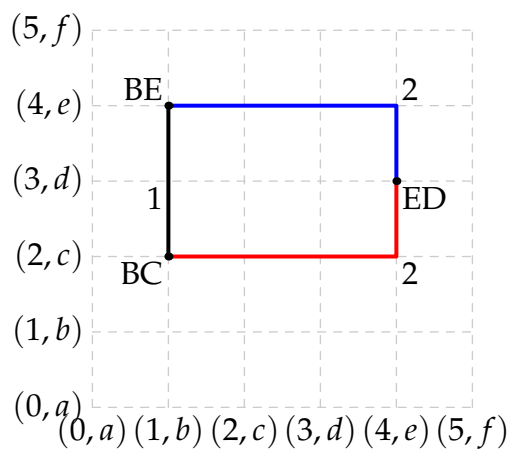}
  \caption{The strings ``BE'', ``BC'' and ``ED'' represented as points in $\mathbb{N}^2$ and the hamming distance between them.}
    \label{fig:representation}
\end{figure}
We call $d$ the dimensionality or features of the data.
\setlength{\intextsep}{0pt}
\begin{wraptable}{O}[2cm]{6cm}
  \begin{tabular}{rr}
    Notation&Description\\
    \hline
    $\mathbb{R}$&Real numbers \\
    $\mathbb{N}$&Natural numbers\\
    $\{0,1\}$& Set of bits\\
    $\emptyset$ &The empty set\\
    $\{\emptyset\}$& Set containing $\emptyset$\\
    $[n]$&Integers from $1$ to $n$\\
    $\B{x}{r}$& Radius $r$ ball around $x$\\
    $X^d$&  $\{(x_1,\ldots,x_d)|x_i\in X\}$\\
    $X \cup Y$& Union of $X$ and $Y$\\
    $X \cap Y$& Intersection of $X$ and $Y$\\
    $\mathcal{P}(X)$ & The power set of $X$\\
    \hline
\end{tabular}
\caption{Set notation}
\label{tab:sets}
\end{wraptable}

Next, we need to define what it means for two documents to be similar?
Often the concept of \emph{similarity} is intuitively understood, but hard to put an exact measure on.
For our purposes we will need exact measures.
In our example, one idea is to consider strings to be similar if they contain the same letters in many positions.
This would suggest using the Hamming distance, $H$, i.e. counting the number of positions where the letters differ, illustrated in Figure~\ref{fig:representation}.
We then have an exact \emph{distance} function that we can use as an inverse measure for similarity.
When data is represented as points in some set $X$ and distances between the points are measured using a distance function $D$ we say that the data is in the space $(X,D)$.
In the example we used $(X=\mathbb{N}^2,D=H)$.
Of course we could have chosen many other distances functions, it depends entirely on the desired notion of similarity.
In this dissertation we will assume that our data is already mapped into a well defined space.
Further, we will assume that the distance function used captures the similarities relevant to the given application.
From now on ``similarity'' will be a formalized, measurable concept, and it will be the inverse of ``distance''.
We will return to this discussion in Section~\ref{sec:distance-functions}.

\begin{table}
  \centering
  \begin{tabular}{rr}
    Notation&Description\\
    \hline
    $S$     &  Input data set\\
    $n$     &  $|S|$ \\
    $q$     &  Query point\\
    $d$     &  Data dimensionality\\
    $D$     &  Distance function\\
    \hline
\end{tabular}
\caption{Frequently used symbols and their meaning.}
\label{tab:notation}
\end{table}

\subsection{Scale and dimension}

Solving problems \emph{at scale} means that we have to be able to keep up with the explosive growth in data.
For most similarity search problems, including the NN problem, we can always answer a query by computing the similarity of the query point, $q$, and every point in the input data set, $S$.
This works well when $S$ is small, but when suddenly the amount of data explodes, so does our query time.
We say that the query time is linear in $n$, where $n$ is the size of $S$.
To handle the explosive growth in data, we must be able to answer the query while only looking at a small part of $S$.
In fact, as $S$ grows, the percentage of $S$ we need to look at must rapidly decrease.
That is, we will be interested in solutions that provide query time sub-linear in $n$.
Imagine that we are given a set of surnames and tasked with building a phone book.
Instead of mapping each letter to a number like in the previous example, we might simply map each name to its alphabetic order.
With this mapping we can represent the strings in a 1-dimensional space, simply points along a line (see Figure~\ref{fig:phonebook}).
\begin{wrapfigure}{l}[2cm]{4cm}
  \centering
  \includegraphics{grid_fig.1}
  \caption{A very short phonebook.}
  \label{fig:phonebook}
\end{wrapfigure}
When we are looking up a name in the phone book we are solving a 1-dimensional search problem.
Using binary search we can solve it in logarithmic time in the number of names.
Logarithmic query time is a very desirable property because it is highly sub-linear.
Roughly speaking, every time the length of the phone book doubles, we will only need to look at one extra name as we search for a number.
This enables us to ``keep up'' with the explosive growth in data (See figure.~\ref{fig:keepup}).
\begin{wrapfigure}{O}[2cm]{6cm}
  \centering
  \includegraphics[width=0.5\textwidth]{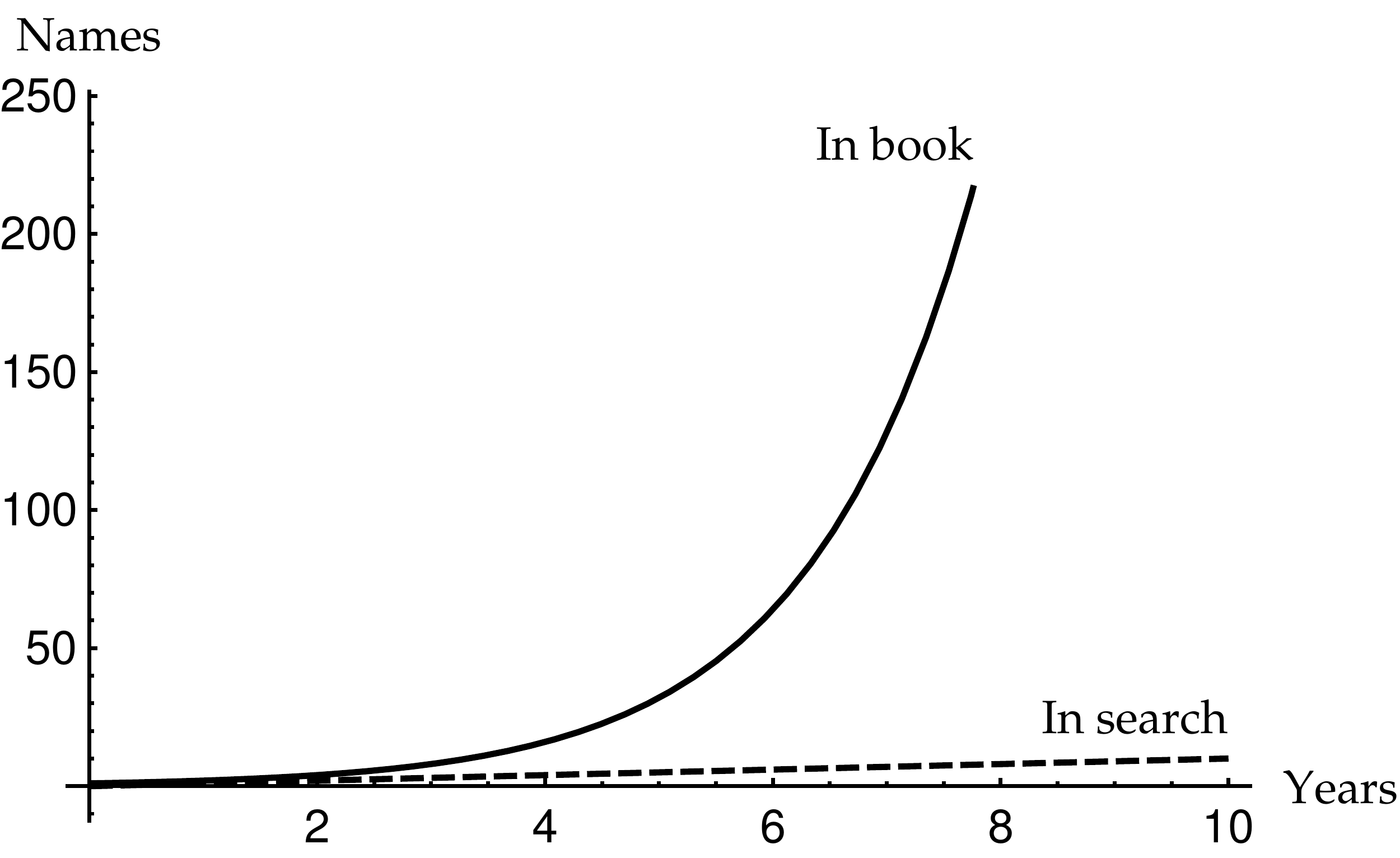}
  \caption{Names in the phone book and how many we will see using a binary search.}
    \label{fig:keepup}
\end{wrapfigure}

If we are trying to capture more complicated relations than a strict ordering, having only one dimension is very limiting.
When understanding if two pieces of music are similar, we might employ a myriad of dimensions, from tempo to scale to the meaning of the lyrics etc.
To capture these complicated relationships we need to work in high-dimensional spaces.
As an example, consider extending the mapping in Figure~\ref{fig:representation} from strings of length $d=2$ to length $d=50$.
While we can no longer easily visualize the space, the mathematical concepts of e.g. $(\mathbb{N}^{50},H)$ are perfectly sound and workable.
However, it is a challenge to develop scaling algorithms when $d$ is large.

For $d=2$ this is already much more difficult.
A classical result in computational geometry is the use of the \emph{Voronoi diagram} (See Figure~\ref{fig:voronoi}) for solving the NN problem in $(\mathbb{R}^2,\ell_2)$.
The diagram partitions $\mathbb{R}^2$ into $n$ cells, one for each data point.
For any location in a cell the nearest point in $S$ is the data point associated with the cell.
Using this diagram we are again able to get logarithmic query time using \emph{point location}: Given a new point $x\in\mathbb{R}^2$, find the associated cell.
Having found the cell, the answer to the NN problem is simply the point associated with that cell.
Both computing the Voronoi diagram and solving the point location problem have long histories and many different approaches have been developed, see~\cite{CGbook08} for an overview. 
\begin{figure}
  \centering
  \includegraphics[width=0.5\textwidth]{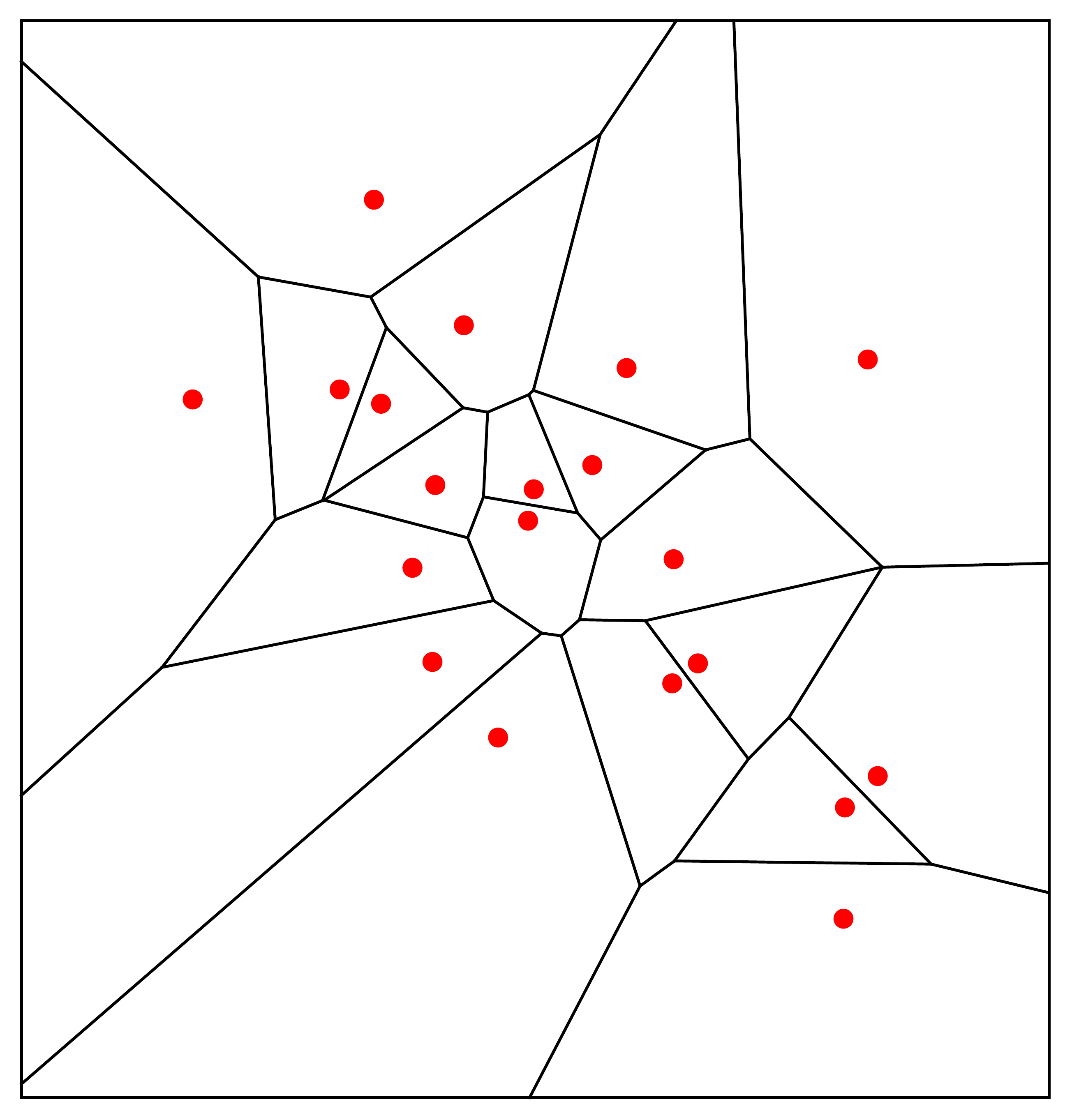}
  \caption{A voronoi diagram for $n=20$ points in $\mathbb{R}^{d=2}$}
    \label{fig:voronoi}
\end{figure}

We can expand the idea of the Voronoi diagram to $d>2$, however the size of the diagram grows exponentially as $n^{d/2}$, so this will only be viable as long as $d$ is small.
In general for low-dimensional metric spaces there are many well-known similarity search algorithms (e.g. \cite{Clarkson1988}, \cite{Bentley1976}) but they all suffer from exponential growth in either storage or query time as $d$ grows.
Although $d$ is not growing as fast as $n$, this is prohibitively expensive even for relatively small $d$.
It is not doing us much good to achieve sub-linear growth in $n$ if all gains are minute to the costs incurred from the high-dimensional setting.
This makes these methods prohibitively expensive for large-scale, high-dimensional data.
There are strong indications that this is not a failing of the solutions, but rather an inherent property of the problem~\cite{Williams05,AbboudR17}.
To avoid this \emph{curse of dimensionality} a field of approximation algorithms has been thriving in recent years. 
Here we concede to losing accuracy in exchange for algorithms that have query time sub-linear in $n$ and linear in $d$.
The space is allowed to grow linearly in $n$ and $d$.
In this thesis we further expand this field with a set of new algorithms for solving approximate similarity search problems in high-dimensional spaces.

\section{Problems and results}
\label{sec:problems}

The articles that make up this dissertation are listed below in the order their content appears here:

\begin{enumerate}
\item Rasmus Pagh, Francesco Silvestri, Johan Sivertsen and Matthew Skala. Approximate Furthest Neighbor in High Dimensions~\cite{Pagh2015a}. SISAP 2015. Chapter~\ref{sec:furthest-neighbor}.
\item Rasmus Pagh, Francesco Silvestri, Johan Sivertsen and Matthew Skala. Approximate furthest neighbor with application to annulus query~\cite{Pagh2015b}. Information Systems 64, 2017. Chapters~\ref{sec:furthest-neighbor} and \ref{sec:annulus-query}.
\item Mayank Goswami, Rasmus Pagh, Francesco Silvestri and Johan Sivertsen. Distance Sensitive Bloom Filters Without False Negatives~\cite{GoswamiP0S17}. SODA 2017. Chapter~\ref{sec:dist-sens-appr}.
\item Johan Sivertsen. Fast Nearest Neighbor Preserving Embeddings. Unpublished. Chapter~\ref{cha:simil-pres-embedd}.
\item Tobias Christiani, Rasmus Pagh and Johan Sivertsen. Scalable and robust set similarity join. Unpublished. Chapter \ref{cha:set-similarity}.
\end{enumerate}

We will state the problems for any space $(X,D)$ but the results are all for particular spaces. (See Table~\ref{tab:sets} and ~\ref{tab:distances}).

Perhaps the most central problem in similarity search is the \emph{nearest neighbor} (NN) problem.
Using the notation in Table~\ref{tab:notation}, we state the problem as:

\paragraph{Nearest Neighbor (NN)}
Given $S\subseteq X$ and $q\in X$. Return $x\in S$, such that $D(q,x)$ is minimized.\\

\noindent
To circumvent the curse of dimensionality we will be relaxing our problems in two ways.
We will use the NN problem to illustrate the relaxations.

First, we will accept an answer $x'$ if it is a $c$-approximate nearest neighbor.
That is, we will require only that $D(q,x')\leq cD(q,x)$, where $x\in S$ is the actual nearest neighbor.
The furthest neighbor and annulus query algorithms presented in this dissertation are generally only applicable when $c>1$.
However, this is often the case: Since we are searching for similar, but not necessarily equal things, the most similar and the \emph{almost} most similar will often be equally useful.

We can also consider cases where the similar thing is much closer (more than a factor $c$) to the query than the rest of the dataset.
In such settings the returned $c$-approximate nearest neighbor is also the actual nearest neighbor.

Secondly we will allow the distance, $r$, to be a parameter to the problem.
We say that $x'$ is $r$-near if $D(q,x')\le r$.
The relaxed approximate near neighbor problem (ANN) is then stated as:
\begin{definition}[$(c,r)$-Approximate Near Neighbor]
  \label{def:ANN}
  For $c>1$, $r>0$. 
  If there exists a point $x\in S$ such that $D(x,q)\leq r$, report some point $x'\in S$ where $D(x',q) \leq cr$, otherwise report nothing.
\end{definition}

\noindent
These relaxations were first introduced by Indyk and Motwani in~\cite{Indyk1998}.
They also show that we can use $(c,r)$-approximate near neighbor to find the $c$-approximate nearest neighbor by searching over settings of $r$.
In many applications achieving a fixed similarity might also suffice on its own, regardless of the existence of closer points.

Next we will introduce the problems addressed in this dissertation.
For each problem we will give a formal definition, as well as a an overview of the main results and ideas used to obtain them.

\subsection {$c$-Approximate Furthest Neighbor (AFN)}

While a lot of work has focused on nearest neighbor, less effort has gone into furthest neighbor.
That is, locating the item from a set that is least similar to a query point.
This problem has many natural applications, consider for example building a greedy set cover by selecting the point furthest from the points currently covered.
Or as we will see in Chapter~\ref{sec:annulus-query} we might use furthest neighbor in combination with near neighbor to find things that are ``just right''.
We formally define the approximate furthest neighbor problem:

\begin{definition}[$c$-Approximate Furthest Neighbor]
  \label{def:c-AFN}
  For $c>1$. Given $S\subseteq X$ and $q\in X$.
  Let $x\in S$ denote the point in $S$ furthest from $q$, report some point $x'\in S$ where $D(x',q) \geq D(x,q)/c$. 
\end{definition}

The furthest neighbor in some sense exhibits more structure than the nearest neighbor.
Consider a point set $S\subseteq\mathbb{R}^2$.
No matter what $q\in X$ is given, the point $x\in S$ furthest from $q$ will be a point on the \emph{convex hull} of $S$ as illustrated in Figure~\ref{fig:hull}.
\begin{wrapfigure}[15]{O}[2cm]{5cm}
  \vspace{0.2cm}
  \centering
  \includegraphics{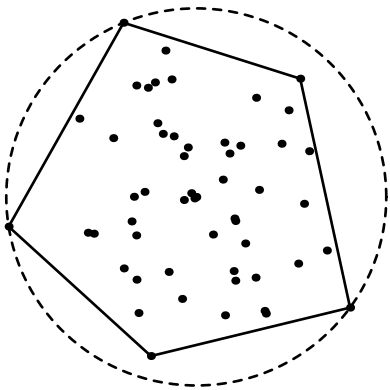}
  \caption{A point set with its convex hull and minimum enclosing ball.}
  \label{fig:hull}
\end{wrapfigure}
If the convex hull is small and easily found an exact result could be efficiently produced by iterating through it.
However, the convex hull can contain $\BOx{n}$ points and in high dimensions they are not easily found.
A way to proceed is to approximate the convex hull, for example by the \emph{minimum enclosing ball}.
This always contains a $\sqrt{2}$-AFN~\cite{Goel2001}, but for $c<\sqrt{2}$ we need a better approximation.

\noindent
In Chapter~\ref{sec:furthest-neighbor} we present an algorithm for $c$-AFN.
We get $\TOx{dn^{1/c^2}}$ query time using $\TOx{dn^{2/c^2}}$ space.
This work is the result of a collaborative effort with Rasmus Pagh, Francesco Silvestri and Matthew Skala.
The work was published as an article~\cite{Pagh2015a}, and later in an extended journal version~\cite{Pagh2015b}.
Here we give a brief high-level introduction to the main result.
Chapter~\ref{sec:furthest-neighbor} also contains analysis for a query independent variation of the data structure, space lower bounds as well as experimental results. 
The main algorithm is similar to one introduced by Indyk \cite{Indyk2003}.
His work introduces a decision algorithm for a fixed radius version of the problem and proceeds through binary search.
We solve the $c$-AFN problem directly using a single data structure.

We use the fact that in $\mathbb{R}^d$, projections to a random vector preserve distances as stated in the following lemma:

\begin{lemma}[See Section 3.2 of Datar et al.~\cite{Datar04}]
For every choice of vectors $x,y \in \mathbb{R}^d$:
\begin{equation*}
\frac{a_i \cdot (x-y)}{\|x-y\|_2}\sim \mathcal{N}(0,1).
\end{equation*}
when the entries in $a_i$ are sampled from the
standard normal distribution $\mathcal{N}(0,1)$. 
\end{lemma}

So we can expect distances between projections to be normally distributed around the actual distance.
Points further from $q$ will generally project to larger values as illustrated in Figure~\ref{fig:Delta}.
This is helpful since it means that we can use well known bounds on the normal distribution to argue about the probability of a point projecting above or below some threshold $\Delta$.
We want to set $\Delta$ so points close to $q$ have a low probability of projecting above it, but points furthest from $q$ still has a reasonably large probability of projecting above $\Delta$.
If we then examine all points projecting above $\Delta$, one of them will likely be a $c-$AFN.
To do this our structure uses a priority queue to pick the points along each random vector with largest projection value, as described in Section~\ref{sec:pq}.

\begin{figure}
\centering
  \includegraphics{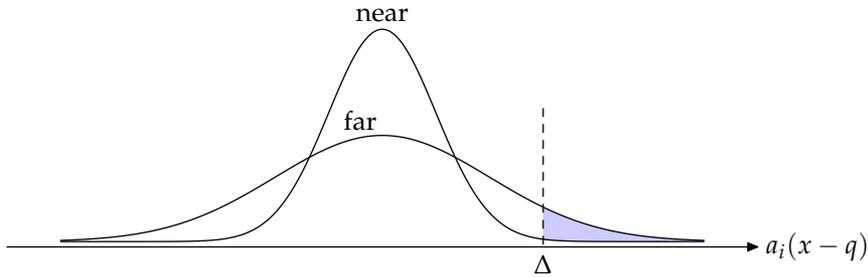}
  \caption{Distribution of $a_i\cdot (x-q)$ for near and far $x$}
\label{fig:Delta}
\end{figure}
\clearpage
\subsection{$(c,r,w)$-Approximate Annulus Query (AAQ)}


Sometimes we want to find the points that are not too near, not too far, but ``just right''.
We could call this the \emph{Goldilocks problem}, but formally we refer it as the \emph{annulus query} problem.
\begin{figure}[h]
  \vspace{0.2cm}
  \centering
\includegraphics[width=0.6\textwidth]{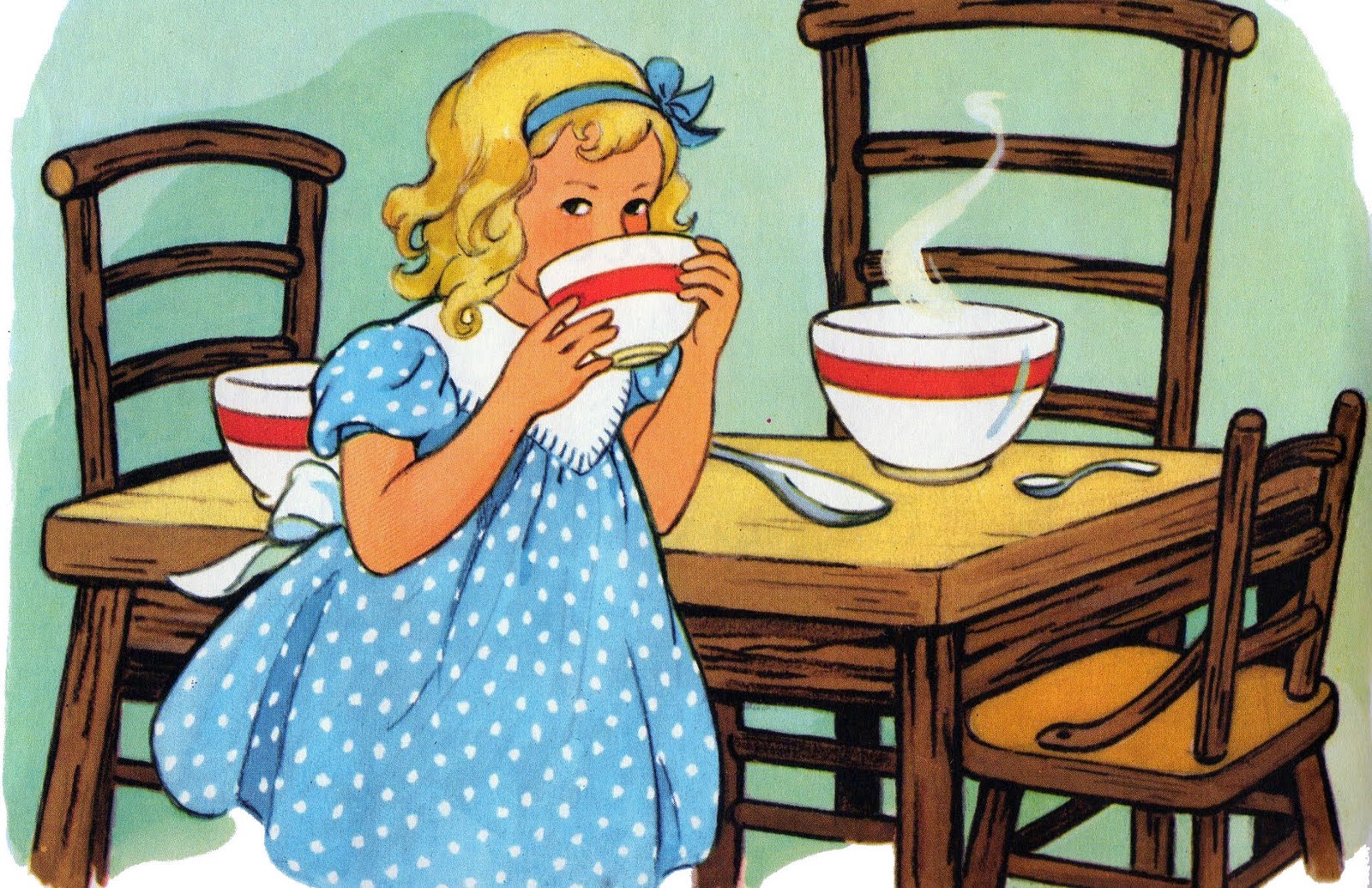}
\caption{Goldilocks finds the porridge that is not too cold, not too hot, but ``just right''.~\scriptsize\copyright Award Publications ltd.}
\label{fig:goldi}
  \vspace{0.2cm}
\end{figure}

Annulus is french and latin for \emph{ring} and the name comes from the shape of the valid area in the plane.
The exact annulus query is illustrated in Figure~\ref{fig:annulus}. Again, we will be working with an approximate version:

\begin{definition}[$(c,r,w)$-Approximate Annulus Query]
  \label{def:aaq}
For $c>1$, $r>0$,$w>1$. 
If there exists a point $x\in S$ such that $r/w\leq D(x,q)\leq rw$ report some point $x'\in S$ where $r/cw\geq D(x',q)\leq crw$, otherwise report nothing.
\end{definition}

A natural way to approach this problem is with a two part solution, one part filtering away points that are too far and the other removing those that are too near.
In Chapter~\ref{sec:annulus-query} we present a solution like this in $(\mathbb{R}^d,\ell_2)$ where we use locality sensitive hashing(LSH, see Section~\ref{sec:lsh}) for the first part and the AFN data structure from Chapter~\ref{sec:furthest-neighbor} for the second.
Using an LSH with gap $\rho$, our combined data structure answers the $(c,r,w)$-Approximate Annulus Query with constant success probability in time $\TOx{dn^{\rho+1/c^2}}$ while using $\TOx{n^{2(\rho+1/c^2)}}$ additional space.
This result was published in the journal version of the AFN paper~\cite{Pagh2015b}.

\label{sec:annulus}
\begin{figure}
\centering
  \includegraphics{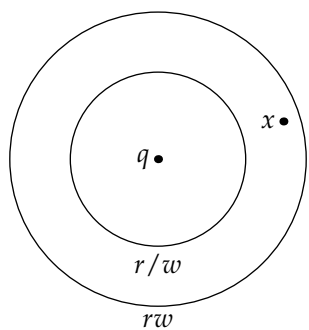}
  \caption{The Annulus query around $q$ returns $x$.}
\label{fig:annulus}
\end{figure}

\subsection{$(r,c,\boldsymbol{\epsilon})$-Distance Sensitive Approximate Membership Query (DAMQ)}

Given a set $S$ and a query point $q$ a \emph{membership query} asks if $q$ is in $S$.
In a famous result from 1970 Burton Bloom showed that the question can be answered using $\BOx{n\log \frac 1 \epsilon }$ space, where $\epsilon$ is the probability of returning a $\emph{false positive}$~\cite{Bloom1970}.
Importantly there are no $\emph{false negatives}$ (See Table~\ref{tab:membershipanswers}).
This one-sided error is of great importance in practice:
If the set $S$ we care about is relatively small in comparison to the universe $X$ and queries are sampled more or less uniformly from $X$, we expect that in most cases $q\notin S$.
Since we never return a false negative, our only errors can occur on the small fraction of queries where we answer ``yes''.

\begin{table}[h]
  \centering
  \begin{tabular}{rll}
    Answer & $x\in S$ & $x\notin S$ \\
    \hline
    ``Yes'' & correct & \emph{false positive}\\
    ``No'' & \emph{false negative} & correct\\
    \hline
\end{tabular}
\caption{Membership query answers and error types.}
\label{tab:membershipanswers}
\end{table}

\vspace{0.3cm}
\noindent
We can then use a secondary data structure to double checks all positive answers.
Since we will use it rarely, we can place the secondary structure somewhere slower to access, but where space is cheaper.
For example on disk, as opposed to in memory. Or on a server somewhere, as opposed to locally.
In this way the one-sided error allows us to use the approximate data structure to speed up most queries, while still giving exact answers.
In a similarity search context we extend membership queries to be distance sensitive.
We want a positive answer when something in $S$ is similar to $q$, although perhaps not an exact match.

\clearpage
\begin{definition}[$(r,c,\boldsymbol{\epsilon})$-Distance Sensitive Approximate Membership Query]
  \label{def:damq}
  For $r>0$, $c\geq1$ and $\epsilon\in[0,1]$. Given $S\subseteq X$ and $q\in X$.
  \begin{itemize}
\item If $\exists x\in S$ such that $D(q,x)\leq r$ report \emph{yes}.
\item If $\forall x\in S$ we have $D(q,x)> cr$ report \emph{no} with probability at least $1-\epsilon$, or \emph{yes} with probability at most $\epsilon$.
  \end{itemize}
\end{definition}

There is some prior work~\cite{Kirsch, Hua2012}, but these solutions yield false positives as well as negatives. 
In Chapter~\ref{sec:dist-sens-appr} we present the first non-trivial solution with one-sided error.
This work was co-authored with Mayank Goswami, Rasmus Pagh and Francesco Silvestri and was published at SODA in 2017~\cite{GoswamiP0S17}.
It turns out that unlike in the regular membership query, it is important to specify what the $\epsilon $ error probability is over.

If $\epsilon $ is over the choice of $q$, the problem seems easier than if it is over the random choices made in the data structure and valid for all $q$.
In the first case we call $\epsilon $ the \emph{average error}, in the latter the \emph{point-wise error}.
For $(\{0,1\}^d,H)$ we present lower bounds (Section~\ref{sec:lower-bounds}) for both cases as well as almost matching upper bounds for most parameter settings (Section~\ref{sec:upper-bounds}).
For a reasonable choice of parameters we get a space lower bound of $\BOMx{n(r/c+\log\frac 1 \epsilon )}$ bits for $\epsilon$ point-wise error.

To construct our upper bounds we represent the points in $S$  with signatures that we construct to have some special properties.
We let $\gamma(x,y)$ denote the \emph{gap} between the signatures of $x$ and $y$.
The value of the gap depends on the distance between the original two points.
Crucially our construction guarantees that when the original distance is less than $r$ the gap is always below a given threshold, but often above it when the original distance is greater $cr$.
We can then answer the query by comparing the query signature to the collections of signatures from $S$.
The space bounds follows from analyzing the necessary length of the signatures.

\subsection{Fast Nearest Neighbor preserving embeddings}

So far we have been trying to circumvent the issues arising from high dimensionality by designing algorithms that give approximate results.
Another was to achieve this is through \emph{dimensionality reduction}.
Broadly speaking the desire here is to find embeddings $\Phi:\mathbb{R}^d\rightarrow \mathbb{R}^k$ with the property that $D(\Phi x,\Phi y)\approx D(x,y)$ and importantly $k\ll d$.

While finding $\Phi$ is not in it self a similarity search problem, it is a way of improving the performance on \emph{all} approximate similarity search problems.
With $\Phi$ we can move a similarity problem from $(\mathcal{R}^d,D)$ into $(\mathcal{R}^k,D)$ and solve it there instead.
In a famous result, Johnson and Lindenstrauss~\cite{Johnson1984} showed a linear embedding from $(\mathbb{R}^d,\ell_2)$ with $k=\BOx{\frac {\log{n}}{\epsilon^2}}$,
while distorting distances by a factor at most $(1+\epsilon )$~(See lemma~\ref{sec:JLlemma}).

The first aspect we might hope improve is getting $k$ even smaller, but it has recently been shown that the original result is optimal~\cite{Larsen17,Larsen16}.
However, Indyk and Naor~\cite{IN07} showed that if we only care about preserving nearest neighbor distances, we can get significantly smaller~$k$.
Specifically, $k$ can be made to depend not on $n$ but on $\lambda_s$, the doubling constant of $S$~(See def.~\ref{def:doubling}).
We call such embeddings \emph{nearest neighbor preserving}~(See def.~\ref{def:NNPE}).

Aside from $k$, an important aspect is of course the time it takes to apply $\Phi$.
We can think of $\Phi$ as an $k\times d$ matrix, so it takes $\BOx{kd}$ time to apply it once.
In 2009 Ailon and Chazelle~\cite{Ailon09} showed that the embedding matrix can be sparse if it is used in combination with some fast distance preserving operations.
If $f$ is the fraction of non-zero entries, this allow us to use fast matrix multiplication to apply the embedding in time $\BOx{kdf}$.
They showed a construction that gets $f=\BOx{\frac{\log n}{d}}$.

In Chapter~\ref{cha:simil-pres-embedd} we show that these two results can be happily married to yield fast nearest neighbor preserving embeddings:

\begin{theorem}[Fast Nearest Neighbor Preserving Embeddings]
  For any $S\subseteq \mathbb{R}^d, \epsilon \in (0,1)$ where $|S|=n$ and
  $\delta\in (0,1/2)$ for some
  \[k=\BO{\frac
    {\log{(2/\epsilon)}}{\epsilon^2}\log{(1/\delta)}\log{\lambda_S}}\]
  there exists a nearest neighbor preserving embedding $\Phi:\mathbb{R}^d\rightarrow\mathbb{R}^k$ with
  parameters $(\epsilon,1-\delta)$ requiring expected \[\BO{d\log(d)+\epsilon^{-2}\log^3{n}}\] operations.
\end{theorem}

The embedding construction is as suggested by~\cite{Ailon09}, but with $k$ bounded as in \cite{IN07}.
Our contribution is in analysing the requirements for nearest neighbor preserving embeddings and showing that they can be fulfilled by this sparse construction.
We also offer some slight improvement to the constants in $f$.

\subsection{$(\lambda,\varphi)$-Set Similarity join}

The join is an important basic operation in databases.
Typically records are joined using one or more shared key values.
The similarity join is a variation where we instead join records if they are sufficiently similar:

\begin{definition}[Similarity Join]
  \label{def:simjoin}
  Given two sets $S$ and $R$ and a threshold $\lambda$, return the set $S\bowtie_\lambda R=\{(x,y)| x\in S, y\in R, D(x,y)\leq\lambda\}$.
\end{definition}

We will look at this problem not for sets of points, but for sets of sets, i.e. \emph{Set Similarity Join}.
To understand this change of setting let us briefly revisit the embedding of strings into $(\mathbb{N}^2,H$) in figure~\ref{fig:representation}.
If the ordering of the letters in each string is irrelevant or meaningless in a given application, the hamming distance seems a poor choice of distance function.
We want $D("ED","DE")$ to be $0$, not $2$.
This is captured by interpreting a string as a set of elements, as illustrated in Figure~\ref{fig:sets}.

\begin{wrapfigure}[8]{O}[2cm]{4cm}
    \vspace{0.2cm}
  \centering
\includegraphics[width=0.25\textwidth]{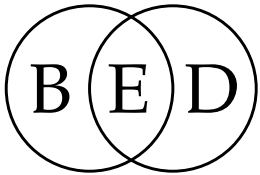}
\caption{The sets $\{B,E\}$ and $\{E,D\}$.}
\label{fig:sets}
\end{wrapfigure}
\noindent
Using the same letter to integer mapping as before, we think for example of ``DE'' as the set $\{3,4\}$ and $S$ as a set of such sets.
In this setting we think of the dimension $d$ as the number of different elements, as opposed to using it for the size of the sets.
In the example $d$ is the size of the alphabet.
We have then moved to the set $\mathcal{P}([d])$ and we switch to using similarity measures (See Section~\ref{sec:distance-functions}).
The Set Similarity Join originates in databases where we might use it to perform entity resolution ~\cite{augsten2013similarity,Chaudhuri_ICDE06, sarawagi2004efficient}.
That is, identify pairs $(x\in S,y\in R)$ where $x$ and $y$ correspond to the same entity.
These can then be used to merge data.
Another popular use of similarity joins in practice is recommender systems.
Here we link two similar, but different, entities in order to use the preferences of one to make recommendations to the other.

In Chapter~\ref{cha:set-similarity} we present a new algorithm, the \cpsj, that solves the set similarity join problem with probabilistic bounds on recall, formalized as:

\begin{definition}[$(\lambda,\varphi)$-Set Similarity Join]
  \label{def:setsimjoin}
  Given two sets of sets $S$ and $R$, a threshold $\lambda\in (0,1)$ and recall probability $\varphi \in (0,1)$.
  Return $L \subseteq S \bowtie_{\lambda} R$ such that for every $(x, y) \in S \bowtie_{\lambda} R$ we have  $\Pr[(x, y) \in L] \geq \varphi$.
\end{definition}

The \cpsj is named after the \textsc{Chosen Path} algorithm~\cite{christiani2017set} for the approximate near neighbor problem.
We can think of the \cpsj as an adaptive version of the \textsc{Chosen Path} algorithm, tailored to the $(\lambda, \varphi)$-Set Similarity Join problem (See Section~\ref{sec:comp-chos-path} for a full comparison).

The core idea is a randomized divide and conquer strategy.
We can view the algorithm as running in a number of steps.
At each step, we may either
\begin{enumerate}
\item solve the problem by brute force, or 
\item divide the problem into smaller problems and handle them in separate steps.
\end{enumerate}

The idea of the \textsc{Chosen Path} algorithm is to perform the division by selecting a random element from $[d]$.
A new subproblem is then formed out of all entities containing that element.
In this way the probability that $x$ and $y$ end up in the same subproblem is proportional to $|x\cap y|$.
This is repeated enough times to get the desired recall.

We can view this process as forming a tree, at each step branching into smaller subproblems, until the leaves are eventually brute forced.
The central question then is at what depth to stop branching.
Building on previous techniques would suggest using either a global worst case depth, $k$, for all points~\cite{christiani2017set,Gionis99,PaghSIMJOIN2015}, or an individual $k_x$ pr. point depth~\cite{ahle2017}.
We develop an adaptive technique that picks out a point when the expected number of comparisons to that point stops decreasing. 
We show that our adaptive strategy has several benefits.
Our main theoretical contribution is showing that the query time is within a constant factor of the individually optimal method.
The \cpsj uses time
	\begin{equation*}
		\tilde{O}\left(\sum_{x \in S}\min_{k_{x}} \left( \sum_{y \in S \setminus \{ x \}} (\simil(x, y) / \lambda)^{k_{x}} + (1/\lambda)^{k_{x}} \right) \right).
	\end{equation*}
        It achieves recall $\varphi = \BOMx{\varepsilon / \log(n)}$ and uses $\BOx{n\log(n)/\epsilon}$ working space with high probability.
        Note that we are trading time against recall \emph{and} space.

We also implemented \cpsj and performed extensive experiments.
Our experiments show speed-ups in the $2-50\times$ range for $90\%$ recall on a collection of standard benchmark data sets.


\section{Preliminaries and techniques}
In this section we introduce the techniques that will be used throughout the thesis.
The section serves both to acknowledge prior work and to highlight new techniques in the context of their priors.
Readers familiar with randomized algorithms and data structures can skip ahead to Section~\ref{sec:furthest-neighbor}.
We will describe the techniques from a high level perspective for the purpose of establishing shared intuition and a common language.

\subsection{Computational Model}

When devising new algorithms we will primarily be interested in their cost in terms of two resources, time and space~(i.e. storage).
To build precise arguments about the cost of a given algorithm we will need a mathematical model for how the algorithm will be carried out.
Here we face a trade-off between the simplicity of the model, the general applicability and the precision.

While these first two demands are somewhat correlated, it is difficult to fulfill all three simultaneously. 
However, our focus is on finding time and space costs that can be used to compare different algorithms and give an insight into their \emph{relative} performance.
Hence precision is of less importance, as long as algorithms are somewhat evenly affected.
Unit cost models are well suited for this task.
We will base our model on the \textsc{real-RAM} model as introduced by Shamos~\cite{Shamos1978}.
We could also use the \textsc{word-RAM} model, but by using full reals we avoid discussing issues of precision that are not at the core of the algorithms.
However we do not have numerically unstable processes and results should hold in both models.
To avoid unrealistic abuses, say by packing the input set into a single real, we do not have a modulo operation or integer rounding.
We model the computers memory as consisting of infinitely many indexed locations $M_i$, each location holding a real number:
\[M=\{(i,M_i\in\mathbb{R})|i\in\mathbb{N}\}.\]
We assume that we can carry out any of the following operations at unit cost:
\begin{itemize}
\item \texttt{Read} or \texttt{Write} any $M_i$.
\item Compare two reals, $\leq,<,=,>,\ge$.
\item Arithmetic operations between two reals $+,-,\cdot,/$.
\item Sample a random variable from a uniform, normal or binomial distribution.
\end{itemize}

The time cost of an algorithm will then be the total number of these operations performed.
When talking about search algorithms we will often split the time cost into \emph{preprocessing-} and \emph{query} time.
All operations that can be carried out without knowledge of the query point(s) are counted as preprocessing.
Query time counts only the remaining operations.
The storage requirement is simply the number of memory locations accessed.
Of course, we will also try to be considerate of other resources, like how complex something is to implement, but we do not include these concerns in the model.

We will be using standard \textsc{\Osymbol}-notation~\cite{Knuth1976} to give bounds in the model.
In short, let $g,f:X\rightarrow\mathbb{R}$:
\begin{itemize}
\item $\BOx{f(x)}$ denotes the set of all $g(x)$ such that there exists positive constants $C$ and $x_0$ with $|g(x)|\leq C f(x)$ for all $x\geq x_0$.
\item $\BOMx{f(x)}$ denotes the set of all $g(x)$ such that there exists positive constants $C$ and $x_0$ with $g(x)\geq C f(x)$ for all $x\geq x_0$.
\end{itemize}
Although $\BOx{f}$ is a set, it is standard to use $g=\BOx{f}$ and ``g is $\BOx{f}$'' to mean $g\in\BOx{f}$.
$\TOx{}$ is used to omit polylog factors.

\subsection{Distance functions and similarity measures}
\label{sec:distance-functions}

We will mostly be formalizing ``similarity'' through the inverse notion of distance.
Given a point in space, similar things will be close, differing things far away.
But we will also sometimes use direct similarity measures.
Table~\ref{tab:distances} contains the distance functions and similarity measures we will be using throughout.
We write $D(\cdot,\cdot)$ for distance functions, and $\simil(\cdot,\cdot)$ for similarity measures.
This is a little confusing, but done for historical reasons.
Both notions are well established in separate branches of mathematics.

The distance functions are central in geometry, dating back to the ancient Greeks.
Most of our work will focus on $\ell_p$ norms, in particular the Euclidean distance $\ell_2$.
For a thorough discussion of the $\ell_p$ norms we refer to~\cite{Rudin1986}.

The similarity measures originated in biology where they where developed to compare subsets of a bounded set, like $[d]$ or the set of all flowers. 
In Chapter \ref{cha:set-similarity} we use Jaccard similarity as well as the Braun-Blanquet variation.
These measures range between $0$ and $1$, with $0$ being no common elements and $1$ being duplicate sets.

The odd space out is the Hamming space.
We could define the Hamming similarity as $(d-H(x,y))/d$, but it is standard in the literature to use Hamming distance.

A practitioner wondering about the correct embedding for a concrete application might use the notion of ``opposite'' as a start.
It is always easy to define equal, but we can only define opposite in a bounded space.
If we are in an unbounded space, say $(\mathbb{R}^d,\ell_2)$, no matter where we would put ``the opposite'' of a point, there is always something a little further away.
If on the other hand it is easy to identify two things as \emph{completely} different, a bounded space is probably the right choice.

\begin{table}
  \centering
  \bgroup
  \def\arraystretch{2}%
  
  \begin{tabular}{lll}
    Name&Input&Distance function\\ 
    \hline
    Hamming distance&$x,y\in X^d$&$H(x,y)=\sum_i^d\begin{cases}1\text{ if } x_i =y_i \\ 0\text{ else } \end{cases}$\\
    Minkowski distance& $x,y\in X^d$&$\ell_p(x,y)=\left(\sum_i^d|x_i-y_i|^p\right)^{1/p}$\\
    Euclidian distance& $x,y\in X^d$&$\ell_2(x,y)=\sqrt{\sum_i^d|x_i-y_i|^2}$\\
    Jaccard similarity& $A,B\subseteq X$ & $J(A,B)=\frac {|A\cap B|}{|A\cup B|}$\\
    Braun-Blanquet similarity & $A,B\subseteq X $ & $BB(A,B)=\frac{ A\cap B}{max(|A|,|B|)}$\\
    \hline
  \end{tabular}
  \egroup
  
\caption{Distance functions and Similarity measures}
\label{tab:distances}

\end{table}

\subsection{Notation}

An overview of the notation used for sets is available in Table~\ref{tab:sets}.
Table~\ref{tab:notation} contains the reserved symbols we use when solving similarity search problems.
In Table~\ref{tab:distances} we list the distance functions and similarity measures used.
For random variables we write $X\sim Y$ when $X$ and $Y$ have the same distribution (See Section~\ref{sec:rand-conc}).

We frequently work with balls, so some special notation for these is helpful.
The $d$ dimensional ball is defined in $(X^d,D)$ as

\[\B[d]{x}{r}=\{p\in X^d | D(p.x)\leq r\}.\]

If we are arguing about the any ball of a given radius we write $B_d(r)$.
We will omit the $d$ subscript when it is clear from the context.

\subsection{Divide and Conquer}

Much of the successful early work in similarity search is based on divide and conquer designs \cite{Bentley1976}. 
The main idea here is to divide $S\subseteq X$ into halves along each dimension and recursively search through the parts until the nearest neighbor is found.
This leads to powerful data structures in low-dimensional spaces, but ultimately also to the amount of work growing exponentially in $d$.
The \emph{kd-tree} is a well known data structure based on this design.
For the NN problem it  promises $\BO{\log{n}}$ query time on random data, $\BO{n \log n}$ preprocessing time and $\BO{n}$ storage \cite{Bentley1975,Friedman1976}, but for high dimensions it converges toward linear query time.

In Chapter~\ref{cha:set-similarity} the paradigm is used to recursively break problems into smaller sub-problems that are then individually handled.
Of course the challenge then is to ensure that all relevant answers to the larger problem emerge as answers in one of the sub problems.
For exact algorithms, like the classical \emph{closest pair in two dimensions}~\cite{Shamos1975} problem, this is handled by checking all possible ways a solution could have been eliminated when generating sub problems.
In Chapter~\ref{cha:set-similarity} we handle it by randomly generating enough sub-problems to give probabilistic guarantees that all close pairs are checked.

\subsection{Randomization and Concentration}
\label{sec:rand-conc}
Algorithms that make random choices, or \emph{Randomized algorithms},  are at the heart of modern similarity search.
Randomization was already important in early work to speed up construction of the $d$ dimensional voronoi diagram~\cite{Clarkson1988}, and it is essential in the later LSH based techniques.
In order to analyse such algorithms we will borrow many ideas and results from the field of probability theory.
We only cover a few of the most used tools here.
See for example~\cite{mitzenmacher2005probability,motwani2010randomized} or~\cite{Feller1968} for an overview.

If our algorithm is to take a random choice it must have access to a source of randomness. 
In reality this will normally be simulated with psudo-random numbers generated by another algorithm, but we will assume that we can access some random process to generate a random event.
Let the sample space, $\Omega$, be the set of all possible outcomes of a random event.

\begin{definition}
  A random variable, $X$, is a real valued function on the sample space $X:\Omega\rightarrow\mathbb{R}$.
  A \emph{discrete} random variable $X$ takes on only a finite or countably infinite number of values.
\end{definition}

We will say that random variable $X$ has cumulative distribution function $F$ if

\[~F(x)=\Pr[X \leq x].\]

When $F(x)$ has the form $\int_{-\infty}^x f(t) dt$, or $\sum_{x'\leq x} f(x')$ for discrete $X$, we say that $X$ has probability density function $f(x)$.
When two random variables $X,Y$ have the same cumultative distribution function, i.e.
\[\forall x\Pr[X \leq x]=\Pr[Y \leq x],\]

it implies that
\begin{equation}
  \label{eq:eqdist}
  \forall x\Pr[X=x]=\Pr[Y=x],
\end{equation}
and we say that $X$ and $Y$ have the same distribution.
We write this as $X\sim Y$.
We avoid using eq.~\ref{eq:eqdist} directly for this, because if $X$ is not discrete $\Pr[X=x]=0$ for all $x$.
For some distributions that we encounter often we use special symbols:
\begin{definition}[The normal distribution]
  We write 
  \[X\sim\mathcal{N}(\mu,\sigma^2).\]
  When $X$ follows the normal distribution with mean $\mu\in\mathbb{R}$ and variance $\sigma^2\in\mathbb{R}$, defined by probability density function:
  \[f(x)=\frac 1 {\sqrt{2\pi\sigma^2}} e^{\frac {(x-\mu)^2} {2\sigma^2}}.\]
\end{definition}
We refer to $\mathcal{N}(0,1)$ as the \emph{standard} normal distribution.
The normal distribution is also called the Gaussian distribution and we sometimes refer to random variables as ``Gaussians'' if they follow the normal distribution.
When building randomized algorithms we often return to the Gaussian distribution.
One reason is that it is \textsc{$2-$stable}~\cite{Zolotarev1986}:

We call a distribution $\mathcal{D}$ over $\mathbb{R}$ \textsc{$p-$stable} where $p\geq0$,
if for any real numbers $v_1,\ldots,v_d$ and for $X,X_1,X_2,\ldots,X_d\sim\mathcal{D}$:
\[ \sum_i^dv_iX_i \sim \left(\sum_i^d|v_i|^p\right)^{1/p}X\]

So for $X,X_1,X_2,\ldots,X_d\sim \mathcal{N}(0,1)$ and some vector $x\in\mathbb{R}^d$ we have $\sum X_ix_i\sim \|x\|_2X \sim Y$ where $Y\sim\mathcal{N}(0,\|x\|^2_2)$.

\begin{definition}[The binomial distribution]
  We write
  \[X\sim\mathcal{B}(n,p).\]
  When $X$ follows the binomial distribution with $n\in \mathbb{N}$ trials and success probability $p\in[0,1]$.
  The probability density function is
 \[f(x)=\binom n x p^x(1-p)^{n-x}.\]
\end{definition}

The binomial distribution can be understood as counting the number of heads in a series of $n$ coin flips, if the coin shows head with probability $p$.
A single flip of the coin is referred to as a Bernoulli trial.
Note that $X$ is then a discrete random variable, the only outcomes are the integers from $0$ to $n$.
This distribution arises often when working in Hamming space due to the binary nature of the space.\\

\noindent
A very powerful tool that we use frequently to analyse random processes is Markov's inequality:

\begin{theorem}[Markov's inequality]
  \label{thm:markov}
  Let $X$ be a non-negative random variable. Then, for all $a>0$,
  \[\Pr[X\geq a]\leq \frac {\E[X]} a\enspace.\]
\end{theorem}

Using Markov's inequlity directly yields useful, but pretty loose bounds.
If we have a good grasp of the moment generating function of $X$, $M(t)=\E[e^{tX}]$, we can get much stronger bounds out of  Markov's inequality.
The idea is to analyse $e^X$ rather than $X$.
Since $e^{X}\geq0$ even if $X<0$ this also expands the range of variables we can use.
We refer to bounds derived this way as ``Chernoff bounds''.
For example for $X\sim\mathcal{B}(n,1/2)$ and $\epsilon >0$ we can use Markov's inequality directly to get

\[\Pr[X\geq (1+\epsilon)\frac n2 ]\leq \frac 1 {1+\epsilon}\enspace.\]

While a Chernoff bound yields exponentially stronger bounds and captures the increasing concentration in $n$,

\[\Pr[X\geq (1+\epsilon)\frac n2 ]\leq e^{-\epsilon^2n/2}\enspace. \]

Even if we fix $n$, this is a lot better as illustrated in Figure~\ref{fig:chernoffpower}.

\begin{figure}[h]
  \centering
  \includegraphics[width=0.5\textwidth]{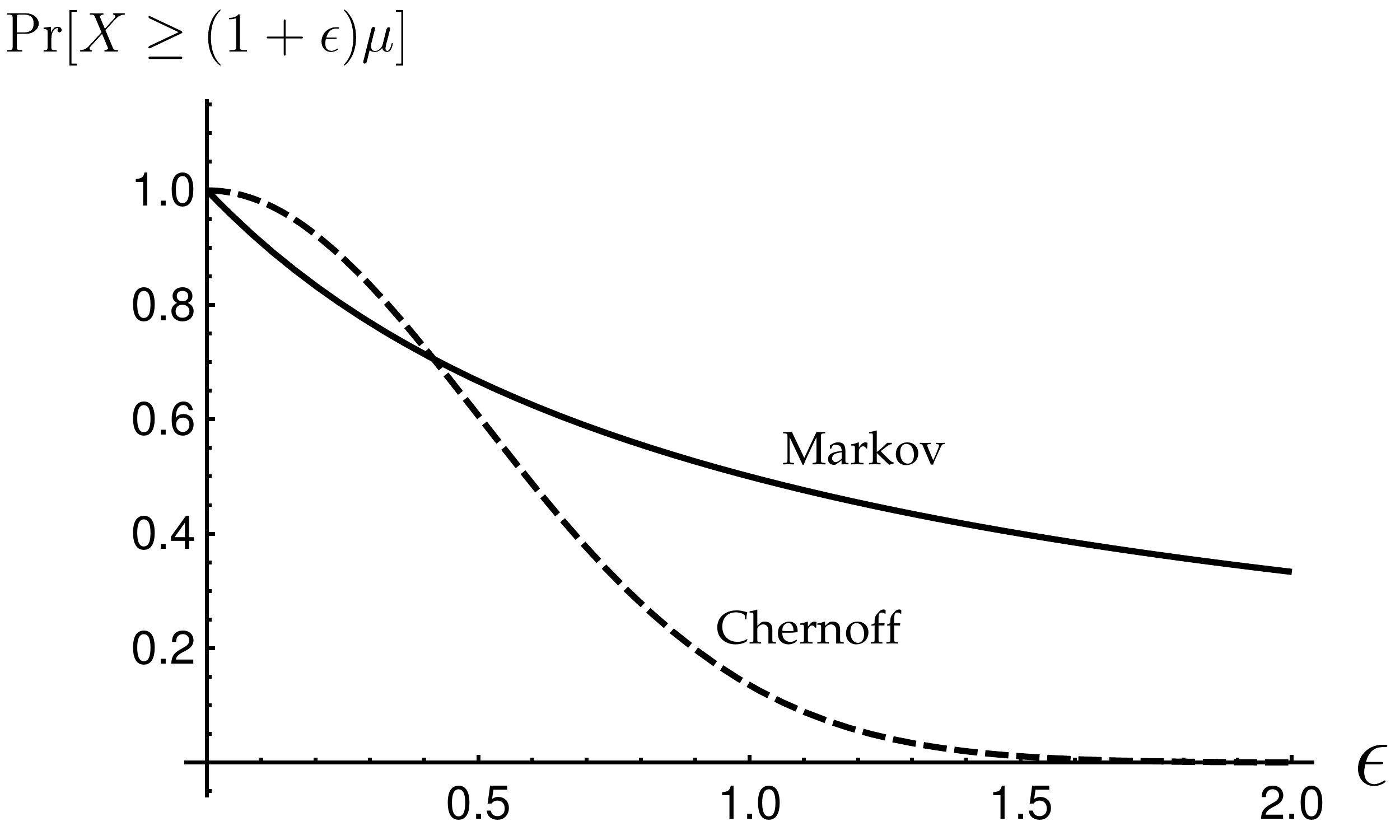}
  \caption{Illustration of Markov and Chernoff type bounds.}
  \label{fig:chernoffpower}
\end{figure}

\subsection{Hashing}
\label{sec:hashing}
We use hashing as an umbrella term for applying functions that map some universe $\mathcal{U}$ into a limited range of $M$ integers.

Often we want functions that spread a large universe evenly over the output range.
This idea was formalized by Carter and Wegman in the notions of \emph{universal hashing}~\cite{Carter1979} and \emph{$k$-independent hashing}~\cite{Carter79II}.
We call a family of hash functions $\mathcal{H}$ \emph{universal} if for a randomly chosen $h\in\mathcal{H}$ and distinct $x_1,x_2\in\mathcal{U}$ and randomly chosen $y_1,y_2\in[M]$:

\[\Pr[h(x_1)=y_1\wedge h(x_2)=y_2]\leq \frac 1 {M^2}\]

And we say that the family is $k$-independent if for any keys $(x_1,\cdots,x_k)\in\mathcal{U}^k$ and any $(y_1,\cdots,y_k)\in[M]^k$:

\[\Pr[h(x_1)=y_1\wedge\cdots\wedge h(x_k)=y_k]=M^{-k}\]

So for $k\geq2$, $k$-independent families are strongly universal.

Another useful property was introduced by Broder et. al.~\cite{Broder1997,Broder2000}:

Let $S_n$ be the set of all permutations of $[n]$.
We say that a family of permutations $\mathcal{F}\subseteq S_n$ is min-wise independent if for any $X\subseteq [n]$ and any $x\in X$, when $\pi$ is chosen at random from $\mathcal{F}$ we have

\[\Pr[\min(\pi(X))=\pi(x)]=\frac 1 {|X|}\enspace.\]

That is, every element of $X$ is equally likely to permute to the smallest value.
We call $\mathcal{H}$ a familiy of \textsc{MinHash} functions if for a random $h\in\mathcal{H}$, $h(X)=min(\pi(X))$ where $\pi$ is a random permutation from a min-wise independent family of permutations.

\textsc{MinHash} functions are very useful in Set Similarity because

\[\Pr[h(x)=h(y)]=\frac {|x \cap y|} {|x \cup y|}=J(x,y)\enspace.\]

Let $X_i=1$ if $h_i(x)=h_i(y)$ and $0$ otherwise.
A Chernoff bound tells us that if $X=\frac 1t\sum_i^tX_i$,
\begin{equation}
  \label{eq:minest}
  \Pr[|X-J(x,y)|\geq \sqrt{\frac{3\ln t}t}J(x,y)]\leq2e^{-J(x,y)\ln{t}}=\frac 2{t^{J(x,y)}}\enspace.
  \end{equation}

  So we can get precise estimates of the Jaccard similarity from a small number of hash functions.
  
  Of course the number of permutations of $[n]$ is $n!$ so in practice we allow $\mathcal{F}\subseteq S_n$ to be \emph{$\epsilon-$min-wise independent}:

\[\Pr[\min(\pi(X))=\pi(x)]\in\frac {1\pm\epsilon}{|X|}\]

In practice we also want hash functions that are fast to evaluate and easy to implement.
Zobrist hashing, or \emph{simple tabulation hashing}, fits this description.
It is $\epsilon-$min-wise independent with $\epsilon$ shrinking polynomially in $|X|$~\cite{Patrascu2012}, $3-$independent and fast in practice~\cite{Thorup2015}.
Tabulation hashing works by splitting keys $x=(x_0,\cdots,x_{c-1})$ into $c$ parts.
Each part is treated individually by mapping it to $[M]$, say with a table of random keys$t_o,\cdots,t_{v-1}:\mathcal{U}\rightarrow [M]$. Finally $h:\mathcal{U}^c\rightarrow [M]$ is computed by:
\[ h(x)=\oplus_{i\in[c]} t_i(x_i)\]
Where $\oplus$ denotes the bit-wise \textsc{XOR} operation.

\subsection{Locality Sensitive Hashing}
\label{sec:lsh}
 
Locality Sensitive Hashing(LSH) is the current state of the art for solving the ANN problem(Definition~\ref{def:ANN}).
The technique was first introduced by Indyk, Gionis and Motwani \cite{Indyk1998,Gionis99} with an implementation that is still the best know for Hamming space.
Since then it has been a subject of intense research.
See~\cite{Andoni2009} for an overview.
The basic idea is to partition the input data using a hash function, $H$, that is sensitive to the metric space location of the input.
This means that the collision probability is larger for inputs close to each other than for inputs that are far apart.
This requirement is normally formalized as:

      \begin{equation}
          \Pr\left[H(u)=H(v)\right]\begin{cases}
            \geq P_1 \mbox{ when } D(u,v) \leq r\\
            \leq P_2 \mbox{ when } D(u,v) \geq cr\\
          \end{cases}
        \end{equation}
        where $P_1>P_2$.
        So the of points in $S$ colliding with $q$ under $H$ are likely near neighbors.
        The key to success for this method is in achieving a large gap between $P_1$ and $P_2$, quantified as $\rho=\frac{-\ln P_1}{-\ln P_2}$(See Figure~\ref{fig:lsh}).
        Ideally, $P_2$ would be $0$, for example by the hash function returning the cell of the voronoi diagram associated with a point.
        But that would trap the function in the curse of dimensionality, either taking up too much space or time.
        So instead we use several functions that each return imperfect partitioning, as illustrated in Figure~\ref{fig:lsh-partitioning}, but are fast to evaluate.

        \begin{wrapfigure}[20]{O}[2cm]{6cm}
          \vspace{0.5cm}
  \centering
  \includegraphics[width=0.4\textwidth]{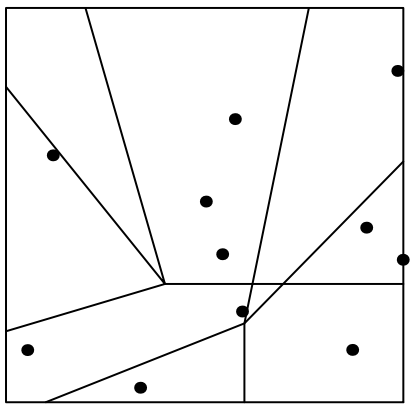}
  \caption{A non-perfect partitioning of points in $\mathbb{R}^2$}
    \label{fig:lsh-partitioning}
\end{wrapfigure}
\noindent
Using a hash function with these properties the $(c,r)$-ANN problem can be solved using $n^{1+\rho+o(1)}$ extra space with  $dn^{\rho+o(1)}$ query time \cite{Har-Peled2012}. 
Recently lower bounds have been published on $\rho$ for the $\ell_1$ \cite{Har-Peled2012} and $\ell_2$ \cite{O'Donnell2014} norm, and a result for $\ell_2$ with $\rho=1/c^2$ has been know for a some years \cite{Andoni2006a}.
In Chapter~\ref{sec:annulus-query} we explore the idea of storing the contents of the LSH buckets in a particular order.
In our case we use projection values onto a random line as approximation of nearness to the convex hull.
However the technique could be expanded to other ways of prioritizing points in scenarios where some subset of the nearest neighbors are of more interest than others.

\begin{figure}
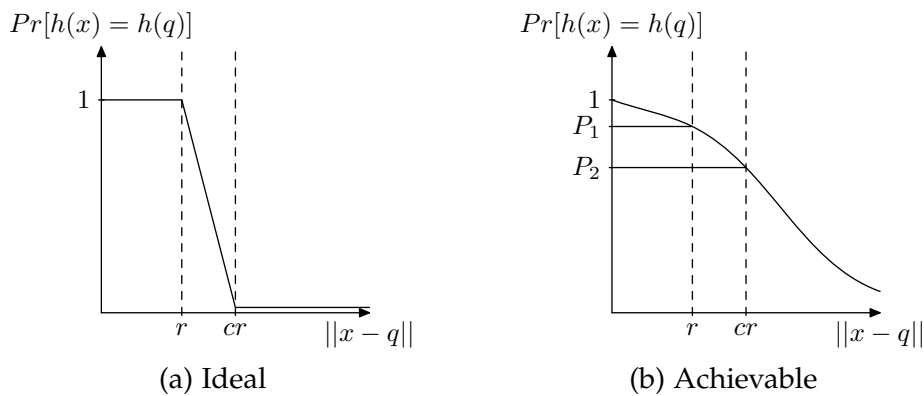

          \centering
  \begin{subfigure}{0.49\textwidth}
    \centering
    \includegraphics{lsh.1}
    \caption{Ideal}
  \label{fig:lsh_left}
\end{subfigure}
  \begin{subfigure}{0.49\textwidth}
    \centering
    \includegraphics{lsh.2}
    \caption{Achievable}
  \label{fig:lsh_right}
\end{subfigure}
  \caption{Ideal vs. achievable LSH function.}
  \label{fig:lsh}
\end{figure}


\chapter{Furthest Neighbor}
\label{sec:furthest-neighbor}

Much recent work has been devoted to approximate nearest neighbor queries. 
Motivated by applications in recommender systems, we consider
\emph{approximate furthest neighbor} (AFN) queries and present a simple,
fast, and highly practical data structure for answering AFN queries in
high-dimensional Euclidean space.
The method builds on the technique of Indyk (SODA 2003), storing random projections to provide sublinear query time for
AFN\@. However, we introduce a different query algorithm, improving on Indyk's approximation 
factor and reducing the running time by a logarithmic factor.  We also present 
a variation based on a query-independent ordering of the database points; while 
this does not have the provable approximation factor of the query-dependent data 
structure, it offers significant improvement in time and space complexity.  
We give a theoretical analysis, and experimental results.

\section{Introduction}
The furthest neighbor query is an important primitive in computational geometry.
For example it can been used for computing the minimum spanning tree or the diameter of a set of points~\cite{AgarwalMS92,Eppstein95}.
It is also used in recommender systems to  create more diverse recommendations~\cite{said2013user,said2012increasing}. 
In this Chapter we show theoretical and experimental results for the \emph{$c$-approximate furthest neighbor} problem ($c$-AFN, Definition~\ref{def:c-AFN}) in $(\mathbb{R}^d,\ell_2)$.
We present a randomized solution with a bounded probability of not returning a $c$-AFN.
The success probability can be made arbitrarily close to $1$ by repetition.

\begin{wrapfigure}{O}[2cm]{7cm}
  \includegraphics[width=0.5\textwidth]{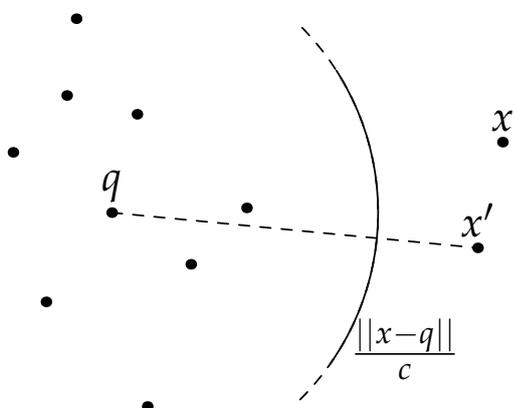}
  \caption{Returning a $(c)$-AFN.}
  \label{fig:fn_query}
\end{wrapfigure}

We describe and analyze our data structures in Section~\ref{sec:alg}.
We propose two approaches, both based on random projections but differing in 
what candidate points are considered at query time.
In the main query-dependent version the candidates will vary depending on the given query,
while in the query-independent version the candidates will be a fixed set.

The query-dependent data structure is presented in Section~\ref{sec:pq}.
It returns the $c$-approximate furthest neighbor, for any $c>1$, with probability at least
$0.72$.  When the number of dimensions is $\BOx{\log n}$, our result
requires $\TOx{n^{1/c^2}}$ time per query and $\TOx{n^{2/c^2}}$ total space,
where $n$ denotes the input size.
Theorem~\ref{thm:space} gives bounds in the general case.
This data structure is closely
similar to one proposed by Indyk~\cite{Indyk2003}, but we use a different approach for the  query algorithm. 

The query-independent data structure is presented in
Section~\ref{sub:query-ind}.  When the approximation factor is a constant
strictly between $1$ and $\sqrt{2}$, this approach requires $2^{\BO{d}}$
query time and space.  This approach is significantly faster than the
query dependent approach when the dimensionality is small.

The space requirements of our data structures are quite high: the query-independent data structure requires space exponential in the dimension, while 
the query-dependent one requires more than linear
space when $c<\sqrt{2}$. 
However, we claim that this bound cannot be significantly improved.
In Section~\ref{sec:lb} we show that any data structure that solves the $c$-AFN by storing a suitable subset of the input points
 must store at least $\min\{n, 2^{\BOM{d}}\}-1$
data points when $c<\sqrt{2}$.

Section~\ref{sec:exp} describes experiments on our data
structure, and some modified versions, on real and
randomly-generated data sets.  In practice, we can achieve approximation
factors significantly below the $\sqrt{2}$ theoretical result,
even with the query-independent version of the algorithm.
We can also
achieve good approximation in practice with significantly fewer projections
and points examined than the worst-case bounds suggested by the theory.  Our
techniques are much simpler to implement than existing methods for
$\sqrt{2}$-AFN, which generally require convex
programming~\cite{clarkson1995vegas,matouvsek1996subexponential}.  Our
techniques can also be extended to general metric spaces.

\subsection{Related work}

\paragraph{Exact furthest neighbor}
In two dimensions the furthest neighbor problem can be solved in linear
space and logarithmic query time using point location in a furthest point
Voronoi diagram (see, for example,~de Berg et al.~\cite{CGbook08}).  However, the space usage of
Voronoi diagrams grows exponentially with the number of dimensions, making
this approach impractical in high dimensions.  More generally, an efficient data
structure for the \emph{exact} furthest neighbor problem in high dimension
would lead to surprising algorithms for satisfiability~\cite{Williams04}, so
barring a breakthrough in satisfiability algorithms we must assume that such
data structures are not feasible.
Further evidence of the difficulty of exact furthest neighbor is the following
reduction: Given a set
$S\subseteq \{-1,1\}^d$ and a query vector $q\in \{-1,1\}^d$, a
furthest neighbor (in Euclidean space) from $-q$ is a vector in $S$ of minimum
Hamming distance to $q$.  That is, exact furthest neighbor is at least as hard
as exact nearest neighbor in $d$-dimensional Hamming space, which seems to be very hard to do in $n^{1-\BOMx{1}}$ without using exponential space~\cite{Williams04,AhlePRS16}.

\paragraph{Approximate furthest neighbor}
Agarwal et al.~\cite{AgarwalMS92} proposes an algorithm for computing the $c$-AFN 
for \emph{all} points in a set $S$ in time $\BO{n/(c-1)^{(d-1)/2}}$ where $n=|S|$ and $1<c<2$.  
Bespamyatnikh~\cite{Bespam:Dynamic} gives a dynamic data structure for $c$-AFN.
This data structure relies on fair split trees and requires $\BO{1/(c-1)^{d-1}}$ time per query and $\BO{dn}$ space, with $1<c<2$.
The query times of both results exhibit an exponential dependency on the
dimension.  Indyk~\cite{Indyk2003} proposes the first approach avoiding this
exponential dependency, by means of multiple random projections of the data and query points to one
dimension.
More precisely, Indyk shows how to solve a
\emph{fixed radius} version of the problem where given a parameter $r$ the
task is to return a point at distance at least $r/c$ given that there exist
one or more points at distance at least $r$.  
Then, he gives a solution to the furthest
neighbor problem with approximation factor $c+\delta$, where $\delta > 0$ is a sufficiently small constant,
by reducing it to queries on many copies of that data structure.  The
overall result is
space $\tilde O(d n^{1+1/c^2})$ and query time $\tilde O(d n^{1/c^2})$, which improved the previous lower bound when $d=\BOM{\log n}$. 
The data structure presented in this chapter shows that the same basic method, multiple random projections to one
dimension, can be used for solving $c$-AFN directly, avoiding the intermediate data structures for the fixed radius version.
Our result is then a simpler data structure that works for all radii and, being
interested in static queries, we are able to reduce the space to
$\TOx{dn^{2/c^2}}$.

\paragraph{Methods based on an enclosing ball}
Goel et al.~\cite{Goel2001} show that a $\sqrt{2}$-approximate furthest
neighbor can always be found on the surface of the minimum enclosing ball of
$S$.  More specifically, there is a set $S^*$ of at most $d+1$ points from
$S$ whose minimum enclosing ball contains all of $S$, and returning the
furthest point in $S^*$ always gives a $\sqrt{2}$-approximation to the
furthest neighbor in $S$.
(See also Appendix~\ref{cha:2afn}).
This method is \emph{query independent} in the
sense that it examines the same set of points for every query.  Conversely,
Goel et al.~\cite{Goel2001} show that for a random data set consisting of $n$
(almost) orthonormal vectors, finding a $c$-approximate furthest neighbor
for a constant $c < \sqrt{2}$ gives the ability to find an
$O(1)$-approximate near neighbor.  Since it is not known how to do that in
time $n^{o(1)}$ it is reasonable to aim for query times of the form
$n^{f(c)}$ for approximation $c < \sqrt{2}$.
We also give a lower bound supporting this view in Section~\ref{sec:lb}.

\paragraph{Applications in recommender systems}
Several papers on recommender systems have investigated the use
of furthest neighbor search~\cite{said2013user,said2012increasing}. 
The aim there was to use furthest neighbor search to create more diverse recommendations.
However, these papers do not address performance issues related to   
furthest neighbor search, which are the main focus of our efforts. 
The data structures presented in this chapter are intended to improve performance in recommender systems relying on furthest neighbor queries.
Other related works on recommender systems include those of
Abbar et al.~\cite{abbar2013real} and Indyk et
al.~\cite{indyk2014composable}, which use core-set techniques to return a
small set of recommendations no two of which are too close.  In turn,
core-set techniques also underpin works on approximating the minimum
enclosing ball~\cite{badoiu2008optimal,KMY03}.

\subsection{Notation}
\label{sec:ball-notation}
In this chapter we will use  $\arg\max_S^m f(x)$ for the set of $m$ elements from $S$ that have the largest values of $f(x)$, breaking ties arbitrarily.



\section{Algorithms and analysis}\label{sec:alg}

\subsection{Furthest neighbor with query-dependent candidates}\label{sec:pq}

Our data structure works by choosing a random line and storing the order of
the data points along it.  Two points far apart on the line are likely
far apart in the original space.  So given a query we can the points
furthest from the query on the projection line, and take those as candidates
for furthest point in the original space.  We build several such data
structures and query them in parallel, merging the results.

Given a set $S\subseteq \mathbb{R}^d$ of size $n$ (the input data),
let $\ell=2n^{1/c^2}$ (the number of random lines) and
$m =1+e^2 \ell \log^{c^2/2-1/3} n $ (the number of candidates to be examined
at query time), where $c>1$ is the desired
approximation factor.  We pick $\ell$ random vectors $a_1,\dots,a_\ell \in
\mathbb{R}^d$ with each entry of $a_i$ coming from the
standard normal distribution $N(0,1)$.  

For any $1\leq i \leq \ell$, we let $S_i = \arg\max^m_{x\in S} a_i\cdot x$
and store the elements of $S_i$ in sorted order according to the value
$a_i\cdot x$.  Our data structure for $c$-AFN consists of $\ell$ subsets
$S_1,\dots,S_\ell \subseteq S$, each of size $m$.  Since these subsets come
from independent random projections, they will not necessarily be disjoint
in general; but in high dimensions, they are unlikely to overlap very much. 
At query time, the algorithm searches for the furthest point from the query
$q$ among the $m$ points in $S_1,\dots,S_\ell$ that maximize $a_i
x- a_i q$, where $x$ is a point of $S_i$ and $a_i$ the random vector used
for constructing $S_i$.  The pseudocode is given in
Algorithm~\ref{alg:basic-query}.  We observe that although the data
structure is essentially that of Indyk~\cite{Indyk2003}, our technique
differs in the query procedure.

\begin{algorithm}
  \KwIn{$\ell$ orderings of the input set $S$. Each $S_{1\leq i\leq\ell}$ referencing $S$ in decreasing order of $a_i\cdot x$. A query point $q$.}
  $\mathit{P} \leftarrow$ An empty priority queue of (point, integer) pairs\;
  $Q \leftarrow $An empty array of reals\;
  $I \leftarrow$An empty array of iterators\;
    \For{$i=1$ to $\ell$}{
      $Q_i\leftarrow a_i\cdot q$\;
      $I_i\leftarrow $An iterator into $S_i$\;
      Retrieve $x$ from $I_i$ and advance $I_i$\;
      Insert $(x,i)$ into $P$ with priority $a_i\cdot x- Q_i$\;
    }
    $\mathit{rval} \leftarrow \bot$\;
    \For{$j=1$ to $m$}{
      $(x,i)\leftarrow $ Highest priority element from $P$\;
      \If{$\mathit{rval} = \bot$ or $x$ is further than $\mathit{rval}$ from $q$}{
        $\mathit{rval} \leftarrow x$
      }
      Retrieve $x$ from $I_i$ and advance $I_i$\;
      Insert $(x,i)$ into $P$ with priority $a_i\cdot x- Q_i$\;
    }
    \Return{$\mathit{rval}$}
  
  \caption{Query-dependent approximate furthest neighbor}\label{alg:basic-query}
\end{algorithm}

Note that early termination is possible if $r$ is known at query time.

\paragraph{Correctness and analysis}
The algorithm examines distances to a set of $m$ points with maximal projection values, we will call the set $S_q$:
$$S_q=\arg\max^m_{x\in\cup S_i}a_i\cdot(x-q), ~ |S_q|= m.$$
We choose the name $S_q$ to emphasize that the set changes based on $q$.
Our algorithm succeeds if and only if $S_q$ contains a $c$-approximate
furthest neighbor. We now prove that this happens with constant
probability. 

We make use of the following standard lemmas that can be
found, for example, in the work of Datar et al.~\cite{Datar04} and Karger,
Motwani, and Sudan~\cite{KMS98}.

\begin{lemma}[See Section 3.2 of Datar et al.~\cite{Datar04}]\label{lem:n0}
For every choice of vectors $x,y \in \mathbb{R}^d$:
\begin{equation*}
\frac{a_i \cdot (x-y)}{\|x-y\|_2}\sim N(0,1).
\end{equation*}
\end{lemma}

\begin{lemma}[See Lemma 7.4 in Karger, Motwani, and Sudan~\cite{KMS98}]\label{lem:normalbound}
For every $t>0$, if $X\sim N(0,1)$ then
\begin{equation*}
\frac{1}{\sqrt{2\pi}}\cdot\left(\frac{1}{t}-\frac{1}{t^3}\right)\cdot e^{-t^2/2}\leq \Pr[X\geq t]\leq \frac{1}{\sqrt{2\pi}}\cdot\frac{1}{t}\cdot e^{-t^2/2}
\end{equation*}
\end{lemma}

The next lemma follows, as suggested by Indyk~\cite[Claims 2-3]{Indyk2003}.
\begin{lemma}
  \label{lem:prob}
Let $p$ be a furthest neighbor from the query $q$ with $r=\|p-q\|_2$, and
let $p'$ be a point such that $\|p'-q\|_2<r/c$.
Let $\Delta = rt/c$ with $t$ satisfying the equation
$e^{t^2/2}t^{c^2}=n/(2\pi)^{c^2/2}$
(that is, $t=\BO{\sqrt{\log n}}$).
Then, for a sufficiently large $n$, we have
\begin{gather*}
\Pr_a\left[a\cdot (p'-q)\geq \Delta\right]\leq \frac{\log^{c^2/2-1/3} n}{n}
\\
\Pr_a\left[a\cdot (p-q)\geq \Delta\right]\geq (1-o(1)) \frac{1}{n^{1/c^2}}\, .
\end{gather*}

\begin{proof}
Let $X\sim \mathcal{N}(0,1)$. By Lemma~\ref{lem:n0} and the right part of
Lemma~\ref{lem:normalbound}, we have for a point $p'$ that
\begin{align*}
\Pr_a\left[a\cdot (p'-q)\geq \Delta\right]
&=\Pr_a\left[X\geq \Delta/\| p'-q\|_2\right]
  \leq \Pr_a\left[X\geq \Delta c/r\right] = 
  \Pr_a\left[X\geq t\right] 
  \\
&\leq \frac{1}{\sqrt{2\pi}} \frac{e^{-t^2/2}}{t}
  \leq \left(t \sqrt{2\pi}\right)^{c^2-1} \frac{1}{n} 
\leq \frac{\log^{c^2/2-1/3} n}{n}.
\end{align*}
The last step follows because $e^{t^2/2}t^{c^2}=n/(2\pi)^{c^2/2}$ implies that $t=\BO{\sqrt{\log n}}$, and holds for a sufficiently large $n$.
Similarly, by Lemma~\ref{lem:n0} and the left part of Lemma~\ref{lem:normalbound},
 we have for a furthest neighbor $p$ that
\begin{align*}
\Pr_a\left[a\cdot (p-q)\geq \Delta\right]
&=\Pr_a\left[X\geq \Delta/\| p-q\|_2\right]
  =\Pr_a\left[X\geq \Delta/r\right]
  =\Pr_a\left[X\geq t/c\right]\\
&\geq \frac{1}{\sqrt{2\pi}} \left(\frac{c}{t}-\left(\frac{c}{t}\right)^3\right){e^{-t^2/(2c^2)}}
  \geq (1-o(1))\frac{1}{n^{1/c^2}}.
\end{align*}
\end{proof}
\end{lemma}

\begin{theorem}
The data structure when queried by Algorithm~\ref{alg:basic-query}
returns a $c$-AFN of a given query with probability
$1-2/e^2>0.72$ in 
\begin{equation*}
\BO{n^{1/c^2}\log^{c^2/2-1/3}{n}(d + \log{n})}
\end{equation*}
time per query.  The data structure requires $\BOx{n^{1+1/c^2}(d + \log
n)}$ preprocessing time and total space
\begin{equation*}
  \BO{\min\left\{dn^{2/c^2}\log^{c^2/2-1/3}n, \,
    dn+n^{2/c^2}\log^{c^2/2-1/3}n\right\}} \, .
\end{equation*}

\begin{proof}
The space required by the data structure is the space required for storing
the $\ell$ sets $S_i$.  If for each set $S_i$ we store the $m\le n$ points and
the projection values, then $\BO{\ell m d}$ memory locations are required.  On
the other hand, if pointers to the input points are stored, then the total
required space is $\BO{\ell m + nd}$.  The representations are equivalent,
and the best one depends on the value of $n$ and $d$.  The claim on space
requirement follows.  The preproceesing time is dominated by the
computation of the $n\ell$ projection values and by the sorting for
computing the sets $S_i$.  Finally, the query time is dominated by the at
most $2m$ insertion or deletion operations on the priority queue  
and the $md$ cost of searching  for the furthest neighbor, $\BO{m(\log{\ell}+d)}$.

We now upper bound the success probability.
Again let $p$ denote a furthest neighbor from $q$ and $r=\|p-q\|_2$.
Let $p'$ be a point such that $\|p'-q\|_2<r/c$, and $\Delta = rt/c$ with $t$ such that $e^{t^2/2}t^{c^2}=n/(2\pi)^{c^2/2}$.
The query succeeds if
\begin{enumerate}
\item $a_i(p-q)\geq \Delta$ for at least one projection vector $a_i$, and 
\item the (multi)set $\hat{S}=\{p' | \exists i : a_i(p'-q)\geq \Delta,\|p'-q\|_2<r/c\}$ contains at most $m -1$ points.
\end{enumerate}
If both (1) and (2) hold, then the size $m$ set of candidates $S_q$ examined by the algorithm must contain the furthest neighbor $p$.
Note that we do not consider points at distance larger than $r/c$ but smaller
than $r$: they are $c$-approximate furthest neighbors of $q$ and can only
increase the success probability of our data structure.

By Lemma~\ref{lem:prob}, (1) holds with probability $1/n^{1/c^2}$.
Since there are $\ell=2n^{1/c^2}$ independent projections, this event fails
to happen with probability at most $(1-1/n^{1/c^2})^{2n^{1/c^2}}\!\leq
1/e^2$.  For a point $p'$ at distance at most $r/c$ from $q$, the
probability that $a_i(p'-q)\geq \Delta$ is less than $(\log^{c^2/2-1/3}
n)/n$ by Lemma~\ref{lem:prob}.  Since there are $\ell$ projections of $n$
points, the expected number of such points is $\ell \log^{c^2/2-1/3} n$. 
Then, we have that $|\hat{S}|$ is greater than $m -1$ with probability at most
$1/e^2$ by the Markov inequality.  Note that a Chernoff bound cannot be used
since there exists a dependency among the projections onto the same random
vector $a_i$.  By a union bound, we can therefore conclude that the
algorithm succeeds with probability at least $1-2/e^2\geq 0.72$.
\end{proof}
\end{theorem}

\subsection{Furthest neighbor with query-independent
candidates}\label{sub:query-ind}

Suppose instead of determining the candidates depending on the query point
by means of a priority queue, we choose a fixed candidate set to be
used for every query.  The $\sqrt{2}$-approximation
the minimum enclosing sphere is one example of such
a \emph{query-independent} algorithm.  In this section we consider a
query-independent variation of our projection-based algorithm.

During preprocessing, we choose $\ell$ unit vectors $y_1,y_2,\ldots,y_\ell$
independently and uniformly at random over the sphere
of unit vectors in $d$ dimensions.  We project the $n$ data points in $S$
onto each of these unit vectors and choose the extreme data point in each
projection; that is,
\begin{equation*}
  \left\{ \left. \arg \max_{x \in S} x\cdot y_i \right| i \in [\ell] \right\} \, .
\end{equation*}

The data structure stores the set of all data points so chosen;
there are at most $\ell$ of them, independent of $n$.
At query time, we check the query point $q$ against all the points we
stored, and return the furthest one.

To prove a bound on the approximation, we
will use the following result of B\"{o}r\"{o}czky and
Wintsche~\cite[Corollary~1.2]{Boroczky:Covering}.
Note that their notation differs from
ours in that they use $d$ for the dimensionality of the surface of the
sphere, hence one less than the dimensionality of the vectors, and $c$ for
the constant, conflicting with our $c$ for approximation factor.  We state the
result here in terms of our own variable names.

\begin{lemma}[See Corollary~1.2 in B\"{o}r\"{o}czky and
Wintsche~\cite{Boroczky:Covering}]\label{lem:cdc}
For any angle $\varphi$ with $0<\varphi<\arccos 1/\sqrt{d}$, in
$d$-dimensional Euclidean space, there exists a
set $V$ of at most $C_d(\varphi)$ unit vectors such that
for every unit vector $u$, there exists some $v \in V$ with the angle between
$u$ and $v$ at most $\varphi$, and 
\begin{equation}
  |V| \le
  C_d(\varphi) = \gamma \cos \varphi \cdot \frac{1}{\sin^{d+1} \varphi}
    \cdot {(d+1)}^{\frac{3}{2}} \ln (1+{(d+1)}\cos^2 \varphi) \, ,
  \label{eqn:cdc}
\end{equation}
where $\gamma$ is a universal constant.
\end{lemma}

Let $\varphi_c=\frac{1}{2}\arccos \frac{1}{c}$; that is half
the angle between two unit vectors whose dot product is $1/c$, as shown in
Figure~\ref{fig:varphi}.
Then by choosing $\ell={O(C_d(\varphi_c) \cdot \log C_d(\varphi_c))}$
unit vectors uniformly at
random, we will argue that with high probability we choose
a set of unit vectors such that
every unit vector has dot product at least $1/c$ with at least one of
them.  Then the data structure achieves $c$-approximation on all queries. 

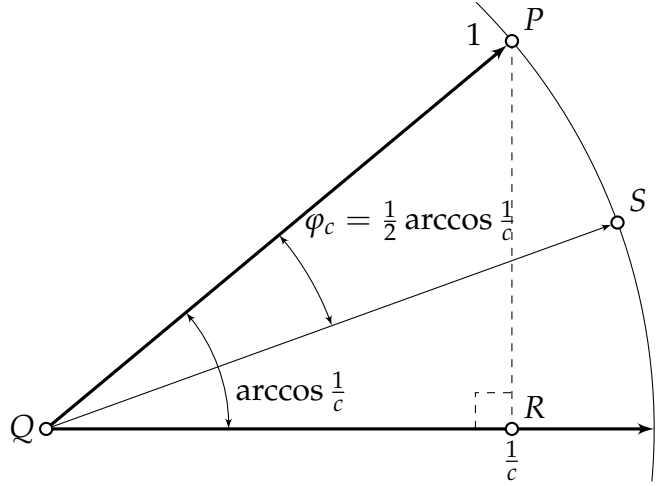
\begin{figure}
\centering
\begin{tikzpicture}[>=latex',scale=1.6]
  \draw[very thick,->] (0,0) -- (5,0);
  \draw[->] (0,0) -- (20:4.95);
  \draw[very thick,->] (0,0) -- (40:4.95);
  \draw[dashed] (40:5) -- (3.83,0);
  \draw[dashed] (3.53,0) -- (3.53,0.3) -- (3.83,0.3);
  \draw[<->] (1.5,0) arc[radius=1.5,start angle=0,end angle=40];
  \draw[<->] (20:2.5) arc[radius=2.5,start angle=20,end angle=40];
  \draw (-5:5) arc[radius=5,start angle=-5,end angle=45];
  \node[anchor=north] at (3.83,0) {$\frac{1}{c}$};
  \node[anchor=south east] at (40:4.8) {1};
  \node[anchor=west] at (12:1.5) {$\arccos \frac{1}{c}$};
  \node[anchor=south west] at (36:2.5)
    {$\varphi_c=\frac{1}{2}\arccos \frac{1}{c}$};
  \draw[black,thick,fill=white] (0,0) circle[radius=0.05];
  \draw[black,thick,fill=white] (20:5) circle[radius=0.05];
  \draw[black,thick,fill=white] (40:5) circle[radius=0.05];
  \draw[black,thick,fill=white] (3.83,0) circle[radius=0.05];
  \node[anchor=east] at (0,0) {$Q$};
  \node[anchor=south west] at (40:5) {$P$};
  \node[anchor=south west] at (20:5) {$S$};
  \node[anchor=south west] at (3.83,0) {$R$};
\end{tikzpicture}
\caption{Choosing $\varphi_c$.}\label{fig:varphi}
\end{figure}

\begin{theorem}
With $\ell=O(f(c)^d)$ for some function $f$ of $c$ 
and any $c$ such that $1<c<2$,
with high probability over the choice of the projection
vectors, the data structure returns a $d$-dimensional $c$-approximate
furthest neighbor on every query.
\end{theorem}

\begin{proof}
Let $\varphi_c=\frac{1}{2} \arccos \frac{1}{c}$.  Then, since $\frac{1}{c}$ is
between $\frac{1}{2}$ and $1$, we can apply the usual half-angle
formulas as follows:
\begin{gather*}
  \sin \varphi_c =
    \sin \frac{1}{2} \arccos \frac{1}{c}
  = \frac{\sqrt{1-\cos \arccos 1/c}}{\sqrt{2}}
  = \frac{\sqrt{1-1/c}}{\sqrt{2}} \\
  \cos \varphi_c =
    \cos \frac{1}{2} \arccos \frac{1}{c}
  = \frac{\sqrt{1+\cos \arccos 1/c}}{\sqrt{2}}
  = \frac{\sqrt{1+1/c}}{\sqrt{2}} \, .
\end{gather*}

Substituting into \eqref{eqn:cdc} from Lemma~\ref{lem:cdc}
gives
\begin{align*}
  C_d(\varphi_c)
  &= \gamma \frac{2^{d/2}\sqrt{1+1/c}}{(1-1/c)^{(d+1)/2}}
    {(d+1)}^{3/2} \ln \left(1+{(d+1)}\frac{1+1/c}{2}\right) \\
  &= O\left( \left( \frac{2}{1-1/c}\right) ^{(d+1)/2}
    d^{3/2} \log d \right) \, .
\end{align*}

Let $V$ be the set of $C_d(\varphi_c)$ unit vectors from
Lemma~\ref{lem:cdc}; every unit vector on the sphere is within angle at most
$\varphi_c$ from one of them.  The vectors in $V$ are the centres of a set
of spherical caps that cover the sphere.

Since the caps are all of equal size and they cover the sphere, there is
probability at least $1/C_d(\varphi_c)$ that a unit vector chosen uniformly
at random will be inside each cap.  Let $\ell= 2 C_d(\varphi_c) \ln
C_d(\varphi_c)$.  This $\ell = O(f(c)^d)$.  Then for each of the caps, the
probability none of the projection vectors $y_i$ is within that cap is
$(1-1/C_d(\varphi_c))^\ell$, which approaches $\exp (-2\ln C_d(\varphi_c)) =
(C_d(\varphi_c))^{-2}$.  By a union bound, the probability that every cap is
hit is at least $1-1/C_d(\varphi_c)$.  Suppose this occurs.

Then for any query, the vector between the query and the true furthest
neighbor will have angle at most $\varphi_c$ with some vector in $V$, and
that vector will have angle at most $\varphi_c$ with some projection vector
used in building the data structure.  Figure~\ref{fig:varphi} illustrates
these steps: if $Q$ is the query and $P$ is the true furthest neighbor, a
projection onto the unit vector in the direction from $Q$ to $P$ would give a
perfect approximation.  The sphere covering guarantees the existence of a
unit vector $S$ within an angle $\varphi_c$ of this perfect projection; and
then we have high probability of at least one of the random projections also
being within an angle $\varphi_c$ of $S$.  If that random projection returns
some candidate other than the true furthest neighbor, the worst case is if
it returns the point labelled $R$, which is still a $c$-approximation.
We have such approximations for all queries simultaneously with high
probability over the choice of the $\ell$ projection vectors. 
\end{proof}

Note that we could also achieve $c$-approximation deterministically, with
somewhat fewer projection vectors, by applying Lemma~\ref{lem:cdc} directly
with $\varphi_c=\arccos 1/c$ and using the centres of the covering caps
as the projection vectors instead of choosing them randomly.  That would
require implementing an explicit construction of the covering, however.
B\"{o}r\"{o}czky and
Wintsche~\cite{Boroczky:Covering} argue that their result is optimal to
within a factor $O(\log d)$, so not much asymptotic improvement is possible.

\subsection{A lower bound on the approximation factor}\label{sec:lb}

In this section, we show that a data structure aiming at an approximation
factor less than $\sqrt{2}$ must use space $\min\{n, 2^{\BOM{d}}\}-1$ on
worst-case data.
The lower bound holds for those data structures that compute the approximate
furthest neighbor by storing a suitable subset of the input points.

\begin{theorem}\label{thm:space}
Consider any data structure $\mathcal D$ that computes the $c$-AFN of an $n$-point input set $S\subseteq \mathbb{R}^d$ by storing a subset of the data set.
If $c=\sqrt{2}(1-\epsilon)$ with $\epsilon\in(0,1)$, then the algorithm must store at least  $\min\{n, 2^{\BOM{\epsilon^2 d}}\}-1$ points.
\end{theorem}
\begin{proof}
Suppose there exists a set $S'$ of size $r=2^{\BOM{\epsilon'^2 d}}$ such that for any $x\in S'$ we have $(1-\epsilon') \leq \|x\|_2^2\leq (1+\epsilon')$ and 
$x\cdot y\leq 2\epsilon'$, with $\epsilon'\in(0,1)$.
We will later prove that such a set exists.
We now prove  by contradiction that any data structure requiring less than $\min\{n, r\}-1$ input points cannot return a $\sqrt{2}(1-\epsilon)$-approximation.

Assume $n\leq r$. Consider the input set $S$ consisting of $n$ arbitrary points of $S'$.
Let the data structure, $\mathcal D\subset S$, be any $n-1$ of these points.
Set the query $q$ to $-x$, where $x\in S\setminus \mathcal D$.
The furthest neighbor is $x$ and it is at distance $\| x - (-x) \|_2\geq 2\sqrt{1-\epsilon'}$. 
On the other hand, for any point $y$ in $\mathcal D$, we get
\begin{equation*}
  \| y - (-x) \|_2 = \sqrt{\|x\|^2_2 + \|y\|^2_2 + 2 x\cdot y}
    \leq \sqrt{2(1+\epsilon') + 4\epsilon'}.
\end{equation*}
Therefore, the point returned by the data structure cannot be better than a $c'$ approximation with 
\begin{equation}\label{eq:approx}
c'= \frac{\| x - (-x) \|_2}{ \| y - (-x) \|_2} \geq 
\sqrt{2} \sqrt{\frac{1-\epsilon'}{1+3\epsilon'}}.
\end{equation}
The claim follows by setting $\epsilon'={(2\epsilon-\epsilon^2)/(1+3(1-\epsilon)^2)}$.

Assume now that $n> r$.  Without loss of generality, let $n$ be a multiple of $r$.
Consider as an input the $n/r$ copies of each vector in $S'$, each copy expanded by a factor $i$ for any $i\in[n/r]$;
specifically, let $S=\{i x | x\in S', i\in[n/r] \}$.
Let $\mathcal D$ be any $r-1$ points from $S$.
Then there exists a point $x\in S'$ such that for every $i\in[1, n/r]$, $i x$ is not in the data structure.
Consider the query $q=-h x$ where $h=n/r$.
The furthest neighbor of $q$ in $S$ is $-q$ and it has distance $\| q - (-q) \|_2\geq 2h\sqrt{1-\epsilon'}$. 
On the other hand, for every point $y$ in the data structure, we get
\begin{equation*}
  \| y - (-hx) \|_2 = \sqrt{h^2\|x\|^2_2 + \|y\|^2_2 + 2 h x\cdot y}
  \leq \sqrt{2h^2(1+\epsilon') + 4h^2\epsilon'}.
\end{equation*}
We then get the same approximation factor $c'$ given in equation~\ref{eq:approx}, and the claim follows.
  
The existence of the set $S'$ of size $r$ follows from the
Johnson-Lindenstrauss lemma~\cite{Matousek:JL}. Specifically, consider an orthornormal base $x_1, \ldots x_r$ 
 of $\mathbb{R}^r$. 
Since $d=\BOM{\log r / \epsilon'^2}$, by the Johnson-Lindenstrauss
lemma there exists a linear map $f(\cdot )$ such that
$(1-\epsilon')\|x_i-x_j\|^2_2\leq \|f(x_i)-f(x_j)\|^2_2\leq
(1+\epsilon)\|x_i-x_j\|^2_2$ and $(1-\epsilon') \leq \|f(x_i)\|_2^2\leq
(1+\epsilon')$ for any $i,j$.  We also have that $f(x_i) \cdot
f(x_j)=(\|f(x_i)\|^2_2 + \|f(x_j)\|^2_2 - \|f(x_i)-f(x_j)\|^2_2)/2$, and
hence $-2\epsilon \leq f(x_i) \cdot f(x_j) \leq 2\epsilon$.  It then
suffices to set $S'$ to $\{f(x_1),\ldots, f(x_r)\}$.

\end{proof}

The lower bound translates into the number of points that must be read by each query. 
However, this does not apply for query dependent data structures.


\section{Experiments}\label{sec:exp}

We implemented several variations of furthest neighbor query in both the C
and F\# programming languages.  This code is available
online\footnote{https://github.com/johanvts/FN-Implementations}.  Our C
implementation is structured as an alternate index type for the SISAP C
library~\cite{SISAP:Library}, returning the furthest neighbor instead of the
nearest.

We selected five databases for experimentation:  the ``nasa''
and ``colors'' vector databases from the SISAP library; two randomly
generated databases of $10^5$ 10-dimensional vectors each, one using a
multidimensional normal distribution and one uniform on the unit cube; and
the MovieLens 20M dataset~\cite{Harper:MovieLens}.  The
10-dimensional random distributions were intended to represent
realistic data, but their intrinsic dimensionality as measured by the $\rho$
statistic of Ch\'{a}vez and Navarro~\cite{Chavez:Intrinsic} is significantly
higher than what we would expect to see in real-life applications.

For each database and each choice of $\ell$ from 1 to 30 and
$m$ from $1$ to $4\ell$, we made 1000 approximate furthest neighbor
queries.  To provide a representative sample over the randomization of both
the projection vectors and the queries, we used 100 different
seeds for generation of the projection vectors, and did 10 queries (each
uniformly selected from the database points) with each seed.  We computed
the approximation achieved, compared to the true furthest neighbor found by
brute force, for every query. The resulting distributions are summarized in
Figures~\ref{fig:uniform}--\ref{fig:movies}.

\begin{figure}
\includegraphics{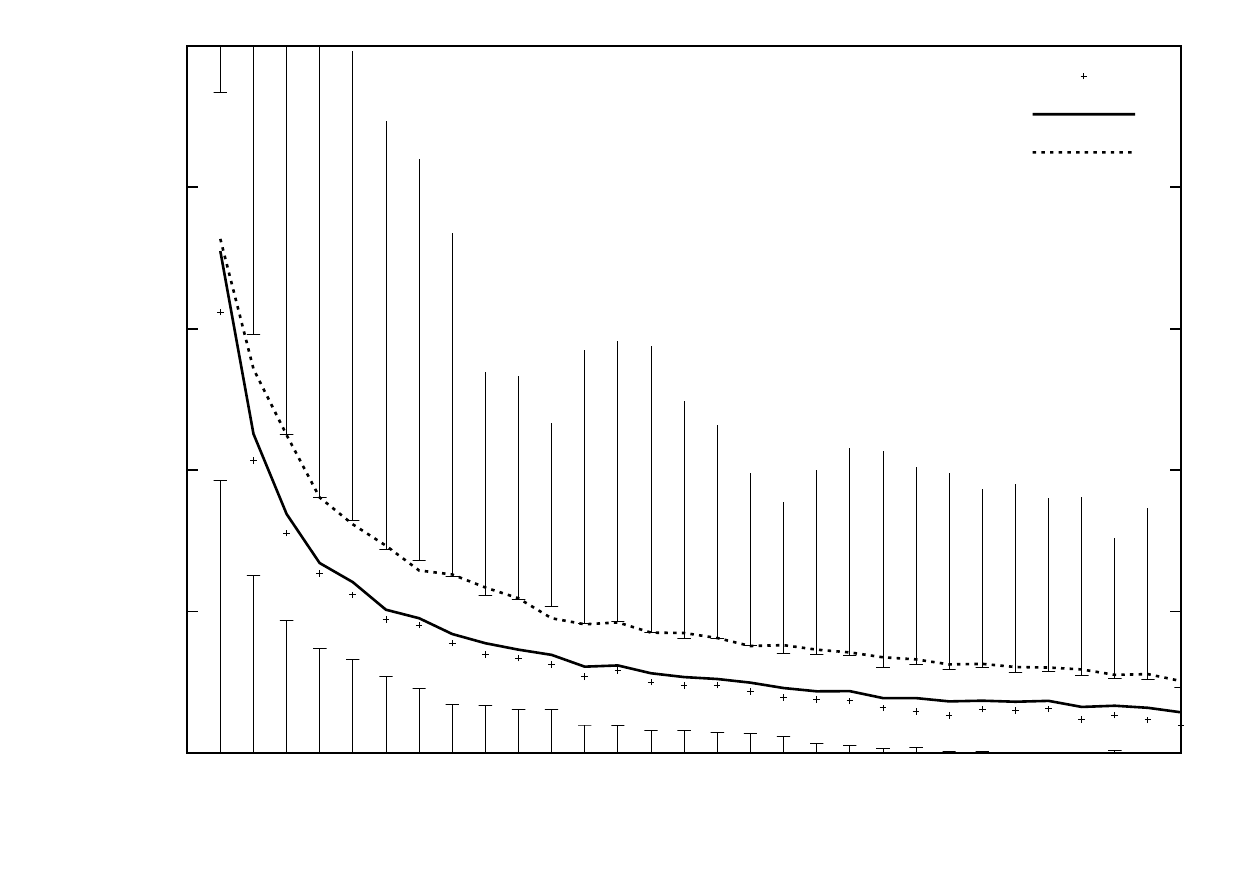}
\caption{Experimental results for 10-dimensional uniform distribution}
\label{fig:uniform}
\end{figure}

We also ran some experiments on higher-dimensional random vector
databases (with 30 and 100 dimensions, in particular) and saw approximation
factors very close to those achieved for 10 dimensions.

\paragraph{$\ell$ vs.\ $m$ tradeoff}

The two parameters $\ell$ and $m$ both improve the approximation as they
increase, and they each have a cost in the time and space bounds. 
The best tradeoff is not clear from the analysis. 
We chose $\ell=m$ as a typical value, but we
also collected data on many other parameter choices. 

\begin{figure}
\includegraphics{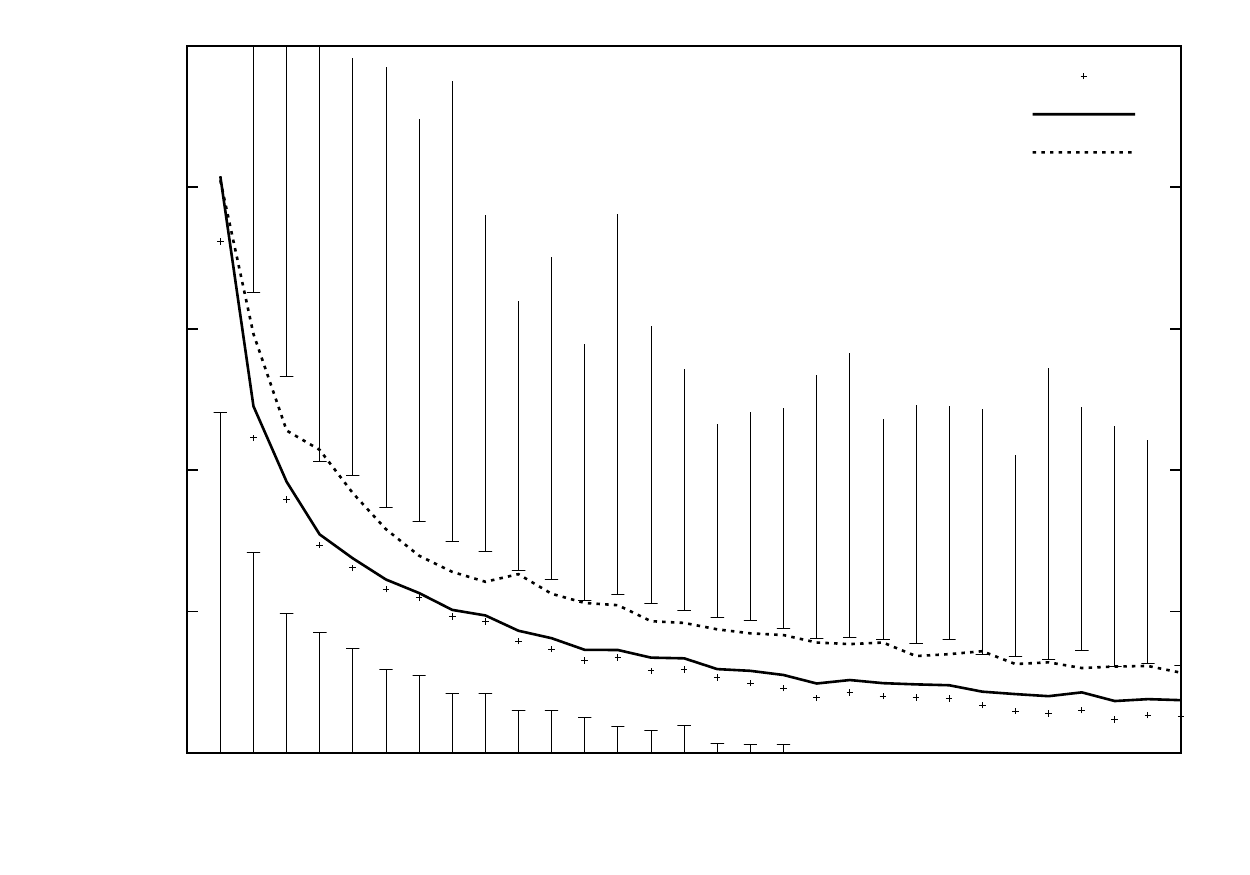}
\caption{Experimental results for 10-dimensional normal distribution}
\label{fig:normal}
\end{figure}

\begin{figure}
\includegraphics{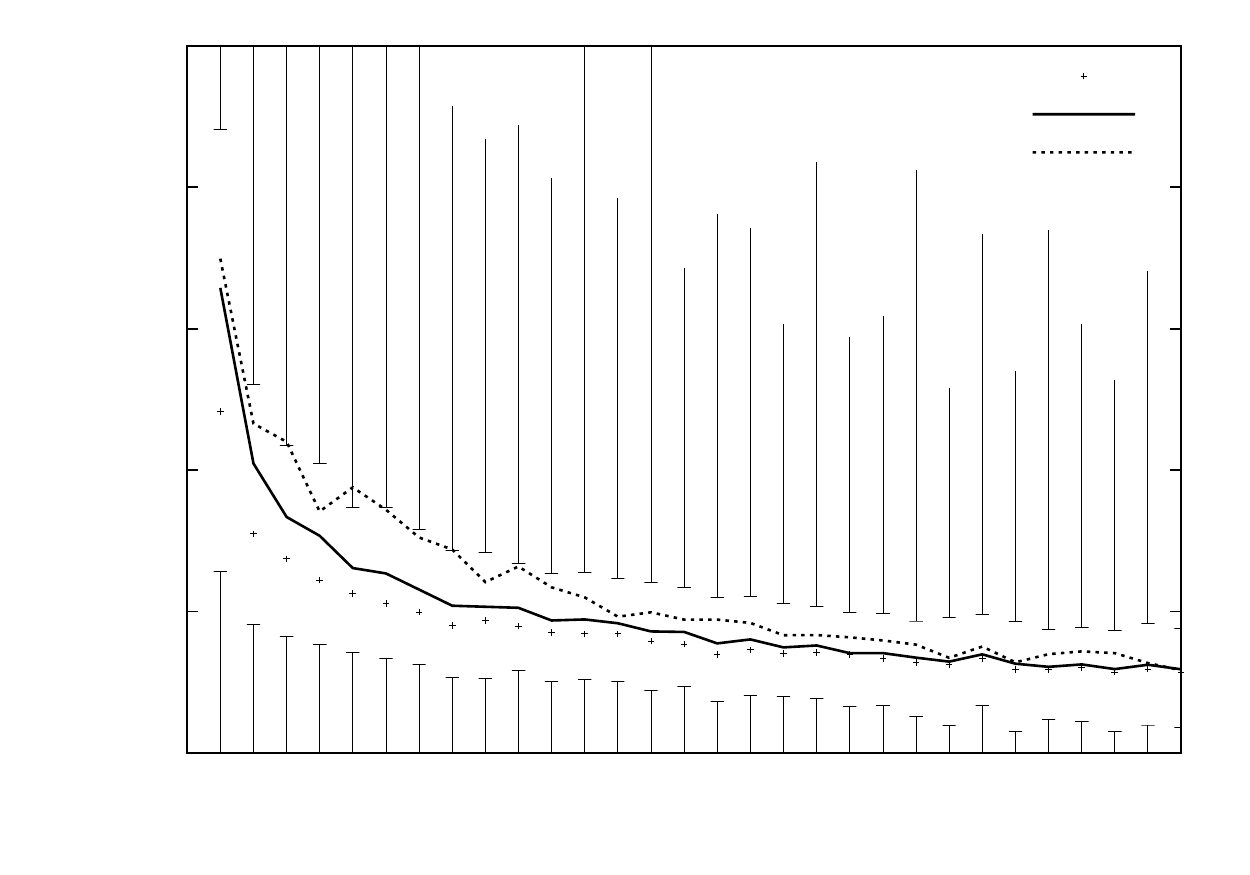}
\caption{Experimental results for SISAP nasa database}
\label{fig:nasa}
\end{figure}

\begin{figure}
\includegraphics{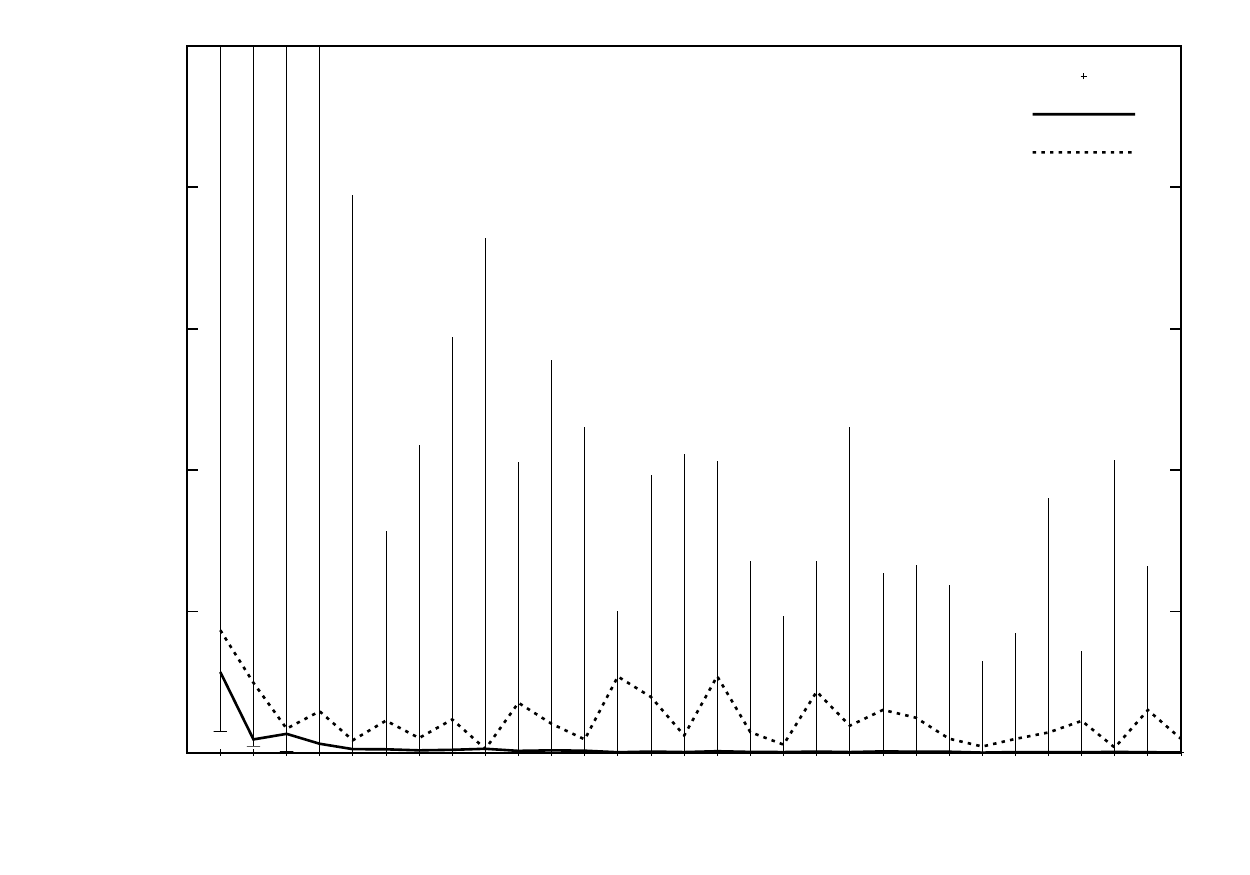}
\caption{Experimental results for SISAP colors database}
\label{fig:colors}
\end{figure}

\begin{figure}
\includegraphics{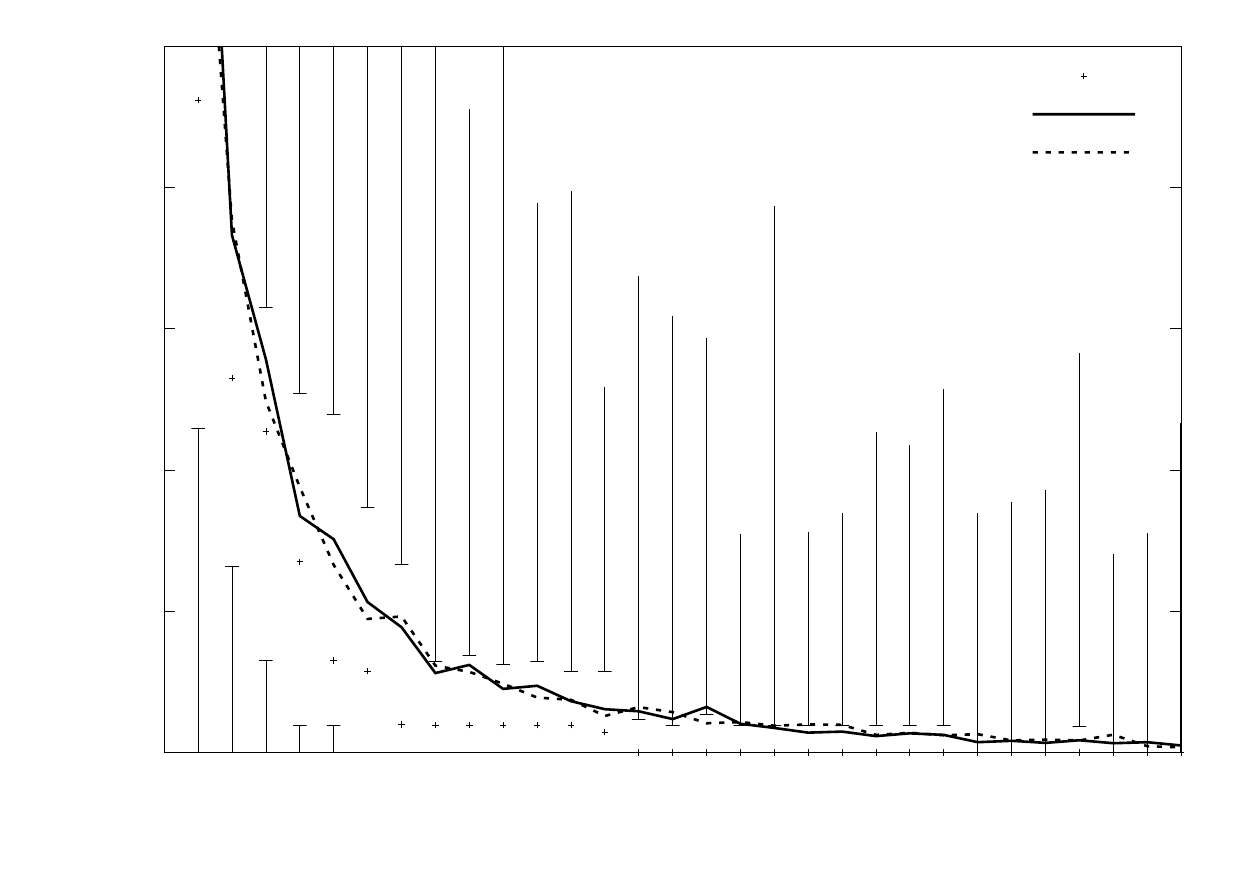}
\caption{Experimental results for MovieLens 20M database}
\label{fig:movies}
\end{figure}

\begin{figure}
\includegraphics{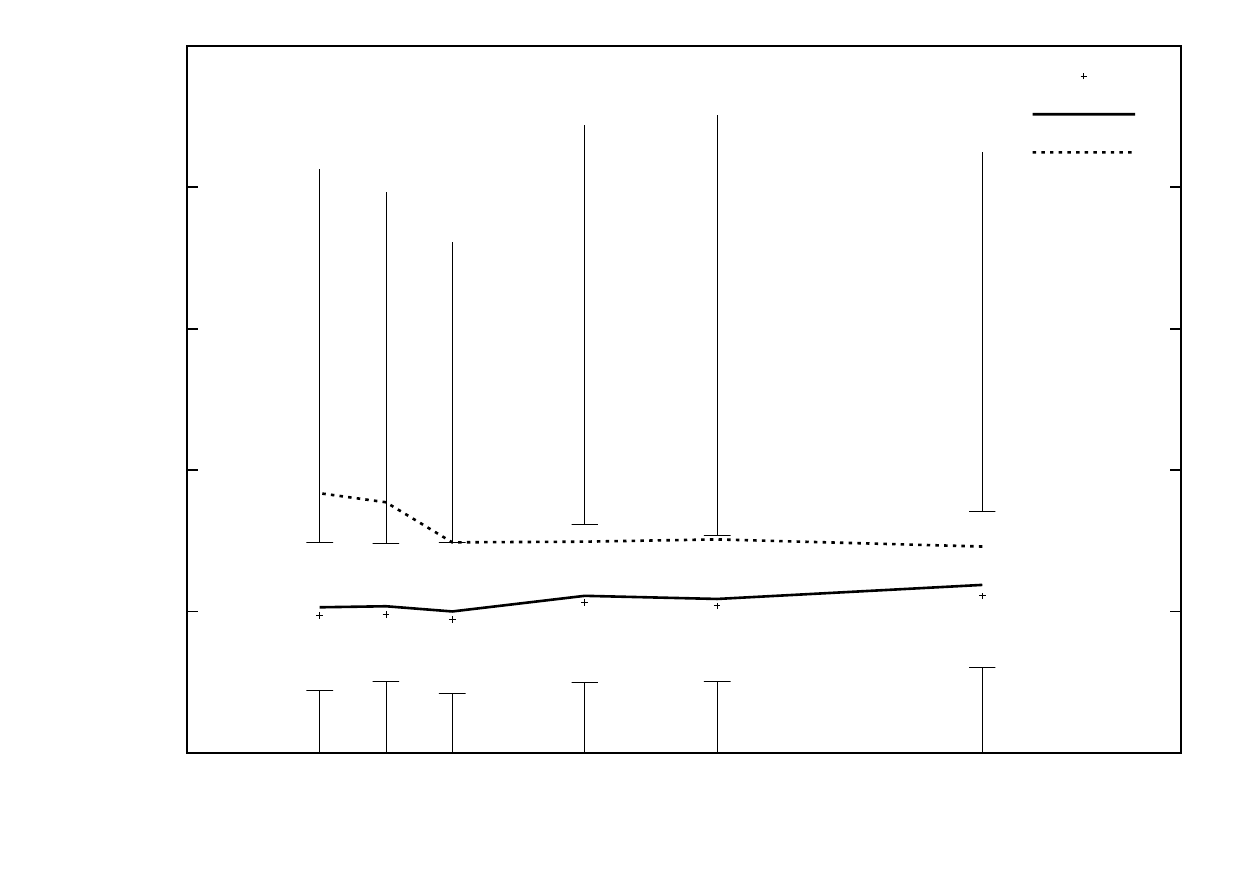}
\caption{The tradeoff between $\ell$ and $m$ on 10-dimensional normal
vectors}
\label{fig:tradeoff}
\end{figure}

Figure~\ref{fig:tradeoff} offers some insight into the tradeoff:
since the cost of doing a query is roughly proportional to both $\ell$ and
$m$, we chose a fixed value for their product, $\ell \cdot m=48$, and
plotted the approximation results in relation to $m$ given that, for the
database of normally distributed vectors in 10 dimensions.
As the figure shows, the approximation factor does not change
much with the tradeoff between $\ell$ and $m$.

\paragraph{Query-independent ordering}

The furthest-neighbor algorithm described in Section~\ref{sec:pq} examines
candidates for the furthest neighbor in a \emph{query dependent} order.
In order to compute the order for arbitrary queries, we must store
$m$ point IDs for each of the $\ell$ projections, and use a priority queue
data structure during query, incurring some costs in both time and space. 
It seems intuitively reasonable that the search will usually examine points
in a very similar order regardless of the query: first those that are
outliers, on or near the convex hull of the database, and then working its
way inward.  

We implemented a modified version of the algorithm in which the index stores
a single ordering of the points.  Given a set $S\subseteq \mathbb{R}^d$ of
size $n$, for each point $x \in S$ let
$\mathit{key}(x)=\max_{i \in 1\ldots\ell} a_i\cdot x$.  The key for each
point is its greatest projection value on any of the $\ell$
randomly-selected projections.  The data structure stores points (all of
them, or enough to accomodate the largest $m$ we plan to use) in order of
decreasing key value:  $x_1$, $x_2$, $\ldots$ where $\mathit{key}(x_1) \ge
\mathit{key}(x_2) \ge \cdots$.  Note that this is not the
same query-independent data structure discussed in
Section~\ref{sub:query-ind}; it differs both in the set of points
stored and the order of sorting them.

The query 
examines the first $m$ points in
the \emph{query independent} ordering and returns the one furthest from the
query point.  Sample mean approximation factor for this algorithm in our
experiments is shown by the dotted lines in
Figures~\ref{fig:uniform}--\ref{fig:tradeoff}.

\begin{algorithm}[t]
\KwIn{The input set $S$ sorted by $\max_{i \in 1\ldots\ell} a_i\cdot x$. A query point $q$.}
$\mathit{rval} \leftarrow \bot$\;
\For{$x\in$ The top $m$ elements of $S$}{
  \If{$\mathit{rval} = \bot$ or $x$ is further than $\mathit{rval}$ from $q$}{
    $\mathit{rval} \leftarrow x$
  }
}
return $\mathit{rval}$
\caption{Query-independent approximate furthest
neighbor}\label{alg:independent}

\end{algorithm}

\paragraph{Variations on the algorithm}

We have experimented with a number of practical improvements to the
algorithm.  The most significant is to use the rank-based \emph{depth} of
projections rather than the projection value.  In this variation we sort the
points by their projection value for each $a_i$.  The first and last point
then have depth 0, the second and second-to-last have depth 1, and so on up
to the middle at depth $n/2$.  We find the minimum depth of each point over
all projections and store the points in a query independent order using the
minimum depth as the key.  This approach seems to give better results in
practice.  A further improvement is to break ties in the minimum depth by
count of how many times that depth is achieved, giving more priority to
investigating points that repeatedly project to extreme values.  Although
such algorithms may be difficult to analyse in general, we give some results
in Section~\ref{sub:query-ind} for the case where the data structure
stores exactly the one most extreme point from each projection.

The number of points examined $m$ can be chosen per query and even
during a query, allowing for interactive search.  After returning the best
result for some $m$, the algorithm can continue to a larger $m$ for a
possibly better approximation factor on the same query.  The smooth tradeoff
we observed between $\ell$ and $m$ suggests that choosing an $\ell$ during
preprocessing will not much constrain the eventual choice of $m$.

\paragraph{Discussion}

The main experimental result is that the algorithm works very well for the
tested datasets in terms of returning good approximations of the furthest
neighbor.  Even for small $\ell$ and $m$ the algorithm returns good
approximations.  Another result is that the query independent variation of
the algorithm returns points only slighly worse than the query dependent. 
The query independent algorithm is simpler to implement, it can be queried
in time $\BO{m}$ as opposed to $\BO{m \log{\ell+m}}$ and uses only $\BO{m}$
storage.  In many cases these advances more than make up for the slightly
worse approximation observed in these experiments.  However, by
Theorem~\ref{thm:space}, to guarantee $\sqrt{2}-\epsilon$ approximation the
query-independent ordering version would need to store and read $m=n-1$
points.

In data sets of high intrinsic dimensionality, the furthest point from a
query may not be much further than any randomly selected point, and we can
ask whether our results are any better than a trivial random selection from
the database.  The intrinsic dimensionality statistic $\rho$ of Ch{\'a}vez
and Navarro~\cite{Chavez:Intrinsic} provides some insight into this
question.  Note that instrinsic dimensionality as measured by $\rho$ is not
the same thing as the number of coordinates in a vector. For real data sets it is often much smaller than that.  Intrinsic dimensionality also applies to
data sets that are not vectors and do not have coordinates.  Skala proves
a formula for the value of $\rho$ on a multidimensional normal
distribution~\cite[Theorem~2.10]{Skala:Dissertation}; it is
$9.768\ldots$ for the 10-dimensional distribution used in
Figure~\ref{fig:normal}.  With the definition $\mu^2/2\sigma^2$, this means
the standard deviation of a randomly selected distance will be about 32\% of
the mean distance.  Our experimental results come much closer than that to
the true furthest distance, and so are non-trivial.

The concentration of distances in data sets of high intrinsic dimensionality
reduces the usefulness of approximate furthest neighbor.  Thus, although we
observed similar values of $c$ in higher dimensions to our 10-dimensional
random vector results, random vectors of higher dimension may represent a
case where $c$-approximate furthest neighbor is not a particularly
interesting problem.  However, vectors in a space with many dimensions
but low intrinsic dimensionality, such as the colors database, are
representative of many real applications, and our algorithms performed well
on such data sets.

The experimental results on the MovieLens 20M data
set~\cite{Harper:MovieLens}, which were not included in the conference
version of the present work, show some interesting effects resulting from
the very high nominal (number of coordinates) dimensionality of this data
set.  The data set consists of 20000263 ``ratings,'' representing the
opinions of 138493 users on 27278 movies.  We treated this as a database of
27278 points (one for each movie) in a 138493-dimensional Euclidean space,
filling in zeroes for the large majority of coordinates where a given user
did not rate a given movie.  Because of their sparsity, vectors in this data
set usually tend to be orthogonal, with the distance between two simply
determined by their lengths.  Since the vectors' lengths vary over a wide
range (length proportional to number of users rating a movie, which varies
widely), the pairwise distances also have a large variance, implying a
low intrinsic dimensionality.  We measured it as $\rho=0.263$.

The curves plotted in Figure~\ref{fig:movies} show similar behaviour to that of 
the random distributions in Figures~\ref{fig:uniform}
and~\ref{fig:normal}.  Approximation factor improves rapidly with more
projections and points examined, in the same pattern, but to a greater
degree, as in the 10-coordinate vector databases, which have higher
intrinsic dimensionality.  However, here there is no noticeable penalty for
using the query-independent algorithm.  The data set appears to be dominated
(insofar as furthest neighbours are concerned) by a few extreme outliers:
movies rated very differently from any others.  For almost any query, it is
likely that one of these will be at least a good approximation of the true
furthest neighbour; so the algorithm that identifies a set of outliers in
advance and then chooses among them gives essentially the same results as
the more expensive query-dependant algorithm.


\section{Conclusion}

We have proposed a data structure for solving the $(c)$-AFN problem.
The data structure retrieves candidate points based on their rankings along random projections.
To do so efficiently it employs a priority queue that is populated at query time.

We give theoretical guarantees on the space and time requirements, as well as experimental confirmation of these.
Further we give a space lower bound on any data structure that works to return the $(c)$-AFN by iterating a fixed list.
This bound supports the suspicions raised by Goel et. al\cite{Goel2001} that query time polynomial in $d$ cannot be
achieved for $c<\sqrt{2}$.
We also suggest a simplified algorithm that can be viewed as an approximation of the convex hull.
While harder to analyse, it is faster and gives very satisfactory experimental results.

Our data structure extends naturally to general metric spaces.  Instead of
computing projections with dot products, which requires a vector space, we
could choose some random pivots and order the points by distance to each
pivot. The query operation would be essentially unchanged.  Analysis and
testing of this extension is a subject for future work.


\chapter{Annulus Query}
\label{sec:annulus-query}

The annulus query problem from Section~\ref{sec:annulus} can be viewed as a problem of finding nearest and furthest neighbors simultaneously.
An obvious path to follow is to combine techniques for these problems into a single data structure.

\section{Introduction}
Similarity search is concerned with locating elements from a set $S$ that are close to a given query $q$.
The query can be thought of as describing criteria we would like returned items to satisfy.
For example, if a customer has expressed interest in a product $q$, we may want to recommend similar products.
However, we might not want to recommend products that are \emph{too} similar.
Thinking of e.g. a book recommendation, we do not want to recommend e.g. just an older translation of the same work.
We claim that a solution to the $(c,r,w)$- approximate annulus query problem (Definition~\ref{def:aaq}) can be found by suitably combining Locality Sensitive Hashing techniques(LSH, See Section~\ref{sec:lsh}), with the approximation technique for furthest neighbor presented in Chapter~\ref{sec:furthest-neighbor}.
In this short chapter we show such a solution in $(\mathbb{R}^d,\ell_2)$ with constant failure probability.

\subsection{Notation}
Consider an LSH function family $\mathcal{H}=\{\mathbb{R}^d\rightarrow U\}$. We say that $\mathcal{H}$ is $(r_1,r_2,p_1,p_2)$-sensitive for $(\mathbb{R}^d,\ell _2)$ if:
\begin{enumerate}
\item$\Pr_{\mathcal{H}}[h(q)=h(p)]\geq p_1 \text{ when } \|p-q\|_2\leq r_1$
\item$\Pr_{\mathcal{H}}[h(q)=h(p)]\leq p_2 \text{ when } \|p-q\|_2> r_2$
\end{enumerate}
We will be using $A(q,r,w)$ for the annulus between two balls, that is $A(q,r,w) = \B{q}{rw}\setminus\B{q}{r/w}$. 

\section{Upper bound}

\begin{theorem}
  \label{thm:aaq-structure}
  Consider a $(wr,wcr,p_1,p_2)$-sensitive hash family $\mathcal{H}$ for $(\mathbb{R}^d,\ell_2)$ and let $\rho = \frac{\log 1/p_1}{\log 1/p_2}$.
  For any set $S\subseteq\mathbb{R}^d$ of at most $n$ points there exists a data structure for $(c,r,w)$-AAQ such that:
  \begin{itemize}
    \item Queries can be answered in time $\BO{dn^{\rho+1/c^2}\log^{(1-1/c^2)/2}{n}}$.
    \item The data structure takes space $\BO{n^{2(\rho+1/c^2)}\log^{1-1/c^2}{n}}$ in addition to storing $S$.
  \end{itemize}

The failure probability is constant and can be reduced to any $\delta>0$ by increasing the space and time cost by a constant factor.
\end{theorem}

We will now give a description of such a data structure and then prove that it has the properties stated in Theorem
\ref{thm:aaq-structure}.

\label{AAQ:datastructure}
Let $k,\ell$ and $L$ be integer parameters to be chosen later.
We construct a function family ${\mathcal{G}: \mathbb{R}^d \rightarrow U^k}$ by concatenating $k$ members of $\mathcal{H}$. Choose $L$ functions $g_1,\ldots,g_L$ from $\mathcal{G}$ and pick $\ell$ random vectors $a_1,\ldots,a_\ell\in\mathbb{R}^d$ with entries sampled independently from $\mathcal{N}(0,1)$.

\paragraph{Preprocessing}
During  preprocessing, all points $x\in S$ are hashed with each of the functions $g_1,\ldots,g_L$.
We say that a point $x$ is in a bucket $B_{j,i}$ if $g_j(x)=i$.
For every point $x\in S$ the $\ell$ dot product values $a_i\cdot x$ are calculated.
These values are stored in the bucket along with a reference to $x$.
Each bucket consists of $\ell$ linked lists, list $i$ containing the entries sorted on $a_i\cdot x$, decreasing from the head of the list.
See Figure~\ref{fig:bucket-contents} for an illustration where $p_{i,j}$ is the tuple $(a_i\cdot x_j,\text{ref}(x_j))$.
A bucket provides constant time access to the head of each list. Only non-empty buckets are stored.

\begin{figure}
\caption{Illustration of a bucket for $\{x_1,x_2,x_3,x_5\}\subset S$. $\ell=3$.}
\label{fig:bucket-contents}
\centering
\vspace{2 mm}
\includegraphics{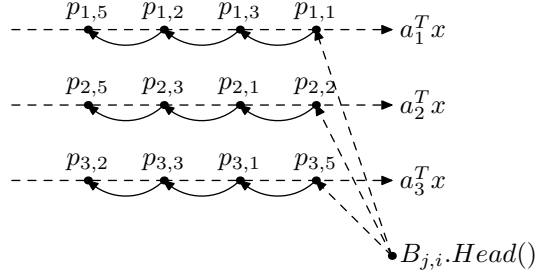}
\end{figure}

\paragraph{Querying}
For a given query point $q$ the query procedure can be viewed as building the set $S_q$ of points from $S$ within $B(q,rcw)$ with the largest $a_{i\in[\ell]}\cdot (p-q)$ values and computing the distances between $q$ and the points in $S_q$.
At query time $q$ is hashed using $g_1,..,g_L$ in $\BOx{dL}$.
From each bucket $B_{j,g_j(q)}$ the top pointer is selected from each list.
The selected points are then added to a priority queue with priority $a_i\cdot (p-q)$.
This is done in $\mathcal{O}(L\ell)$ time.
Now we begin a cycle of adding and removing elements from the priority queue. 
The largest priority element is dequeued and the predecessor link is followed and the returned pointer added to the queue.
If the pointer just visited was the last in its list, nothing is added to the queue. 
If the priority queue becomes empty the algorithm fails. 
Since $r$ is known at query time in the $(c,r,w)$-AAQ it is possible to terminate the query procedure as soon as some point within the annulus is found. 
Note that this differs from the general furthest neighbor problem. 
For the analysis however we will consider the worst case where only the last element in $S_q$ lies in the annulus and bound $|S_q|$ to achieve constant success probability.
\\
We now return to theorem \ref{thm:aaq-structure}
\begin{proof}
\label{prf.1}
  Fix a query point $q$.
  By the problem definition, we may assume $|S \cap A(q,r,w)|\geq1$.
  Define $S_q\subseteq S$ to be the set of candidate points for which the data structure described in section \ref{AAQ:datastructure} calculates the  distance to $q$ when queried. The correctness of the algorithm follows if $|S_q \cap A(q,r,cw)|\geq1$. 

  To simplify the notation let $P_{\text{near}} = S\cap B(q,r/(cw))$ and $P_{\text{far}}=S-B(q,r/w)$.
  The points in these two sets have useful properties. 
  Let $t$ be the solution to the equality:
  \begin{equation*}
    \frac{1}{\sqrt{2\pi}}\frac{e^{\frac{-t^2}{2}}}{t}=\frac{1}{n}
  \end{equation*}

  If we set $\Delta=\frac{rt}{cw}$, we can use the ideas from Lemma \ref{lem:prob} to conclude that:

  \begin{equation}
    \label{pb:near}
    \Pr[a_i(p-q)\geq\Delta]\leq\frac{1}{n}\text{, for }p\in P_{\text{near}}
  \end{equation}
  Also, for $p\in P_{\text{far}}$ the lower bound gives:
  \begin{equation*}
    \Pr[a_i(p-q)\geq\Delta]\geq\frac{1}{(2\pi)^{(1-1/c^2)/2}}n^{-1/c^2}t^{(1-1/c^2)}\left(1-\frac{c^2}{t^2}\right)
  \end{equation*}
  By definition, $t\in\BOx{\sqrt{\log{n}}}$, so for some function $\phi\in\BOx{n^{1/c^2}\log^{(1-1/c^2)/2}{n}}$ we get:
  \begin{equation}
    \label{pb:far}
    \Pr[a_i(p-q)\geq\Delta]\geq\frac{1}{\phi}\text{, for }p\in P_{\text{far}}.
  \end{equation}
  Now for large $n$,
  let $P$ be the set of points that projected above $\Delta$ on at least one projection vector and hashed to the same bucket as $q$ for at least one hash function.
  \begin{equation*}
    P=\{x\in S|\exists j,i:g_j(x)=g_j(q) \text{ and } a_i\cdot(x-q)\geq\Delta\}
  \end{equation*}  
  Let $\ell = 2\phi,m=1+e^2\ell$ and $L=\lceil n^\rho/p_1\rceil$.
  Using the probability bound (\ref{pb:near}) we see that $\text{E}[|P\cap P_{\text{near}}|]\leq\frac{1}{n}n\ell=\ell$.
  So $\Pr[|P\cap P_{\text{near}}|\geq m] < 1/e^2$ by Markov's inequality. 
  By a result of Har-Peled, Indyk, and Motwani~\cite[Theorem
  3.4]{Har-Peled2012}, the total number of points from $S\setminus B(q,rcw)$ across all $g_i(q)$ buckets is at most $3L$ with probability at least $2/3$. So $\Pr[|P\setminus \B{q}{rcw}>3L] < 1/3$.
  This bounds the number of too far and too near points expected in $P$.
  \[\Pr[|P\setminus A(q,r,cw)|\geq m +3L]\leq 1/3+e^{-2}\]
  By applying~\cite[Theorem
  3.4]{Har-Peled2012} again, we get that
  for each $x \in A(q,r,w)$ there exists $i\in[L]$ such that $g_i(x) = g_i(q)$ with probability at least $1-1/e$.
  Conditioning on the existence of this hash function, the probability of a point projecting above $\Delta$ is at least $ 1-(1-1/\phi)^{2\phi}\geq 1-\frac{1}{e^2}$.
  Then it follows that $\Pr[|P\cap A(q,r,w)|< 1]< 1/e+1/e^2$.
  The points in $P$ will necessarily be added to $S_q$ before all other points in the buckets; 
  then, if we allow for $|S_q|=m+3L$, we get
  \[\Pr[|S_q\cap A(q,r,cw)|\geq 1]\geq1-(1/3+1/e+2/e^2)>0.02.\]

  The data structure requires us to store the top $\BO{mL}$ points per projection vector,
  per bucket,
  for a total space cost of  $\BOx{m\ell L^2}$, in addition to storing the dataset, $\BOx{nd}$.
  The query time is $\BOx{(dL+\ell L)+m(d+\log\ell L)}$. 
  The first term is for initializing the priority queue, and the second for constructing $S_q$ and calculating distances.
  Let $\lambda=(1-1/c^2)/2$.
  Since $L=\mathcal{O}(n^{\rho})$ and $\ell,m=\mathcal{O}(n^{1/c^2}\log^\lambda n)$ we get query time:
\begin{equation}
  \BO{dn^\rho+n^{\rho+1/c^2}\log^{\lambda}{n}+n^{1/c^2}
    \log^{\lambda}{n}\left(d+\log{(n^{\rho+1/c^2}
      \log^{\lambda}{n})}\right)}
\end{equation}
Depending on the parameters different terms might dominate the cost.
But they can all be bounded by $\BOx{dn^{\rho+1/c^2}\log^{(1-1/c^2)/2}{n}}$ as stated in the theorem.
The hash buckets take space:

\begin{equation}
\BO{n^{2(\rho+1/c^2)}\log^{1-1/c^2}{n}}.
\end{equation}

Depending on $c$, we might want to bound the space by $\BOx{n\ell L}$ instead, which  yields a bound of $\BOx{n^{1+\rho+1/c^2}\log^{(1-1/c^2)/2}{n}}$.
\end{proof}

\section{Conclusion}

In this short chapter we showed a data structure for the $(c,r,w)$-approximate annulus query problem.
We showed that the query time is sublinear in the size of $S$ and linear in $d$.
This makes the data structure well suited for the high-dimensional, high-volume paradigm, although the storage requirements can be quite large when $c$ is close to $1$.
Later results have shown that similar bounds can be achieved through the combination of LSH with ``anti''-LSH functions~\cite{AumullerCP017}.
It is easy to employ the query-independent variation of the furthest neighbor data structure instead of the query dependent variation.
This would significantly reduce the space usage from $\BOx{m\ell L^2}$ to just $\BOx{m}$.
It would also reduce the query time, although that is not dominated by the priority queue insertions that would be saved.
Given our experimental results in Chapter~\ref{sec:furthest-neighbor} this alternative it seems to offer an attractive, practical solution to the approximate annulus query problem, although more difficult to analyse theoretically.

\chapter{Distance Sensitive Approximate Membership}
\label{sec:dist-sens-appr}

  The Bloom filter~\cite{Bloom1970} is a well-known data structure for answering \emph{approximate membership queries} on a set $S$, i.e., queries of the form ``Is $x$ in $S$?''.
  By allowing some false positive answers (saying `yes' when the answer is in fact `no') Bloom filters use space significantly below what is required for storing $S$.
  In the \emph{distance sensitive} setting we work with a set $S$ of (Hamming) vectors and seek a data structure that offers a similar trade-off, but answers queries of the form ``Is $x$ \emph{close} to an element of $S$?'' (in Hamming distance).
  Previous work on distance sensitive Bloom filters have accepted false positive \emph{and} false negative answers.
  Absence of false negatives is of critical importance in many applications of Bloom filters, so it is natural to ask if this can be also achieved in the distance sensitive setting.
  Our main contributions are upper and lower bounds (that are tight in several cases) for space usage in the distance sensitive setting where false negatives are not allowed.

\section{Introduction}

In this Chapter we present upper and lower bounds on the space complexity of filters for \emph{distance sensitive approximate membership queries}($(r,c,\epsilon)$-DAMQ, Definition~\ref{def:damq}) in $(\{0,1\}^d,H)$.
These filters answer queries of the form ``Is $x$ similar to some element of $S$?'' 
Where ``similar'' means within a given Hamming distance $r$.
We study distance sensitive filters under an approximation factor $c\geq1$: a small false positive rate $\epsilon $ is allowed when $S$ has points at distance more than $cr$ from the query point. However, false negatives are never allowed.
This is in contrast to previous work on this problem~\cite{Kirsch}.
To our best knowledge, ours is the first solution with no false negatives.

\subsection{Motivation and practicality}

Bloom filters are widely used in practice.
One reason is because they require less space than a dictionary data structure for storing $S$.
We argue that the lack of \emph{false negatives} is also of critical importance to their frequent use in practice.

Generally the set $S$ is a subset from some much larger domain.
If queries are roughly uniformly selected from the domain, answers to a membership query should most often be negative.
For this majority of queries the Bloom filter always gives the correct, negative, answer.
Since the filter then rarely gives a positive, possibly wrong, answer, these queries could all be double-checked using an exact, but less space-efficient, less accessible method (perhaps on a different machine).
This allows us to use Bloom filters as a first component in an exact two-level data structure.
Here it acts as an initial filter, reducing the use of a second, slower to access but exact data structure.   
Having false negatives means this two-level structure would fail to be exact.
We would have to choose one of the levels:
Either accept some possibility of getting a wrong answer or perform an expensive exact query every time.
We are motivated by providing a data structure for distance sensitive membership query that \emph{can} be used in this way, i.e. that does not have false negatives.

There are many potential applications for this kind of data structure.
As a concrete example, consider a journal comprising a large collection of academic papers.
When accepting a new paper the journal might want to check if the new paper is very similar to any prior work already published.
By using a distance-sensitive  filter this can be done in a space-efficient manner.
Because we do not allow false negatives, any new paper passing this test (with a `no' result) is guaranteed to be significantly different from all prior work.
In the rare case that a paper fails the test, the submission process could be halted pending a consultation of the full archive.
Furthermore, since the filter provides very little information about the content of the papers it would not need to be subject to the same access control as a full database of all the journals papers might be under.
More interesting examples of applications for distance-sensitive filters can be found in~\cite{Kirsch} and for Bloom filters in general in~\cite{broder2004network}.

\subsection{Our results}

We study the space required for answering distance-sensitive approximate membership queries with no false negatives.
It turns out that, in contrast to approximate membership, we get different bounds depending on how the false positive rate is defined:
\begin{itemize}
\item If we desire a \emph{point-wise} error bound (Definition~\ref{def:DSAM}) for each query at distance $\geq cr$ from $S$, the space usage must be $\BOM{n \left( \frac{r^2}{d} + \log \frac{1}{\varepsilon}\right)}$ for almost all parameters, and $\BOM{n \left(\frac{r}{c}+\frac{c}{c-1} \log \frac{1}{\varepsilon}\right)}$ bits if $n$ is not too large
(see Theorem~\ref{thm:wclb1}).
\item If it suffices to have an $\varepsilon$ \emph{average} false positive rate (Definition~\ref{def:EDSAM}) over all queries at distance $\geq cr$ from $S$, where $Cl<d/2$, the space usage must be $\BOM{n \left( \frac{r^2}{d} + \log \frac{1}{\varepsilon}\right)}$ bits.
(see Theorem~\ref{avg_error_thm}).
\end{itemize}
We match these lower bounds with almost tight upper bounds on space usage in Section.~\ref{sec:upper-bounds}.
We introduce the notion of vector \emph{signature}, which can be seen as a succinct version of a {\textsc CountSketch}~\cite{CharikarCF04}, and then show how to use them to design distance sensitive filters with point-wise and average errors.

Our focus is on space usage rather than query-time, and indeed it would be surprising if poly-logarithmic query time in $n$ is possible since our (point-wise) filter could be used, say with $\varepsilon = 1/n$, to solve the $c$-approximate nearest neighbor problem. The best currently know data structures for this problem use $n^{\Omega(1/c)}$ time~\cite{andoni2015optimal}.

\subsection{Related work}

There is little prior work specifically on distance sensitive approximate membership.
The problem corresponds to querying a standard Bloom filter in a ball around the query point, but this solution is slow, time $\Omega(\binom{d}{r})$, and also not particularly space efficient since we would need to use a Bloom filter with a very small false positive rate to bound the probability that none of the queries yield a false positive.
More precisely, the required space usage for this approach would be $\Omega(n r \log\frac{d}{r})$ bits~\cite{Carter1978}.

Mitzenmacher and Kirsch~\cite{Kirsch} considered data structures that look like Bloom filters but replace standard hash functions with locality sensitive hash (LSH) functions~\cite{Indyk1998} to achieve distance sensitivity.
However, this approach introduces false negatives because LSH is not guaranteed to produce collisions.
In order to reduce the number of false negatives the conjunction used when querying Bloom filters is replaced by a threshold function: there should just be ``many'' hash collisions.
Unfortunately, the achieved approximation factor is large, i.e. $c=\BO{\log n}$.
Hua et al.~\cite{Hua2012} extended the data structure of~\cite{Kirsch} with practical improvements and provided extensive experiments, confirming that false negatives also appear in practice.

There has been some recent progress on developing LSH families that can answer near neighbor queries without false negatives~\cite{Pagh2016}, but it seems inherent to such families that the storage cost grows exponentially with~$r$.
Thus this approach is not promising, perhaps except for very small values of $r$.

Finally, it is known that allowing a constant fraction of false negatives does not asymptotically improve the space usage that can be achieved by approximate membership data structures~\cite{pagh2001lossy}.
It is not apriori clear that space usage will be worse when false negatives are not allowed.


\section{Problem definition and notation}
The Hamming distance $H(p,q)$ between two points $p,q\in\{0,1\}^d$ is the number of positions where $p$ and $q$ differ. 
Given a set $S\subseteq \{0,1\}^d$ of $n$ points and a point $q\in \{0,1\}^d$, 
we extend the meaning of $H(\cdot)$ by defining $H(q,S)$ to be the minimum distance between $q$ and any point in $S$, i.e. $H(q,S)=\min_{p\in S}H(q,p)$.
We use $\binom{A}{n}$ to denote $\{S\subseteq A:|S|=n\}$ when $A$ is a set.
We will be using the $\B[d]{x}{r}$ notation as defined in section \ref{sec:ball-notation}.

We formally define \emph{distance sensitive approximate membership filters}  as follows:

\begin{definition}
  \label{def:DSAMmother}
  {\textsc (Distance sensitive  approximate membership filter)}
Let $r\geq 0$, $c\geq 1$, and $\varepsilon\in[0,1]$. 
Given a set $S \subset \{0,1\}^d$ define the two sets:
\begin{align*}
  Q_\text{near}&=\{x\in\{0,1\}^d:H(x,S)\leq r\},\\
  Q_\text{far}&=\{x\in\{0,1\}^d:H(x,S)> cr\}.
\end{align*}

A $(r,c,\varepsilon)$-distance sensitive approximate membership filter for $S$
is a data-structure that on a query $q\in\{0,1\}^d$ reports: 
\begin{itemize}
\item `\emph{Yes}' if $q\in Q_\text{near}$
\item `\emph{No}' if $q\in Q_\text{far}$, but with some probability of error (i.e. false positives).
\end{itemize}
If $q\notin Q_\text{near}\cup Q_\text{far}$ the data structure can return any answer.
\end{definition}

\begin{figure}[ht]
  \centering
\includegraphics{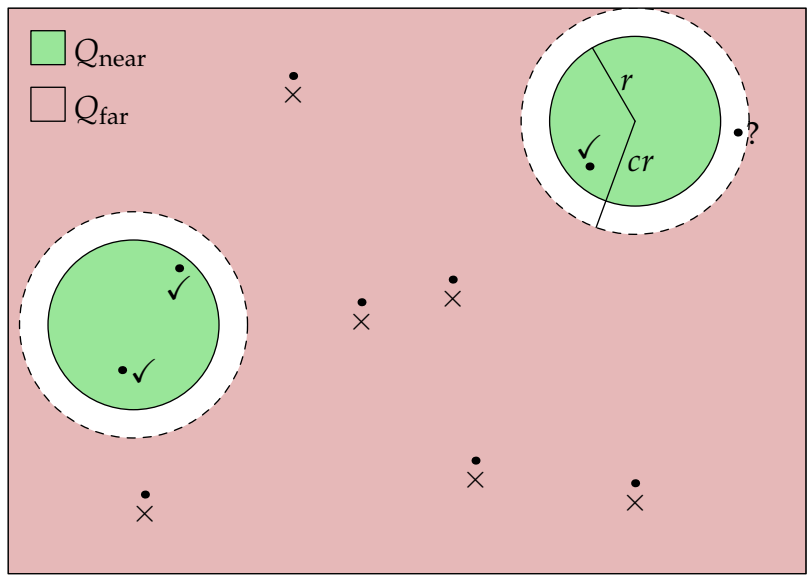}
  \caption[DAMQ data-structure]{Illustration for $n=2$ showing some queries with their desired output: $\checkmark \rightarrow$ `\emph{Yes}', $\times\rightarrow$`\emph{No}', ? $\rightarrow$ Undefined.}
  \label{fig:damq}
\end{figure}

In the rest of the chapter, we study space bounds under two error measures, named point-wise and average errors.

\begin{definition}[Point-wise error]
  \label{def:DSAM}
A $(r,c,\varepsilon)$-distance sensitive approximate membership filter for $S$ has point-wise error $\varepsilon $ if, on a query $q\in\{0,1\}^d$, it reports: 
\begin{itemize}
\item `\emph{Yes}' if $q\in Q_\text{near}$;
\item `\emph{No}' with probability at least $1-\varepsilon$ if $q\in Q_\text{far}$ (the probability is over the random choices of the filter).
\end{itemize}
\end{definition}

This is a strong guarantee since each  point in $Q_\text{far}$ has probability $\varepsilon$ to fail.
If hard queries are not expected, it might be acceptable that some points give false positives in every instance of the data structure, as long as only an $\varepsilon$ total fraction of points in $Q_\text{far}$ give false positives.
We refer to this weaker filter as the \emph{average error} version:

\begin{definition}[Average error]
\label{def:EDSAM}
A $(r,c,\varepsilon)$-distance sensitive approximate membership filter for $S$ has average error $\varepsilon $  if, on a query $q\in\{0,1\}^d$, it reports: 
\begin{itemize}
\item `\emph{Yes}' if $q\in Q_\text{near}$;
\item `\emph{No}' with probability at least $1-\varepsilon$, if $q$ is randomly and uniformly selected from $Q_\text{far}$ (the probability is over the random selection in $Q_\text{far}$ and over the random choices of the filter).
\end{itemize}
\end{definition}

The average-error guarantee implies that the filter 
provides the correct answer to at least a $(1-\varepsilon)$ fraction, in expectation,  of the points in $Q_\text{far}$.
Clearly, a filter with point-wise error is also a filter with average error.
Though the difference between these two error measures may seem small, their properties and analysis differ substantially.


\section{Lower bounds}
\label{sec:lower-bounds}
As a warm-up, we first investigate what can be done when no errors are allowed, that is when $\varepsilon=0$ (in this case the average and point-wise error guarantees are equivalent).
The next theorem shows that, up to constant factors, the optimal filter is no better than one that stores $S$ explicitly.
When $\varepsilon = 0$ there is no distinction between point-wise and average error. 
Throughout this chapter we let $\log x$  denote the logarithm of $x$ in base 2.

\begin{theorem}\label{eps0}
Any distance sensitive approximate membership filter with error $\varepsilon = 0$ must use at least 
$$n \log \left(\frac{2^d}{e n \B[d]{}{cr}}\right)$$
 bits in the worst case. If $d = \omega(\log n)$ and $cr = \SO{d/\log d}$ then it must use $\BOM{nd}$ bits. 
\end{theorem}
\begin{proof}
  The proof is an encoding argument. A set $S \subseteq \{0,1\}^d$ of size $n$ is encoded by Alice and sent to Bob who will recover it.
  Assume the optimal filter  uses $s$ bits in the worst case.
  Alice inserts the given set $S$ into the optimal filter, and runs the query algorithm on each point in the universe.
  Since there are no false positives, the filter says `yes' to a set $P$ of at most $n \B[d]{}{cr}$ points.
  Alice encodes $S$ as a subset of $P$ using $\log  \binom{n  \B[d]{}{cr}}{n}+\BO{1}$ bits.
  Alice then sends the at most $s$ bits of the optimal filter to  Bob along with the strings encoding $S$ as a subset of $P$.

  The decoding procedure is straightforward.
  Bob queries the optimal filter with all points in $\{0,1\}^d$, recovering $P$.
  Then, using $P$ and the second string of bits received from Alice, Bob can recover the initial set $S$.
  
  Since every set $S$ of size $n$ can be encoded, we get that:
\begin{eqnarray*}
 s +\log  \binom{n \B[d]{}{cr}}{n} &\geq& \log  \binom{2^{d}}{n}  
\end{eqnarray*}  
from which follows that
\begin{eqnarray*}
  s & \geq & \log  \left( \left(\frac{2^{d}}{n}\right)^{n} / \left( \frac{en \B[d]{}{cr}}{n}\right)^{n} \right)  \\
   & \geq & nd - n \log  (en) - n \log  \B[d]{}{cr}
\end{eqnarray*}

If $d = \omega(\log n)$, we get $s = \BOM{nd - n \log \B[d]{}{cr}}$.
Further, using that $\B[d]{}{cr} = \sum_{i=0}^{cr} \binom{d}{i} < d^{cr}$ for $cr <d/2$, we get that $s = \BOM{nd - ncr \log d}$, which is $\BOM{nd}$ when $cr = \SO{d/\log d}$.
\end{proof}

\subsection{Average error}

Next we investigate the distance sensitive membership problem with average error $\varepsilon>0$.

\begin{theorem}\label{avg_error_thm}
  Assume that $ n\B[d]{}{cr}/2^d < \varepsilon < 1/4 $. Then any distance sensitive membership filter with average error $\varepsilon$ must use 
  $$\BOM{ n \left(\frac{r^2}{d} + \log \left(\frac{1}{\varepsilon}\right)\right)}$$
   bits in the worst case.
\end{theorem}

Before proving the theorem, we highlight some remarks:
\begin{enumerate}
\item The above theorem holds as long as $n\B[d]{}{cr}< 2^{d-2}$, i.e. the ``membership set'' covers less than a quarter of the full Hamming space.
  This is the most interesting range of parameters. Similarly to Bloom filters, our approach is not optimal when non-members are rare. 
  As we will see later, the $\Omega(nr^{2}/d)$ lower bound holds as long as $n\B[d]{}{cr}<2^{d-1}$, and it starts to deteriorate when $n\B[d]{}{cr}$ approaches $2^d$.
    It is clear that some upper bound on $n\B[d]{}{cr}$ is necessary; if it approaches size $2^{d}-O(n/d)$, then storing the complement exactly in $O(n)$ bits suffices.
    Also note that at the lower limit of $\varepsilon= n\B[d]{}{cr}/2^d$, this lower bound matches the lower bound of the $\varepsilon=0$ case in Theorem~\ref{eps0}. Thus Theorem~\ref{eps0} follows from Theorem~\ref{avg_error_thm}.

\item  The term $\B[d]{}{cr}$ has no simple closed expression for all $c$ and $r$, and so the dependence of the hypothesis of the theorem on $c$, $r$ and $d$ is not straightforward.

\end{enumerate}
The rest of this section is devoted to the proof of Theorem~\ref{avg_error_thm}.

\begin{proof}
The proof is derived for a deterministic version of the distance sensitive membership filter: in this setting, the filter answers `no' to at least a fraction of points in $Q_\text{far}$ (i.e., points at distance at least $cr$ from all points in the input point set $S$), and hence there can be at most $\varepsilon |Q_{\text{far}}|$ false positives. 
We claim that such a lower bound applies also to a randomized filter. 
Suppose that a randomized filter  requires  $s$ bits, with $s$ smaller than the lower bound.
Since the expected number of correct `no' answers is at least $(1-\varepsilon)|Q_\text{far}|$, there must exist random values for which the filter provides the correct solution for at least $(1-\varepsilon) | Q_{\text{far}}|$ points: by using these values, we obtain a deterministic average error filter with space complexity $s$ lower than the lower bound, which is a contradiction.

We first prove a $\BOM{n \log(1/\varepsilon)}$ lower bound. The proof is an encoding argument that extends the scheme presented in the proof of Theorem~\ref{eps0} and in \cite{Carter1978}.
Alice receives a set $S$ of size $n$ from the universe to encode. Assume the optimal distance sensitive filter with $\varepsilon$ average error uses $s$ bits in the worst case.
Alice inserts~$S$ into the filter, and runs the query algorithm on all points in the universe recovering $P$, the set of points the filter answers `Yes' to.
We first claim that $|P|\leq2^{d+1}\varepsilon$.
First, the number of positives not considered false is at most $n\B[d]{}{cr}$ (this bound is achieved when all the balls are disjoint), which is less than $2^{d}\varepsilon$. Also the number of false positives is always at most $2^d \varepsilon$.
Adding these, we find that the total number of positives is at most $2^{d+1}\varepsilon$.
Alice then encodes the set $S$ as a subset of $P$, using at most $\log \binom{2^{d+1}\varepsilon}{n}$ bits. Alice sends these bits to Bob along with the at most $s$ bits representing the optimal filter for $S$.

Bob queries the filter with all $q\in\{0,1\}^d$ and recovers $P$. Bob then uses the extra bits sent by Alice to find the subset of $P$ identical to $S$.
We have that:
\begin{eqnarray*}
s + \log \binom{2^{d+1}\varepsilon}{n} &\geq& \log \binom{2^d}{n}  \\  
\Rightarrow s &\geq& \log \frac{2^d \cdots (2^d-n+1)}{(\varepsilon 2^{d+1})\cdots(\varepsilon 2^{d+1}-n+1)}  \\
\Rightarrow s &\geq&   \log \left(\frac{2^d }{ \varepsilon 2^{d+1}}\right)^n  \\
\Rightarrow s &\geq&  n\log\left(\frac{1}{2\varepsilon}\right) \in \BOM{n \log\left(\frac{1}{\varepsilon}\right)}. 
\end{eqnarray*}

To prove the $n r^{2}/d$ lower bound, we first introduce some notation.
Consider the hypercube graph on the $d$-dimensional Hamming cube where two points $p$ and~$q$ have an edge between them if they have Hamming distance $1$. Given a set $A \subset \{0,1\}^d$, let $A^c$ denote its complement, and define $\partial A$ to be the set of points in $A$ that have an edge to a point in $A^c$ (when either $A^c$ or $A$ is empty, $\partial A$ is the empty set). Also, given an integer $r> 0$, define $A^{-r} = A \setminus \bigcup_{x \in \partial A} \B[d]{x}{r-1}$. $A^{-r}$  contains exactly those points $x \in A$ such that the ball $\B[d]{x}{r}$ is contained inside $A$.

A deterministic filter that uses $s$ bits can be viewed as a function $\mathcal{F}: \binom{\{0,1\}^d}{n} \rightarrow \{0,1\}^s$; given a set $S \subseteq \{0,1\}^d$ of size $n$, $\mathcal{F}(S)$ is the memory representation of $S$ that uses at most $s$ bits. Let $V(S) = |\cup_{x \in S} \B[d]{x}{r}| + \varepsilon( 2^d - |\cup_{x \in S} \B[d]{x}{r}|)$: we note that $V(S)$ is an upper bound to the number of `yes' answers returned by the filter (i.e., both true and false positives), and  $V(S) \leq 2^{d-1}$ by the hypothesis of the theorem.

Running the query algorithm on all points in the Hamming cube for the representation $\mathcal{F}(S)$ returns a set~$P_{S}$ of positives (${P_{S}}^c$ of negatives) such that $|P_{S}| \leq V(S)$. Let us denote by $D$ the function that takes in a set $S$, and outputs the set $P_{S}$ of positives returned by the query algorithm on the representation $\mathcal{F}(S)$. 

Varying over all $S \in  \binom{\{0,1\}^d}{n}$, we get a family $\mathcal{T}$ of sets such that:
\begin{enumerate}
\item $\forall S$, $\exists P \in \mathcal{T}$ such that $\B[d]{x}{r} \subset P$ for all $x \in S$. 
\item For any $P \in \mathcal{T}$ and $\forall S$ such that $D(S) = P$, $|P| \leq V(S)$.
\end{enumerate}

Thus $D$ is a function from $\{0,1\}^s$ to $\mathcal{T}$, the image of which is all of $\mathcal{T}$. This implies that $s\geq \log |\mathcal{T}|$.
So in order to get a lower bound on $s$ it suffices to get a lower bound on the size of the smallest family $\mathcal{T}$ with the above properties.

Fix $P \in \mathcal{T}$. Define $D^{-1}(P)=\{S: D(S) = P\}$. Any ball of radius $r$ around a point $p \in S$ such that $S \in D^{-1}(P)$ must be completely contained inside $P$. The maximum number of such points $p$ is $|P^{-r}|$. Thus we get that $|\cup_{S \in D^{-1}(P)} S| \leq |P^{-r}|$. This implies that $|D^{-1}(P)| \leq \binom{|P^{-r}|}{n}$.

Since all possible sets (from $\binom{\{0,1\}^d}{n}$) need to be covered, we get that $|\mathcal{T}| \geq  \binom{2^d}{n} / \binom{|P^{-r}|}{n}$. We now need an upper bound on the size of $|P^{-r}|$.
Lemma~\ref{lem:ball} states that $|P^{-r}| \leq 2^d e^{-2r^2/d}$.


The proof of the lower bound in Theorem~\ref{avg_error_thm} then follows by applying Lemma~\ref{lem:ball}:
\begin{align*}
 |\mathcal{T}|  &\geq  \binom{2^d}{n} / \binom{|P^{-r}|}{n} \\
 & \geq  \left( \frac{e 2^d}{|P^{-r}|} \right)^{n} \\ 
 & \geq e^{n\left(2r^2/d+1\right)}  , 
\end{align*}
  which implies that
$ s \geq \log \mathcal{T} = \Omega \left(nr^{2}/d  \right)$.
Combining our bounds, we get that when $n,r$ and $c$ satisfy the condition that $nB(cr,d) \leq 2^{d-2}$, any filter must use $\Omega ( n (r^2/d + \log (1/\varepsilon)))$ bits in the worst case. 
\end{proof}

\begin{lemma}\label{lem:ball}
Let $S$, $P$ and $r$ be as above. Then $|P^{-r}| < 2^d e^{-2r^2/d}$. 
\end{lemma}

\begin{proof}
Note that $P$ is the set of positives (after running the query algorithm on all points in the Hamming space) on the filter $\mathcal{F}(S)$. Thus we have that $|P| \leq V(S) \leq 2^{d-1}$. The size of $P^{-r}$ increases as $P$ increases, so we have that $|P^{-r}|$ is at most $\max |A^{-r}|$, where the maximum is taken over all sets $A$ such that $|A| = 2^{d-1}$.

We will first prove that if $|A| = 2^{d-1}$, then $\max |A^{-r}|$ is at most $B(d/2-r,d)$ (the size of the Hamming ball of radius $d/2 - r$). The proof is by induction (the statement is actually true for any $r < d/2$, not just the input parameter $r$, and so we will treat it as a variable). 

For $r=1$, the statement is that $|A^{-1}|$ is maximized when $A$ is the Hamming ball of radius $d/2$. This is the statement of Harper's theorem, also called the vertex-isoperimetric inequality \cite{Bollobas:1986:CSS:7228}, that states that Hamming balls have the smallest vertex boundary among all sets of a given size.

Assume now that the statement is true for $r=k$, i.e., of all sets $A$ such that $|A| = 2^{d-1}$, the one that maximizes $|A^{-k}|$ is the Hamming ball of radius $d/2$. In this case, note that $A^{-k}$ is the Hamming ball of radius $d/2-k$.

Assume that the statement for $r=k+1$ is false, i.e., there is a set $W$ (of size $2^{d-1}$) such that $|\B[d]{0}{d/2}^{-(k+1)}| < |W^{-(k+1)}|$. Note that by the inductive hypothesis, we know that $|\B[d]{0}{d/2}^{-k}| \geq |W^{-k}|$.

However, the vertex-isoperimetric inequality can also be stated as: if a set $W$ (that is not a ball) has size greater then or equal to that of the Hamming ball of radius $R$, then $|W \cup \Gamma(W)|$ is larger than the size of Hamming ball of radius $R+1$, where $\Gamma(W)$ is the set of neighbors of $W$. Thus  $|\B[d]{0}{d/2}^{-(k+1)}| < |W^{-(k+1)}|$ actually implies $|\B[d]{0}{d/2}^{-k}| < |W^{-k}|$, which contradicts the inductive hypothesis.

Finally, we bound $\B[d]{}{d/2-r}$ using  the following Chernoff-Hoeffding bound~\cite{mitzenmacher2005probability} for binomial random variables:

If $X_{i}$ denotes the outcome of the $i$th coin toss with an unbiased coin, and $X = \sum_{i=1}^{d} X_{i}$, then $\Pr[X \leq \mu - a] \leq e^{-2a^2/d}$, for all $0 < a < \mu$, where $\mu = \mathbb{E}[X] = d/2$. 
Let $X \sim \mathcal{B}(d,0.5)$. Now we have that
\begin{align*}
  |P^{-r}| \leq &\B[d]{}{d/2-r} \\
  =& 2^d P[X \leq d/2-r]\\
\leq & 2^d e^{-2r^{2}/d}.
\end{align*}
\end{proof}


\subsection{Point-wise error}
The lower bound for the average case in Theorem~\ref{avg_error_thm} also applies to a filter with point-wise error guarantees.
A $(r,c,\varepsilon)$-filter with point-wise error $\varepsilon $ is also a $(r,c,\varepsilon)$-filter with average error $\varepsilon$:
if each point fails with probability $\varepsilon$, then a random point fails with probability $\varepsilon$.
However, a stronger lower bound holds for point-wise error if the number of points $n$ is not too large. 

\begin{theorem}\label{thm:wclb1}
  Consider an $(r,c,\varepsilon)$-distance sensitive approximate membership filter with point-wise error  on a set $S$ of $n$ points in $\{0,1\}^d$.
  Then, in the worst case, the filter must use:
\begin{itemize}
\item  $\BOM{n \left(\tfrac{r^2}{d} + \log \tfrac{1}{\varepsilon}\right)}$ bits if $n \B[d]{}{cr}/2^d<\varepsilon<1/4$.
\item  $\BOM{n \left(\tfrac{r}{c} + \log\tfrac{1}{\varepsilon} \right)}$ bits  if $n \B[\delta cr]{}{cr}/2^{\delta cr} < \varepsilon<1/4$ for some constant $\delta$.
\end{itemize}
\end{theorem}
\begin{proof}
As already said, the first bound follows by applying Theorem~\ref{avg_error_thm} since a $(r,c,\varepsilon)$-filter with point-wise error  is also a $(r,c,\varepsilon)$-filter with average error.

We now prove the second claim.
Observe that a filter for $d$-dimensional points with point-wise error $\varepsilon $ is also a filter for $d'$-dimensional points with the same guarantees when $d>d'$.
Then, the lower bound obtained by  Theorem~\ref{avg_error_thm} for dimension $d'=\delta cr$, for some small constant $\delta$, applies to dimension $d$, and it is also stronger since the lower bound in  Theorem~\ref{avg_error_thm}  is {decreasing} in $d$.
However, the new bound needs to meet the condition of Theorem~\ref{avg_error_thm}:
given a filter for dimension  $d'=\delta cr$, then the condition states that $n\B[\delta cr]{}{cr}/2^{\delta cr} < \varepsilon < 1/4$. The theorem follows.
\end{proof}

We observe that the proof used to derive the stronger lower bound does not work for the average error measure: indeed, the average error rate relatively to a subspace (e.g., $\{0,1\}^{d'}$) can be much larger than the one in the complete space (i.e., $\{0,1\}^d$).

As we will see in the next section, there exists a filter that almost match the asymptotic lower bound 
if $c\geq 2$. 
However,  if $1<c<2$ and $\varepsilon$ is sufficiently small, the upper bound has a $\BO{1/(c-1)^2}$ overhead:
although the upper bound is not optimal, the next theorem shows that a $1/(c-1)$ overhead is unavoidable when $1<c<2$.
To help in assessing the hypothesis in the theorem, 
we notice that, when $c=1+\frac{1}{\sqrt{r}}$, the theorem holds for $n\leq 2^{\BT{r}}$, $\varepsilon\leq 2^{-\BT{r}}$, $d=2^{\BOM{\sqrt{r}}}$
and it gives a $\BOM{nr^{3/2}}$ bound, whereas the previous theorem only gave $\BOM{nr}$.
We note that the next theorem can be integrated with the previous Theorem~\ref{thm:wclb1} to get an additive $nr/c$ or $nr^2/d$ more (according to the parameters).

\begin{theorem}\label{thm:wclb2}
Let $c\leq 2$, $\varepsilon\leq (c-1)/n$ be such that $ d(c-1)\geq  ((c-1)/\varepsilon)^{6/(r(c-1))} + (r(c-1))^3$.
Consider an $(r,c,\varepsilon)$-distance sensitive approximate membership filter with point-wise error $\varepsilon $ on a $S$ set of $n$ points in $\{0,1\}^d$.
Then, the filter requires  
$$\BOM{\frac{n}{c-1}\log\left(\frac{1}{\varepsilon}\right)}$$ 
bits in the worst case.
\end{theorem}

\begin{proof}
  The main idea of the proof is to use the optimal filter  in a one-way randomized protocol between two players (Alice and Bob) to send an arbitrary element $x$ of a given set $S$ from Alice to Bob who must identify which element he has:
  It is known (See the indexing problem~\cite{KushilevitzN97}) that such a protocol requires $\BOM{\log |S|}$ bits if the protocol succeeds with probability at least $2/3$ and the two players share random bits.
  The proof uses two families of error correcting codes, $\mathcal{C}$ and $\mathcal{M}$, that are explained below.
  Without loss of generality we assume that they are known to both Alice and Bob (the code families can be constructed with a deterministic brute-force algorithm).

Let $k=1/(c-1)$.
The error correcting binary code  $\mathcal C$ has $m=1/(n \varepsilon k)$ codewords, each one with length  $d_{\mathcal{C}}=d/k$ bits, weight $w=r/k$ and minimum Hamming distance between 
two codewords at least $\delta=r/k$.
\cite[Theorem 6]{GrahamS80} shows that such a code exists of size at least
\begin{eqnarray*}
\frac{d_{\mathcal{C}}^{w-\delta/2+1}}{\delta!}& \geq & 
\frac{(d(c-1))^{r(c-1)/2}}{(r(c-1))^{r(c-1)}}\\
& \geq & (d(c-1))^{r(c-1)/6}\\
& \geq & \frac{c-1}{\varepsilon}
\end{eqnarray*}
where in the third inequality we exploit the fact that $d(c-1)\geq (r(c-1))^3$ and in the last step we use $d(c-1)\geq ((c-1)/\varepsilon)^{6/(r(c-1))}$.

The error correcting binary code  $\mathcal M$ has $n$ codewords and minimum Hamming distance $rc$ (there is no requirement on codewords weights); we let $\mathcal M =\{m_1,\ldots, m_{n}\}$.
By the Gilbert-Varshamov~\cite{MacKay02} bound such a code $\mathcal M$ exists with length $d_{\mathcal{M}}= rc+\log n$. 

Alice arbitrary selects  $n$ codes $x_i=(x_{i,1}, \ldots,  x_{i,k-1})$ from the set $\mathcal C^k$.
Then, she encodes each $x_i$ into $\hat x_i= x_{i,1}\cdot \ldots \cdot x_{i,k} \cdot z_0 \cdot m_i $, where $\cdot$ denotes the concatenation of binary sequences, $z_0$ is a sequence of $r/k=r(c-1)$ zeros, and $m_i\in \mathcal M$. 
The length of each $\hat{x}_i$
is $d_x=k d_{\mathcal{C}}+d_{\mathcal{M}}+r/k= d+\log n+r(2c-1)$.
Finally, Alice inserts $\hat  x_0,\ldots, \hat  x_{n-1}$ into the optimal filter and sends the filter to Bob using $S(n,d_X, c, r)$ bits.

We now show  that Bob can reconstruct  each codeword $x_{i}$ by querying the filter at most $1/\varepsilon$ times.  
Codeword $x_{i,1}$ is obtained by performing a query with 
$q=q' \cdot z_2 \cdot z_3 \cdot m_i$ for every possible codeword $q'\in \mathcal{C}$, where $z_2$ is a sequence of $(k-1)\delta = (k-1)r(c-1) $ zeros, $z_3$ is a sequence of $r/k$ ones, and $m_i\in \mathcal M$. The distance between $q'$ and any $\hat x_j$ in the filter is
$D(\hat x_j,q)=D(x_{j,1},q')+D(x_{j,2}\cdot\ldots\cdot  x_{i,k}, z_2)+D(z_0,z_3)+D(m_j,m_i)$.
It holds that: 
\begin{enumerate}
\item  $D(x_{i,1},q')\geq r(c-1)$ if $q\neq x_{i,1}$ and $0$ otherwise; 
\item $D(x_{j,2}\cdot\ldots x_{i,k}, z_1) = (k-1) r(c-1)=r-r(c-1)$ since each codeword in $\mathcal{C}$ has weight $r(c-1)$;
\item  $D(z_0,z_3) = r(c-1)$;
\item  $D(m_{j},m_i)\geq r c$ if $m_j\neq m_i$
 and $0$ otherwise. 
\end{enumerate}
Therefore, $D(\hat x_j,q) = r$ if $x_{i,1}=q'$ and $m_i=m_j$, and $D(\hat x_j,q) \geq rc$ otherwise. 
A similar procedure holds for computing  $x_{i,j}$ for each $i$ and $j$.

Bob performs  $mk$  queries per $x_i$ and $nkm=1/\varepsilon$ queries in total.
The expected number of wrong queries is then $1$
and, if the protocol is repeated independently,  there  is a constant probability that all queries succeed. 
Since Bob is able to reconstruct an entry from the set $\mathcal S=\mathcal{C}^{nk}$, by the aforementioned result in~\cite{KushilevitzN97}, we have 
\begin{eqnarray*}
S(n,d_x, c, r, \varepsilon)&\geq& \BOM{\log \mathcal{S}}\\ 
&\geq & \BOM{ \log |\mathcal{C}|^{nk}}\\
&\geq & \frac{n}{c-1} \log (1/\varepsilon).\end{eqnarray*}
\end{proof}


\section{Upper bounds}
\label{sec:upper-bounds}
In this section we propose distance sensitive approximate membership filters with point-wise and average errors. 
We start in Section~\ref{sec:vect-sign-meth} by introducing the concept of vector signature.
It can be seen  as a succinct version of {\textsc CountSketch}~\cite{CharikarCF04}, where we have thrown away information  not required for answering distance sensitive approximate membership queries.
In Sections~\ref{sec:filter-with-wc} and~\ref{sec:filter-with-average}, we then show how to use vector signatures to derive almost-optimal approximate membership filters with  point-wise and average errors respectively.

\subsection{Vector signatures}
\label{sec:vect-sign-meth}

\renewcommand{\d}{\kappa}
\newcommand{\cd}{c_\text{div}}
\newcommand{\cm}{c_\text{mod}}


A \emph{vector signature} is a suitable function mapping a vector from $\{0,1\}^d$ into 
$\BO{\frac{r}{(c-1)}+\left(\frac{c}{c-1}\right)^2\log\left(\frac{1}{\varepsilon}\right)}$ bits. 
The key feature of  vector signatures is that a suitable function of the signatures of two vectors $x$ and $y$ is smaller than or equal to a certain threshold  $\Psi$ if $D(x,y)\leq r$, while it is larger than $\Psi$ with probability $1-\varepsilon$ if $D(x,y)\geq cr$, as formalized in Theorem \ref{th:mainprop}.

\subsubsection*{Signature construction.}
The construction of the signature uses four parameters $m, \cm, \cd$ and $\delta$ that all depend on $r$, $c$ and $\varepsilon$.
Their values will be provided later.

Let $M$ be a $m\times d$ random matrix with entries chosen as follows.
For each $i\in\{1,\ldots,m\},j\in\{1,\ldots,d\}$, let $M_{i,j}$ denote the element in the $i$th row and $j$th column of $M$, and let $m_{i}$ denote the $i$th row.
Every entry of $M$ is initially set to $0$. 
Then each column $j$ of $M$ is constructed by performing $\delta=\BO{1+ (c/r)\log(1/\varepsilon) }$ updates, where each update is defined by the following three steps:
\begin{enumerate}
\item Select $s$ independently and uniformly from $\{-1, 1\}$.
\item Select a row $i$ uniformly at random from $\{1,\ldots,m\}$.
\item Update the entry at $M_{i,j}$ by adding $s$.
\end{enumerate}
We let $u_i$ denote the number of updates performed on all entries of row $m_{i}$; we have that $\|m_{i}\|_1\leq u_i$ (equality may not hold since two updates can affect the same entry and cancel each other).

For notational simplicity, we introduce the $\modl$ operator: it  is similar to the standard modulo operator, but it maps into the range $\left[ -\lfloor \cm/2\rfloor, \lceil \cm/2 \rceil \right)$ (the range is symmetric around zero when $\cm$ is even). Specifically,
\[\alpha {\modl} \cm = \left (\left(\alpha+\left\lfloor \frac{\cm}{2} \right\rfloor \right) \hspace{-.8em}\mod \cm\right) -\left\lfloor \frac{\cm}{2} \right\rfloor,\]
where $\bmod$ denotes the standard modulo operation into $[0, \cm)$.  

Let $\cd,\cm$ be suitable values with asymptotic value $\BO{c}$. The \emph{signature} of a vector $x\in \{0,1\}^d$ is then the $m$-dimensional vector $\sigma(x)$ defined by
\[
\sigma(x)_i = \left\lfloor \frac{(Mx)_i \modl \cm}{\cd}  \right\rfloor.
\]
Intuitively, the signature is a {\textsc CountSketch} where we remove large values with $\modl \cm$, and remove the less significant bits with the division by $\cd$. 

The \emph{gap vector} between vectors $x$ and $y$ is the $m$-dimensional vector $\Gamma(x,y)$ where the $i$th entry is
\[
\Gamma(x,y)_i = \cd \left(\sigma(x)_i-\sigma(y)_i  \modl \cm \right).
\]
Finally, we refer to $\gamma(x,y)=\|\Gamma(x,y)\|_1$ as  the \emph{gap} between $x$ and $y$.

The following theorem describes the main property of signature vectors.
\begin{theorem}\label{th:mainprop}
Let $m=\BO{\frac{r}{(c-1)}+\left(\frac{c}{c-1}\right)^2\log\left(\frac{1}{\varepsilon}\right)}$, $\delta=\BO{1+ \tfrac{c}{r}\log(1/\varepsilon) }$, $\cd=\BO{c}$, and $\cm=\BO{c}$ be suitable values.
Then, there exists a value $\Psi=\BO{\delta r}$, such that for each pair of vectors $x,y \in \{0,1\}^d$:
\begin{itemize}
\item if $D(x,y)\leq r$, then $\gamma(x,y) \leq \Psi$;
\item if $D(x,y)> cr$, then $\gamma(x,y) > \Psi$ with probability at least $1-\varepsilon$.
\end{itemize}
\end{theorem}

We split the  proof of Theorem~\ref{th:mainprop} into two cases depending on the value of the approximation factor $c$: we first target constant approximation factors, and then we focus on larger values.
In the following proofs, we assume for notational convenience that two given vectors $x$ and $y$ differ on the first $D(x,y)$ positions.
We let $x'$ and $y'$ denote the prefix of length $D(x,y)$ of $x$ and $y$ (i.e., the positions where they differ), $M'$ denote the first $D(x,y)$ columns of $M$, $m'_i$ the $i$th row of $M'$, and $u'_i$ the number of updates affecting $m'_{i}$.

\subsubsection*{Proof of Theorem~\ref{th:mainprop} with $\mathbf{c=O(1)}$.}
For the case $c=\BO{1}$, we set the following parameters:
\begin{align*}
&m=\left\lceil 24 \frac{c^2}{c-1} \max\left\{ r, \frac{2}{c-1}\log\left(\frac{1}{\varepsilon}\right) \right\}\right\rceil,\\
&\cd = 1,\\
&\cm =2,\\
&\delta =1,\\
&\Psi =r.
\end{align*}
Note that the above values are consistent with the asymptotic values stated in Theorem~\ref{th:mainprop} since $c=\BO{1}$.
With these values, the signature definition simplifies to 
$$\sigma(x)_i = (Mx)_i \modl 2,$$ where each column of $M$ is a random vector with exactly one entry in $\{-1,1\}$ and the remaining $m-1$ entries set to zero.  
Then, the gap vector becomes:
$$\Gamma(x,y)_i=M(x-y)_i \modl 2 = M'(x'-y')_i \modl 2.$$ 
The first equality is true because  there is no rounding if $\cd=1$, and $\sigma$ is a linear function of $x$ and $y$.
The second one follows since the bit positions where $x$ and $y$ are equal do not affect the gap vector.

When $D(x,y)\leq r$, $M'$ contains at most $r$
entries in $\{-1,1\}$ and hence $\gamma(x,y)= \|M'(x'-y')\|_1\leq r$, proving the first part of Theorem~\ref{th:mainprop}.

Consider now the case $D(x,y)\geq cr$. The second part of Theorem~\ref{th:mainprop} follows by  two claims:
\begin{enumerate}[leftmargin=*,label=\emph{Claim \arabic*:}]
\item With probability at least $1-\varepsilon$, there are more than $r$ rows of $M'$ affected by an odd number of updates; we refer to these rows as \emph{odd rows}.
\item If $m'_{i}$ is an odd row, then $|\Gamma(x,y)_i|=1$.
\end{enumerate}
The two claims imply that $\gamma(x,y)=\sum_{i=1}^{m} |\Gamma_i(x,y)| > r=\Psi$ and hence Theorem~\ref{th:mainprop} follows.
The following Lemmas~\ref{lem:claim11} and~\ref{lem:claim12} show that the above claims hold.

\begin{lemma}[Claim 1]
  \label{lem:claim11}
Let $x,y$ be two input vectors in $\{0,1\}^d$, and let $M'$ be the sub-matrix of $M$ associated with the positions where $x$ and $y$ differ.
If $x$ and $y$ have distance at least $cr$, then there are more than $r$ odd rows in $M'$ with probability at least $1-\varepsilon$.
\end{lemma}
\begin{proof}
Consider the $D(x,y)$ updates used in the construction of $M'$.
If after the first $D(x,y)-cr$ updates there are more than  $(c+1)r$ rows with an odd number of updates, then the theorem follows: indeed, the remaining $cr$ updates can decrease the number of odd rows by at most $cr$.

Suppose now that there are  $Y_o\leq (c+1)r$ odd rows after the first $D(x,y)-cr$ updates, and consider  the last $cr$ updates.
Let  $Y_j$, with $j\in\{1,\ldots cr\}$ be a random variable set to 1 if the $j$th update affects an odd row, which then becomes an even row; $Y_i$ is set to 0 otherwise.
The probability that $Y_j=1$ is $p\leq (Y_o+j-1)/m\leq 3cr/m$ since there can be at most $Y_o+j-1$ odd rows before the $j$th update: the initial $Y_o$ odd rows and the rows affected by the previous $j-1$ updates. 
Let $Y=\sum_{j=1}^{cr} Y_j$. 
The expected value of $Y$ is $\mu=pcr\leq 3(cr)^2/m$.
Let $\eta=(c-1)r/(2\mu)-1$ (note that $\eta\geq 0$). By a Chernoff bound, we have
\begin{align*}
\Pr[Y\geq (c-1)r/2]=&\Pr[Y\geq \mu(1+\eta)]
\leq  e^{-\eta^2 \mu /2} \\
\leq & e^{-\left(\left(\frac{c-1}{c}\right)^2\frac{m}{24}+\frac{3(cr)^2}{2m}-\frac{(c-1)r}{2}\right)}\\
\leq & e^{-\left(\left(\frac{c-1}{c}\right)^2\frac{m}{24}-\frac{(c-1)r}{2}\right)}\\
\leq & \varepsilon.
\end{align*}

Therefore, with probability at least $1-\varepsilon$, there are  $Y< (c-1)r/2$ updates that affect odd rows and make them even. 
It follows that  the number of odd rows after all updates is then $Y_0+(cr-Y)-Y\geq  cr-2Y> r$.
\end{proof}

\begin{lemma}[Claim 2]
  \label{lem:claim12}
If row $m'_{i}$ is odd, then $|\Gamma_i(x,y)|= 1$.  
\end{lemma}
\begin{proof}
When $\delta=1$, there is one update per column and the number of non zero entries in $m'_{i}$ coincides with the number of updates (this may not happen if $\delta>1$).
Let $h_1,\ldots,h_{u_i}$ denote  the $u_i$ non zero entries in $m'_{i}$.
We have that $m_{i}(x'-y')=\sum_{j=1}^{u_i} M'_{i,h_j} (x'_{h_j}-y'_{h_j})$. Since $(x'_{h_j}-y'_{h_j})$ and $M'_{i,j}$  are in $\{-1,1\}$ and since $u_i$ is odd, then the sum must be odd and  $|\Gamma_i(x,y)|=|m'_{i}(x'-y') \modl 2|= 1$. 
\end{proof}

\subsubsection*{Proof of Theorem~\ref{th:mainprop} for $\mathbf{c=1+\Omega(1)}$.}
Let $\beta=15/(p_1 p_2)^2$ where $p_1$ and $p_2$ are suitable constants (e.g. $p_1 =  0.9$, $p_2 = 0.094$).
The proof presented here then holds for $c\geq \sqrt{5 \beta/(4p_2^2)}\approx 545$. We believe that a smaller approximation factor $c$ can be obtained with a more careful analysis of the constants.
The parameters used in the signature construction are set as follows: 
\begin{align*}
&m=\left\lceil \beta \max\left\{\frac{r}{c}, \log\left(\frac{2}{\varepsilon}\right)\right\}\right\rceil,\\ 
&\cd=  \frac{2c}{\sqrt{5}\beta},\\
&\cm = 8 c,\\
&\delta=\left\lceil \frac{c}{r} \log\left(\frac{2}{\varepsilon}\right) \right\rceil,\\ 
&\Psi = \delta r + \max\left\{r,c \log\left(\frac{2}{\varepsilon}\right)\right\}.
\end{align*}
Note that the above values are consistent with the asymptotic values stated in Theorem~\ref{th:mainprop} since $c=1+\Omega(1)$.
In contrast to the $c=\BO{1}$ case, the gap vector and the gap cannot be expressed as a function of only the positions where $x$ and $y$ differ (i.e., $x'$ and $y'$). 
In fact, due to the division by $\cd$ and the floor operation, the gap vector may  depend on the positions where $x$ and $y$ coincide.
However, we can still provide upper and lower bounds on the gap that depend only on $x'$ and $y'$.
Indeed, it holds that:
\begin{align}\label{eq:approx_gamma}
\begin{split}
  |\Gamma_i(x,y)| &>  |m'_{i} (x'-y') \modl \cm | - \cd  \\
   |\Gamma_i(x,y)| &< |m'_{i} (x'-y') \modl \cm| +\cd.
\end{split}
\end{align}
Suppose $D(x,y)\leq r$, then by (\ref{eq:approx_gamma}) the gap can be upper bounded as follows:
\begin{align*}
\gamma(x,y)&=\sum_{i=1}^{m} |\Gamma_i(x,y)|\\
& \leq \sum_{i=1}^{m} \left(|m'_{i} (x'-y')\modl \cm   |+\cd\right)\\
& \leq  \cd m+\sum_{i=1}^{m} |m'_{i} (x'-y')|\\
& \leq \max\left\{r,c \log\left(\frac{2}{\varepsilon}\right)\right\} + \delta r=\Psi.
\end{align*}
In the third step, it is crucial to use $\modl $ instead of $\mod $ since it guarantees that $| \alpha \modl \cm|\leq |\alpha|$.
The last step is true since entries in $x'-y'$ are in $\{-1,1\}$ and $M'$ contains at most $\delta r$ non-zero entries.
The first part of Theorem~\ref{th:mainprop} follows.

Suppose now that $D(x,y)\geq cr$. 
We say that row $m'_{i}$ is dense if the number of updates $u_i$ is at least $4\delta D(x,y)/(5m)$.
The proof that the gap is larger than $\Psi$ with probability at least $1-\varepsilon$ relies on the following claims:
\begin{enumerate}[leftmargin=*,label=\emph{Claim \arabic*:}]
\setcounter{enumi}{2}
\item With probability at least $1-\varepsilon/2$, the number of dense rows is at least $p_1 m$.
\item With probability at least $p_2$, we have $|\Gamma_i(x,y)|> 2c/\sqrt{5\beta}$ for a dense row $m'_i$.  
\item With probability at least $1-\varepsilon$, there are at least $0.89 p_1 p_2 m$  rows such that $|\Gamma_i(x,y)|> 2c/\sqrt{5 \beta}$.
\end{enumerate}
Then, we have that $\gamma(x,y)=\sum_{i=1}^{m} |\Gamma_i(x,y)| > 0.89  p_1 p_2 m 2c/\sqrt{5\beta}  > 3 \max\left\{r,c \log\left(\frac{2}{\varepsilon}\right)\right\}>\Psi$ since $m=\lceil\beta \max\{r/c, \log(2/\varepsilon)\}\rceil$ and $\beta=15/(p_1 p_2)^2$. Thus, the second part of Theorem~\ref{th:mainprop}  follows.

Before proving the claims in Lemmas~\ref{lem:claim1b}-\ref{lem:claim3b}, we introduce three technical lemmas.
Lemma~\ref{lem:distr} gives a load bound on a balls and bins problem by using the bounded differences method to manage dependent random variables.
Lemma~\ref{lem:mod} bounds the probability of a sum of $\{-1,1\}$ random variables to be in a specified interval after a modular operation.
Finally, Lemma~\ref{lem:sum} gives a lower bound on the tail distribution of the sum of $\{-1,1\}$ random variables by leveraging the Berry-Esseen theorem.

\begin{lemma}\label{lem:distr}
  Consider $p$ balls thrown uniformly and independently at random into $q$ bins, with $p\geq q$.
  For every $\alpha>0$ with probability at least $1-\varepsilon$, there are more than $q\left(1-e^{-\alpha}-\sqrt{ \log(1/\varepsilon)/(2 {q})}\right)$ bins with at least $\left(p/q\right)\left(1-\sqrt{2\alpha q/p}\right)$ balls.
\end{lemma}
\begin{proof}
For every $i\in\{1,\ldots, p\}$ and $j\in\{1,\ldots, q\}$, define the following random variable:
  \[X_{i,j}=\begin{cases}
    1 \text{ if ball $i$ landed in bin $j$}\\
    0 \text{ otherwise}
  \end{cases}
  \]
Let also $X_j=\sum_{i\in[p]} X_{i,j}$ be the number of balls in the $j$th bin; the expected value of $X_j$ is $\mu = p/q$ for each $j$.
  Since the balls are thrown independently a Chernoff bound gives:
  \[\Pr\left[X_j\leq \mu \left(1- \sqrt{2\alpha/\mu}\right)\right]  \leq e^{-\alpha}\]
Consider now the random variable $Y_j$:
\[Y_j=\begin{cases}
    1 \text{ if }X_j > \mu \left(1- \sqrt{2\alpha/\mu}\right)\\
    0 \text{ otherwise}
  \end{cases}
\]
Let $Y=\sum_{j=1}^{q} Y_j$; we use $Y_{Y_1,..,Y_{q}}$ to denote the actual value of $Y$ with the specified values.
Since there is dependency among the $Y_j$, we use the method of bounded differences~\cite{DubhashiP09} to bound the tail distribution, instead of a Chernoff bound.
The random variable $Y$ satisfies the Lipschitz property with constant $1$, that is:
\[|Y_{Y_1, \ldots, Y_i, \ldots,  Y_q} - Y_{Y_1,  \ldots, Y'_i, \ldots, Y_q}|  = |Y_i-Y'_i| \leq 1\] whenever $Y_i\neq Y'_i$  for every $i\in\{1,\ldots, q\}$.
By the method of bounded differences ~\cite[Corollary 5.2]{DubhashiP09}, 
we  get $\Pr\left[Y\leq \E[Y]- t\right] \leq e^{-2t^2/q}$, and
then  $\Pr\left[Y> \E[Y]- t\right] \geq 1-\varepsilon$ if $t= \sqrt{(q/2) \log(1/\varepsilon)}$.
Since 
$\E[Y]\geq q \left(1-\Pr\left[X_j\leq \mu \left(1- \sqrt{2\alpha/\mu}\right)\right]\right)\geq q\left(1-e^{-\alpha}\right),$ 
 the claim follows.
 \end{proof}
 
\begin{lemma}\label{lem:mod}
Consider a sequence $s_1,\ldots, s_k$ of independent and evenly distributed random variables  in $\{1,-1\}$, and an 
arbitrary  value $q\in \mathbb{N}$.
 Let $S=\sum_{i=1}^{k} s_i$ and $S_q = S \modl q$.
 Then for all values $a, b$ such that $0 \leq a < b \leq \lceil q/2\rceil $ and $b-a\geq q/3$, we have:
\begin{equation}
\label{eq:boundmod}
\frac{\Pr[|S|\geq a]}{2} < \Pr[a\leq |S_q| <b ]  < \Pr[|S|\geq a] .
\end{equation}
\end{lemma}
\begin{proof}
Let $k'= k/q$ and assume for the sake of simplicity that $k'$ is an integer, and that $q$, $b$ and $a$ are even (the proof extends to the general case with minor adjustments). We define the following four quantities:
\begin{eqnarray*}
H_1 \hspace{-1em}&= \sum_{\ell=0}^{k'-1} \hfill  &\Pr\left[\ell q + a \leq |S| < \ell q + b\right]; \\
 H_2 \hspace{-1em} & = \sum_{\ell=0}^{ k'-1} &   \Pr\left[\ell q + b \leq |S| \leq (\ell+1)q -b\right] ; \\
 H_3 \hspace{-1em} &= \sum_{\ell=0}^{ k'-1} & \Pr\left[( \ell+1) q -b < |S| \leq (\ell+1) q-a\right]; \\
H_4 \hspace{-1em} &= \sum_{\ell=0}^{ k'-1}  &\Pr\left[(\ell+1)q - a < |S| < (\ell+1)q +a\right].
\end{eqnarray*}
Standard computations show that:
$\Pr[a \leq |S_q| < b]= H_1+H_3$ and that $\Pr[|S| \geq a]=H_1+H_2+H_3+H_4$.
We then have that  $\Pr[a\leq |S_q| <b ] < \Pr[|S|\geq a]$, and the right side of the inequality in~(\ref{eq:boundmod}) follows.

We now focus on the other side of the inequality.
We  prove that $H_1\geq H_2+H_4$.
The random variable $S$ has value $i$, with $i\in [-k,k]$ if there are $(k+i)/2$ terms set to $+1$ and $(k-i)/2$ terms set to $-1$.
If $k+i$ is odd, this cannot happen and hence $\Pr[S=i]=0$.
On the other hand, if $k+i$ is even, we  have $\Pr[S=i] = \binom{k}{(k+i)/2} \frac{1}{2^k}$  since
the $s_i$ terms are independent and evenly distributed.
Note that $\Pr[S=i]$ is decreasing for even values of~$i$.

Let us define ${\alpha \brack \beta/2}$ to  $\binom{\alpha}{\beta/2}$ if $\beta$ is even and to $0$ if $\beta$ is odd: we thus have $\Pr[S=i]={k \brack (k+i)/2}$ for any even/odd $i$.
Let $\beta\geq \alpha$ and $\gamma\geq 1$, we have the following property:
\begin{align*}
{\alpha \brack \beta/2} + 
{\alpha \brack (\beta+1)/2} >
{\alpha \brack (\beta´+\gamma)/2}+
{\alpha \brack (\beta´+\gamma+1)/2}.
\end{align*}
The correctness of the property follows from the fact that there is exactly one non zero term on each side of the inequality by definition of  ${\alpha \brack \beta/2}$, and the non zero one on the right is decreasing in $\gamma$.

We  then have, for any integer $\ell\geq 0$, that :
\begin{align*}
\Pr[a + \ell q &\leq |S| < b +  \ell q]  = 2\sum_{j=a + \ell q }^{ b +  \ell q-1}  {k \brack \frac{(k+j)}{2}} \frac{1}{2^k}\\
\geq 2& \sum_{j=a + \ell q }^{a+(\ell+1) q -2b}  {k \brack \frac{(k+j)}{2}} \frac{1}{2^k}\\ 
&+ 2 
\sum_{j=a+(\ell+1) q -2b+1}^{a+(\ell+1) q -2(b-a)-1}  {k \brack \frac{(k+j)}{2}} \frac{1}{2^k},
\end{align*}
where the  step follows by the initial assumption  $(b-a)\geq q/3$.
By using the  above property of ${\alpha \brack \beta/2}$, we shift the indexes of the above summations (we add $b-a$ to the first sum and $2(b-a)$ to the second one):
\begin{align*}
\hspace{-.3em}\Pr&[a + \ell q \leq |S| < b +  \ell q]  \\
\hspace{-.3em} > &  2\sum_{j=\ell q+b}^{(\ell +1) q-b} {k \brack \frac{(k+j)}{2}} \frac{1}{2^k} +  2\sum_{j=(\ell+1) q-a+1}^{(\ell+1) q+a-1} {k \brack \frac{(k+j)}{2}} \frac{1}{2^k} \\
\hspace{-.3em} \geq & \Pr[\ell q +b {\leq } |S| {\leq} (\ell+1)q-b] \\ & \hspace{2em}+\Pr[(\ell+1)q-a {<} S {<} (\ell+1)q +a]\\
\hspace{-.3em}  \geq & H_2+H_4. 
\end{align*}
(Note that the derivation requires some adjustments when $q$, $b$ or $a$ are not even).
Therefore,  $\Pr[|S|\geq a] = H_1+H_2+H_3+H_4 <  2(H_1+H_3) \leq 2 \Pr[a\leq |S_q| <b ] $. 
The left side of the inequality in~(\ref{eq:boundmod}) follows.
 \end{proof}

\begin{lemma}\label{lem:sum}
Let $S=\sum_{i=1}^{k} s_i$, where the $s_i$ terms are independent and unbiased random variables in $\{-1,+1\}$, and let $\alpha>0$ be any arbitrary value. 
Then, 
$$\Pr[|S|\geq  \alpha \sqrt{k}]\geq \frac{2\alpha}{\sqrt{2\pi}(\alpha^2+1)e^{\alpha^2/2}}-\frac{1}{2\sqrt{k}}.$$
\end{lemma}
\begin{proof}
We observe that $\E[s_i]=0$, $\sigma^2 = \E[s^2_i]=1$ and $\rho = \E[|s_i|^3]=1$.
By the Berry-Esseen theorem~\cite{Berry1941}, we have that the random variable $Q = S/(\sqrt{k}\sigma)=S/\sqrt{k}$ can be approximate by a standard normal distribution $\mathcal N(0,1)$ with error  
$$
| \Pr[Q\leq x] -\Psi(x) | \leq \frac{C \rho}{\sigma^3 \sqrt{k}},
$$ 
where $\Psi(x)$ is the cumulative distribution function of the standard normal distribution $\mathcal N(0,1)$ and $C$ is a suitable constant smaller than $1/2$~\cite{Tyurin10}.
The above inequality can be rewritten as
\[
| \Pr[Q> x] -\Psi^c(x) | \leq \frac{1}{2\sqrt{k}},
\] 
with $\Psi^c(t)=1-\Psi(x)$. 
We then get
\begin{align*}
\Pr[|S|\geq \alpha \sqrt{k} ] &= 2\Pr[S\geq \alpha \sqrt{k} ] \\ 
&= 2 \Pr[Q\geq \alpha]\\
&\geq 2\Psi^c(\alpha) - \frac{1}{2\sqrt{k}}.
\end{align*}
Since  $\Psi^c(x)\geq x/(\sqrt{2\pi} (x^2+1) e^{x^2/2})$~(See e.g. ~\cite{Cook09,Abramowitz74}), the lemma follows by inserting the bound for $\Psi^c(x)$.
\end{proof}

We are now ready to prove the three claims used in the proof of Theorem~\ref{th:mainprop} for $c=\BOM{1}$.

\begin{lemma}[Claim 3]\label{lem:claim1b}
With probability at least $1-\varepsilon/2$, the number of dense rows in $M'$ is at least $p_1 m$, with $p_1 = 0.9$.
\end{lemma}
\begin{proof}
Matrix $M'$ is obtained by performing $\delta$ random updates per column independently and uniformly distributed.
The number of updates $u_i$ affecting row $m'_{i}$ is distributed as the number of balls in a bin after randomly throwing $\delta D(x,y)$ balls into $m$ bins.
By applying Lemma~\ref{lem:distr} with $\alpha=3$, it follows that, with probability at least $1-\varepsilon/2$, there are more than
\[m'\geq (1-1/e^3-\sqrt{\log(2/\varepsilon)/(2m)}) m\geq p_1 m\]  rows where
\begin{align*}
u_i& \geq \frac{\delta D(x,y)}{m} \left(1-\sqrt{\frac{6 m}{\delta  D(x,y)}}\right)\\
   & \geq \frac{4\delta D(x,y)}{5m}
\end{align*} 
as soon as $c\geq 5 \sqrt{6\beta+1}$ (which is true under the initial hypothesis $c\geq \sqrt{5\beta/(4p_2^2)}$). These $m'$ rows are then dense.
\end{proof}

\begin{lemma}[Claim 4]\label{lem:claim2b}
If $m'_{i}$ is dense, then $|\Gamma_i(x,y)|> 2c/\sqrt{5\beta}$ with  probability at least $p_2=0.094$.
\end{lemma}
\begin{proof}  
Let $K=2c/\sqrt{5\beta}(1+1/\sqrt{\beta})$ and assume that the inequality  $|m'_{i} (x'-y') \modl \cm|\geq  K$ holds. 
Then, the lemma follows by applying~(\ref{eq:approx_gamma}):
\begin{align*}
|\Gamma_i(x,y)| & > |m'_{i} (x'-y') \modl \cm| - \cd \\
& \geq K - \cd\\
& = c/(\sqrt{5}\beta) + 2c/\sqrt{5\beta}- \cd\\
& = 2c/\sqrt{5\beta}.
\end{align*}

We now show that the above inequality holds (i.e., $|m'_{i} (x'-y') \modl \cm|\geq  K$).
The inner product $m'_{i} (x'-y')$ can be rewritten as $\sum_{j=1}^{u_i} \sigma_j (x'-y')_{f(j)}$, where $f(j)$ is the position in $m'_{i}$ affected by the $j$th update.
Since $(x'-y')$ has entries in $\{-1,1\}$ and the $\sigma_j$ are independent,  $m'_{i} (x'-y')$
has the same density function as $S=\sum_{j=1}^{u_i} \sigma_j$.
Then,
\begin{align*}
\Pr[|M'_i & (x'-y') \modl \cm|\geq  K]\\
  &=
\Pr[|S \modl \cm| \geq K] \\
& > \frac{\Pr[|S|\geq K]}{2},
\end{align*}
where the last step  follows by applying Lemma~\ref{lem:mod} with $a=K$, $b=\cm/2$ and $q=\cm$ (note that $b-a\geq \cm/3$).
To lower bound $\Pr[|S|\geq K]$, we apply Lemma~\ref{lem:sum} with $\alpha=1+1/\sqrt{\beta}$ since $K\leq \sqrt{u_i}(1+1/\sqrt{\beta})$. Hence,
\begin{align*}
\frac{\Pr[|S|\geq K]}{2}&\geq
\frac{\Pr[|S|\geq (1+1/\sqrt{\beta})\sqrt{u_i}]}{2}\\
&\hspace{-2em}\geq \frac{1+1/\sqrt{\beta}}{\sqrt{2\pi} ((1+1/\sqrt{\beta})^2+1) e^{(1+1/\sqrt{\beta})^2/2}}-\frac{1}{4\sqrt{u_i}}\\
&\hspace{-2em}\geq p_2,
\end{align*}
where the last step follows by observing that $\sqrt{u_i}\geq 2c/\sqrt{5\beta}  \geq 1/p_2$, and then by numerically evaluate the resulting bound.
\end{proof}

\begin{lemma}[Claim 5] \label{lem:claim3b}
With probability at least $1-\varepsilon$, there are at least $0.89 p_1 p_2 m$ rows such that $|\Gamma_i(x,y)|> 2c/\sqrt{5\beta}$.
\end{lemma} 
\begin{proof}
By Lemma~\ref{lem:claim1b}, there are $m'\geq p_1 m$ dense rows with probability $1-\varepsilon/2$. 
For each dense row, let $Y_i$ be a random variable sets to 1 if $|\Gamma_i(x,y)|> c/\sqrt{\beta}$, and 0 otherwise.
By the previous Lemma~\ref{lem:claim2b}, we have that $\Pr[Y_i=1]\geq p_2$.
Let $Y=\sum_{i=1}^{m'} Y_i$.
Since the $Y_i$ are independent and $\E[Y]=p_2 m'$, a Chernoff bound gives:
\[
  \Pr\left[Y< p_2 m'\left(1- \sqrt{2\log(2/\varepsilon)/(p_2 m')}\right)\right]\leq \varepsilon/2.
\]
By plugging in the actual values of variables, 
we have $ \Pr[Y< 0.89 p_1 p_2 m]\leq \varepsilon/2$.

Therefore, by an union bound there are at least $ p_1 m$ dense rows and 
at least $0.89 p_1 p_2 m$ of them satisfy
 $|\Gamma_i(x,y)|> c/\sqrt{\beta}$.
\end{proof}

\subsection{A filter with point-wise error}\label{sec:filter-with-wc}
A distance sensitive approximate membership filter with point-wise error is obtained by just storing the $n$ signatures of the points in $S$. 
We have the following theorem:

\begin{theorem}
There exists a $(r,c,\varepsilon)$-distance sensitive approximate membership filter with point-wise error which requires 
$$
\BO{n\left(\frac{r}{(c-1)}+\left(\frac{c}{c-1}\right)^2\log\left(\frac{n}{\varepsilon}\right)\right)}
$$
bits for any $c>1$ on  a set $S$ of $n$ points.
When $c\geq 2$, the filter uses 
$\BO{n\left(\frac{r}{c} + \log\left(\frac{n}{\varepsilon}\right)}\right)$ bits, and it is optimal if  $r/c\geq \log(n/\varepsilon)$ or $\varepsilon\leq 1/n^{1+o(1)}$.
\end{theorem}
\begin{proof}
We assume a shared source of randomness that can be used to recover the random matrix $M$ without storing it.
Consider the $n$ signatures of points in $S$ constructed with error $\varepsilon'=\varepsilon/n$. 
By an union bound, the $n$ signatures give a false positive with probability $\varepsilon$. 
Since each signature requires 
$
\BO{\frac{r}{(c-1)}+\left(\frac{c}{c-1}\right)^2\log\left(\frac{n}{\varepsilon}\right)}
$
bits by Theorem~\ref{th:mainprop}, the first part of the claim follows.
The optimality with $c\geq 2$ of the filter follows from Theorem~\ref{thm:wclb1}.
\end{proof}

\subsection{A filter with average error}
\label{sec:filter-with-average}
The point-wise error filters are of course  valid average error filters, but in this setting we can also construct space efficient filters with a $c=1$ approximation factor. 
Define $Q_{\text{$r$-far}}=\{x\in\{0,1\}^d\; | \; D(x,S)\geq r\}$ and similarly $Q_{\text{$(r;cr)$-far}}=\{x\in\{0,1\}^d\; | \; r\leq D(x,S)\leq cr\}$.

By setting $c=r$ in the point-wise filter, we obtain an average error filter with $c=1$ which matches the $\BOM{n\log(1/\varepsilon)}$ lower bound of Theorem~\ref{avg_error_thm} for small $r$. 
Interestingly, this space bound shows that it is possible to support distance sensitive membership queries in the average error setting with the asymptotic space bound of a Bloom filter.
\begin{theorem}
Let $r \leq \sqrt{d}$, $n\leq 2^{d/3}$ and $\varepsilon\geq 1/2^{d-2}$. Then, there exists an optimal $(r,1,\varepsilon)$-distance sensitive approximate membership filter with average error which requires  $\BO{n \log(1/\varepsilon)}$ bits on  a set $S$ of $n$ points.
\end{theorem}
\begin{proof}
Let us consider a $(r,r,\varepsilon/4)$-filter $\mathcal F$ with point-wise guarantees.
The amount of false positives accepted by $\mathcal F$ is $P \leq (\varepsilon/4) |Q_\text{$r^2$-far}| + |Q_\text{$(r;r^2)$-far}|$.
We have $|Q_\text{$(r;r^2)$-far}|\leq n r^2 \binom{d}{r^2} \leq (\varepsilon/4) 2^d$ since
$d \geq r^2$, $n\leq 2^{d/3}$ and $\varepsilon \geq 4/2^{d/2}$.
Trivially, we also have that $|Q_\text{$r^2$-far}|\leq 2^d$.
We see that $P\leq  \varepsilon 2^{d-1}$.  

Now note that $|Q_\text{$r$-far}|\geq 2^d-n r \binom{d}{r}\geq 2^{d-1}$ by $d \geq r^2$ and $n\leq 2^{d/3}$.

We combine the two bounds to see $P\leq \varepsilon 2^{d-1}\leq \varepsilon|Q_\text{$r$-far}|$. 
The optimality of $\mathcal{F}$ follows from Theorem~\ref{avg_error_thm} since $r^2/d<1$ and $n \log(1/\varepsilon)$ is a lower bound.
\end{proof}



\section{Conclusion}
To the best of our knowledge, this is the first time upper and lower space bounds are given for the problem of distance sensitive filters without false negatives.
We have introduced distance sensitive signatures for Hamming vectors and used them to derive filters with point-wise and average errors. The proposed filters are optimal under certain assumptions, but it is an open question to close the gap without these assumptions, specifically when $\varepsilon$ is large.

Another interesting research direction is to investigate trade-offs between space and query time: our filter requires reading all signatures at query time and it is not clear to which extent the query time can be improved.

\chapter{Fast Nearest Neighbor Preserving Embeddings}
\label{cha:simil-pres-embedd}

  In this Chapter we show an analogue to the Fast Johnson-Lindenstrauss Transform for Nearest Neighbor Preserving Embeddings in $\ell_2$.
  These are randomized embeddings that preserve the (approximate) nearest neighbors for a set of points.
  The dimensionality of the embedding space is bounded not by the size of the embedded set $n$, but by its doubling dimension $\lambda$. For most large real-world datasets this will mean a considerably lower-dimensional embedding space than possible when preserving all distances.
  However the embedding is slow since it requires multiplication with a dense matrix.
  To reduce the embedding time we propose a sparse mapping.
  The resulting embeddings can be used with existing approximate nearest neighbor data structures to yield speed improvements.

\section{Introduction}

Many algorithmic problems become overwhelmingly difficult in high-dimensional settings.
One way of trying to combat this problem is to discover mappings that preserve the metric relevant to solving a given problem, while embedding it into a lower dimensional setting.
Most famously Johnson and Lindenstrauss~\cite{Johnson1984} showed the lemma:
\begin{lemma}[JL-Lemma~\cite{Johnson1984}]
  \label{sec:JLlemma}
  For any integer $d>0$, and any $\epsilon  >0$, $\delta\in (0,1/2)$, for $k=\Theta(\epsilon^{-2}\log(1/\delta))$ there exists a distribution $\Pi$ such that for $k\times d$ matrices $M\sim \Pi$, for any $x\in\mathbb{R}^d$,
\[Pr\left[(1-\epsilon )\|x\|_2\leq \|Mx\|_2\leq(1+\epsilon)\|x\|_2\right]> 1-\delta\]
\end{lemma}
The JL-Lemma shows the existence of an embedding of any set $S\subseteq \mathbb{R}^d$ into $k=\BO{\log|S|\epsilon^{-2}}$ dimensions while preserving $\ell_2$ distances up to a multiplicative $(1\pm\epsilon)$ distortion. Proofs can be found for many different $M$~\cite{FRANKL1988355,Johnson1984, Dasgupta2003,Achlioptas2003671}, including Gaussian matrices~\cite{Johnson1984, Dasgupta2003} and $\{0,\pm1\}$ matrices~\cite{Achlioptas2003671}.
In fact we might use any sub-gaussian distribution to fill the matrix~\cite{IN07}. 
These low-dimensional embeddings can be used to speed up many fundamental high-dimensional problems like closest pair, nearest neighbor or minimum spanning tree.
They can also be used to decrease the storage requirements of a dataset when we only need to preserve norms. Further discussion and examples can be found for instance in \cite{vempala2004,Indykapplication}.
It is known that if we want to preserve the norm for all $x\in S$, the embedding dimension $k=\BOx{\log|S|\epsilon ^{-2}}$ is optimal, see \cite{Larsen17,Larsen16}.

However it might not be necessary to preserve norms for all \emph{all} points in $S$. If for example we are interested in nearest neighbor queries we require only that neighbors remain close to each other, while far away points do not get too close.
This idea was introduced and formalized as Nearest Neighbor Preserving Embeddings by  Indyk and Naor~\cite{IN07}, who also presented an embedding.
Using a full Gaussian matrix they showed that nearest neighbor distance can be preserved while embedding into fewer dimensions than in the distance preserving setting.
Specifically, $k$ is $\BOx{\epsilon^{-2}\log \lambda_S \log(2/\epsilon) }$ where $\lambda_S$ is the doubling constant of $S$.
By removing the requirement that all distances be preserved we can get $k$ smaller than in the bounds discussed above~\cite{Larsen17,Larsen16,2016arXiv161000239A}.

Another line of research has focused on improving the speed of the embeddings by using sparse matrices while keeping the distortion low~\cite{DBLP:journals/jacm/KaneN14,Dasgupta:2010:SJL:1806689.1806737,Achlioptas2003671,Ailon09}.
Call $f<1$ the sparsity parameter\footnote{Normally $q$ is used for this, but in this dissertation we reserve $q$ for query points}.
If each entry in the used matrix is $0$ with probability $1-f$ we can improve the embedding time from $\BOx{dk}$ to expected time $\BOx{dkf}$ by sparse matrix multiplication.
A classic sparse matrix construction is the Fast Johnson Lindenstrauss Transform (FJLT) $\Phi:\mathbb{R}^d\rightarrow\mathbb{R}^k$~\cite{Ailon09}.
In this chapter we show that the FJLT is in fact a Nearest Neighbor Preserving embedding with $k=\BOx{\epsilon^{-2}\log \lambda_S \log(2/\epsilon) }$ and sparsity parameter $f=\BOx{\log^2{n}/d}$ for $\BO{d\log d + \epsilon^{-2}\log^3n}$ evaluation time.

\section{Preliminaries}

\begin{definition}[Nearest Neighbor Preserving Embeddings~\cite{IN07}]
    \label{def:NNPE}
  Let $\epsilon, \delta \in (0,1)$, and let $S$ be a set of points in
  $\mathbb{R}^d$. For any point $x\in S$ let $x'$ denote the point closest
  to $x$ in $S\setminus\{x\}$ under the $\ell_2$ norm.  We say that an
  embedding $\Phi:\mathbb{R}^d\rightarrow \mathbb{R}^k$ is nearest
    neighbor preserving with parameters $(\epsilon, \delta)$ if for every 
  $x \in S$, the following properties hold with probability at least
  $\delta$:
  \begin{enumerate}
  \item $\min\limits_{z\in S\setminus\{x\}}\|\Phi x-\Phi
    z\|_2\leq(1+\epsilon)\|x-x'\|_2$, and 
  \item $\forall y\in S$:\\ If $\|x-y\|_2> (1+2\epsilon)\|x-x'\|_2$ then $\|\Phi x-\Phi y\|>(1+\epsilon)\|x-x'\|_2$.
  \end{enumerate}
\end{definition} 

\begin{definition}[Fast Johnson-Lindenstrauss Transform~\cite{Ailon09}]
  \label{def:FJLT}
Let an embedding $\Phi$ be defined by a $k \times d$ matrix
  $\Phi:=\textbf{PHD}$ as follows: $\textbf{D}$ is a random $\pm1$ diagonal $d \times d$
  matrix, $\textbf{H}$ is the $d$-dimensional Walsh-Hadamard transform, and $\textbf{P}$
  is a $k \times d$ matrix with entries 
  \[p_{ij}=\begin{cases}X\sim\mathcal{N}(0,f^{-1}) \qquad & \text{ w.p. } f\\
    0 & \text{ w.p. } 1-f \end{cases}~~~~.\] Here $f$ is the expected fraction of non-zero entries,
called the \emph{sparsity parameter} of the FJLT\footnote{We typeset the three matrices with bold to avoid confusion with the definitions of $D$ and $H$ already in use}.
\end{definition}

\begin{definition}[Doubling constant $\lambda_S$]
  \label{def:doubling}
  The \emph{doubling constant} $\lambda_S$ of a point set $S \subseteq \mathbb{R}^d$
  is defined to be the smallest integer $\lambda$ such that for every $x
  \in S$, and every $r > 0$, the point set $\B{x}{r}\cap S$ can be covered
  by at most $\lambda$ balls $\B{z}{r/2}$ where $z\in S$.  We refer to
  $\log_2 \lambda_S$ as the \emph{doubling dimension} of $S$.
\end{definition}

\section{Fast Nearest Neighbor Preserving Embeddings}

Given the definitions above let us state the claim:

\begin{theorem}[Fast Nearest Neighbor Preserving Embeddings]
  \label{thm:fast-near-neighbor}
  For any $S\subseteq \mathbb{R}^d, \epsilon \in (0,1)$ where $|S|=n$ and
  $\delta\in (0,1/2)$ for some
  \[k=\BO{\frac
    {\log{(2/\epsilon)}}{\epsilon^2}\log{(1/\delta)}\log{\lambda_S}}\]
  there exists a nearest neighbor preserving embedding $\Phi:\mathbb{R}^d\rightarrow\mathbb{R}^k$ with
  parameters $(\epsilon,1-\delta)$ requiring expected \[\BO{d\log(d)+\epsilon^{-2}\log^3{n}}\] operations.
\end{theorem}

By picking $\delta$ we can fix the probability of successfully sampling an embedding that is nearest neighbor preserving and close to the expected number of operations.
Indyk and Naor presents a proof for embeddings that are constructed using
full $k\times d$ Gaussian matrices $G$ (see~\cite[Theorem 4.1]{IN07}).
Requiring $\BOx{kd}$ operations to embed each point.
Our contribution will be to show how their techniques can be applied to sparse embeddings.
We first identify the properties of a map that are sufficient for the Indyk-Naor proof to hold, and then construct sparse embeddings exhibiting the properties with a bounded probability of error.

\begin{definition}
  \label{def:IN-prop}
  Let $\epsilon \in (0,1)$. We say that a distribution over maps
  $\Phi=\textbf{PHD}: \mathbb{R}^d \to \mathbb{R}^k$ satisfies the \emph{Indyk-Naor property}
  for a set $S \subseteq \mathbb{R}^d$ with error $\eta \geq 0$ if with probability
  $1-\eta$ over the choice of $\textbf{D}$, the map satisfies that for all $x\in
  S$, $y\in S\cup \{0\}$
  \begin{enumerate}
  \item[(P1)] $\Pr_{\textbf{P}}[\|\Phi (x-y)\|_2 \not\in (1 \pm \epsilon ) \|x-y\|_2 ]\leq
    e^{-\BOMx{k\epsilon ^2}}$, and
  \item[(P2)] $\Pr_{\textbf{P}}[\|\Phi x\|_2\leq \epsilon \|x\|_2 ]\leq
    (3 \epsilon)^k$.
  \end{enumerate}
  Note that the above probabilities are taken only over the choices of
  $\textbf{P}$.
\end{definition}
By bounding $\eta$ with a constant $<1$ we will then be able to extend the proof presented by Indyk and Naor to show the correctness of Theorem~\ref{thm:fast-near-neighbor}.
We will then need to increase $k$ by a corresponding constant to make up for the $\eta$ loss, but the order of $k$ remains unchanged.

We will show that the FJLT\cite{Ailon09} satisfies the Indyk-Naor properties. 
The first property to satisfy is the normal Johnson-Lindenstrauss property, but it is required to hold also for all difference vectors possible from $S$.
The second property is stronger, when $\epsilon\ll1/3$.
We will be referring to $\emph{P1}$ and $\emph{P2}$ as the Distortion and Shrinkage bound respectively.

\subsection{Smoothness}
\label{sec:smooth-setting}

Before we show the two properties from Definition~\ref{def:IN-prop} we will bound the probability of the diagonal matrix $\textbf{D}$ being in a ``smooth'' setting.
Our later proofs of the Distortion and Shrinkage bounds will be conditioned on this.
We call a vector $x \in \mathbb{R}^d$ $s$-smooth if $\| x \|_\infty \leq s \, \|
x\|_2$. Note that since $\textbf{H}$ and $\textbf{D}$ are isometries
$\|\textbf{HD}x\|_2=\|x\|_2$.

\begin{definition}
  For any $s>0$ we say that a given diagonal matrix $\textbf{D}$ is in an $s$-smooth setting if
  \[\forall {x,y\in S\cup\{0\}}, \|\textbf{HD}(x-y)\|_\infty\leq s\|x-y\|_2.\]
\end{definition}

In this section we will bound the probability of $\textbf{D}$ \emph{not} being in an $s$-smooth setting for $s=\BO{\sqrt{\frac{\log n}d}}$, and then in Section~\ref{sec:distortion-bound} and \ref{sec:shrinkage-bound} we show how the Distortion and Shrinkage bounds follow from smoothness.

Let us first consider a single vector $z=(x-y)$ where $x,y\in S\cup \{0\}$.
Assume $\|\textbf{HD}z\|_\infty\geq s\|z\|_2$ then there is some entry $1\leq i\leq d$ such that $|(\textbf{HD}z)_i|\geq s\|z\|_2$.
Let $b=\tfrac 1 {\|z\|_2}$, then $|(\textbf{HD}z)_{i}b|\geq s$ and 
\[\Pr[\|\textbf{HD}z\|_\infty\geq s\|z\|_2]= \Pr[\|\textbf{HD}zb\|_\infty\geq s]\]
where $zb$ is a unit vector. So without loss of generality we can focus on unit vectors:

\begin{lemma}
  \label{lm:inftyx}
  Given a unit vector $x$ in $\mathbb{R}^d$, for any $s>0$ 
  \[ \Pr[\| \textbf{HD}x \|_\infty \geq s
  ] \leq 2de^{-s^2d/2}.\]
\begin{proof}
  See \cite{Ailon09} or \cite{mitzenmacher2005probability}(p.69).
  In short let $u=\textbf{HD}x=(u_1,..,u_d)^T$, so $u_1=\sum_i^dh_ix_i$ where the $h_i$ are i.i.d. uniformly from $\{d^{-1/2},-d^{-1/2}\}$.
  We use:
\begin{align*}
  \E[e^{sdu_1}]&=\prod_i^d\E[e^{sd h_i x_i}]=\prod_i^d\frac{1}{2}(e^{s\sqrt{d} x_i}+e^{-s\sqrt{d} x_i})\\
               &\leq\exp(s^2d\sum_{i=1}^dx_i^2/2)\\
               &=e^{s^2d\|x\|_2^2/2}
\end{align*}
In a standard Chernoff bound~(See Section~\ref{thm:markov}).
\end{proof}
\end{lemma}

As a small contribution we now show a slightly better bound for our setting based on approximating the Kinchine inequality constants.
We use the fact that $s$ will be bounded away from $0$ like $\BOMx{d^{-1/2}}$.

\begin{lemma}
  \label{lm:kinchine}
  Given $c_s>2$ and a unit vector $x$ in $\mathbb{R}^d$, for $s\geq\sqrt{{c_s}/{d}}$ 
  \[ \Pr[\| \textbf{HD} x \|_\infty \geq s ] \leq de^{-s^2d\ln(\frac{c_se}{c_s+1})/2}.\]

\begin{proof}
  Let $u=\textbf{HD}x=(u_1,..,u_d)^T$, so $u_1=\sum_i^dh_ix_i$ where the $h_i$ are i.i.d. uniformly from $\{d^{-1/2},-d^{-1/2}\}$.
  Let $\pm x_i$ denote a uniformly random variable from $\{x_i,-x_i\}$.
  For all $p\geq1$ by Markov's inequality:
  \begin{equation}
    \label{eq:markovfastnn}
    \Pr[|u_1| \geq s]
    = \Pr\left[\left|\sum_{i=1}^d\pm x_i\right|^p \geq (\sqrt{d}s)^p\right]
    \leq \frac{\E\left[\left|\sum_{i=1}^d\pm x_i\right|^p\right]}{(\sqrt{d}s)^p}
  \end{equation}
  
  By the Kinchine inequality there is some constant $B_p$ such that:
  \[\E\left[\left|\sum_i^d\pm x_i\right|^p\right]\leq B_p\|x\|_2^p\]

  For $p>2$ Haagerup~\cite{Haagerup1981} showed that $B_p=2^{(p-2)/2}\frac{\Gamma(\frac{p+1}{2})}{\Gamma(3/2)}$ (See also~\cite{Nazarov2000}).
  Since $\Gamma(3/2)=\frac {\sqrt{\pi}} 2$ we can simplify this to
  \[B_p=2^{\frac p2}\frac{\Gamma(\frac{p+1}{2})}{\sqrt{\pi}}.\]
  
  Now we use that for $x>1$, $\Gamma(x)\leq \frac {x^{x-1/2}} {e^{x-1}}$~\cite{li2007inequalities}:
  \begin{align*}
    B_p&\leq\frac{\sqrt{2}^p}{\sqrt{\pi}}\frac{(\frac{p+1}{2})^{(\frac{p+1}2)-1/2}}{e^{(\frac{p+1}2)-1}}\\
       &=\frac{\sqrt{p+1}^p}{\sqrt{\pi}\sqrt{e}^{p-1}}\\
       &=\sqrt{\frac{e}{\pi}}\sqrt{\frac{p+1}e}^p\leq\sqrt{\frac{p+1}e}^p
  \end{align*}
  
  Plugging back into~\ref{eq:markovfastnn} we have:
\[    \Pr[|u_1| \geq s]\leq\sqrt{\frac{p+1}{es^2d}}^p\]
  We now set $p=s^2d$ to get:
  \begin{equation}
    \Pr[|u_1| \geq s]\leq\left(\frac 1 e+\frac 1{es^2d}\right)^{s^2d/2}
  \end{equation}
  Which gives the result when we use the constraint on $s$.
\end{proof}

\end{lemma}

\subsection{Fixing $s$ and $f$}

Now let $s=d^{-1/2}\sqrt{c\ln(n^2d)}$.
We want to set $c$ as small as possible, but such that $D$ is in an $s$-smooth setting with probability at least $\frac {19} {20}$.
Using lemma~\ref{lm:inftyx} as in \cite{Ailon09} we can get $c=8$, but using lemma~\ref{lm:kinchine} with $c_s=c\ln(n^2d)$ for $x\neq y$ we get:
\begin{align*}
  \Pr[\exists x,y \in S\cup \boldsymbol{0},\|\textbf{HD}(x-y)\|_\infty\geq s\|x-y\|_2]&\leq \frac{n^2d}{e^{c\ln(n^2d)\ln(\frac{c_se}{c_s+1})/2}}\\
                                                                   &=e^{-c\ln(\frac{c_se}{c_s+1})/2}\enspace.
\end{align*}

\begin{wrapfigure}{O}[2cm]{5cm}
  \centering
  \includegraphics[width=5cm]{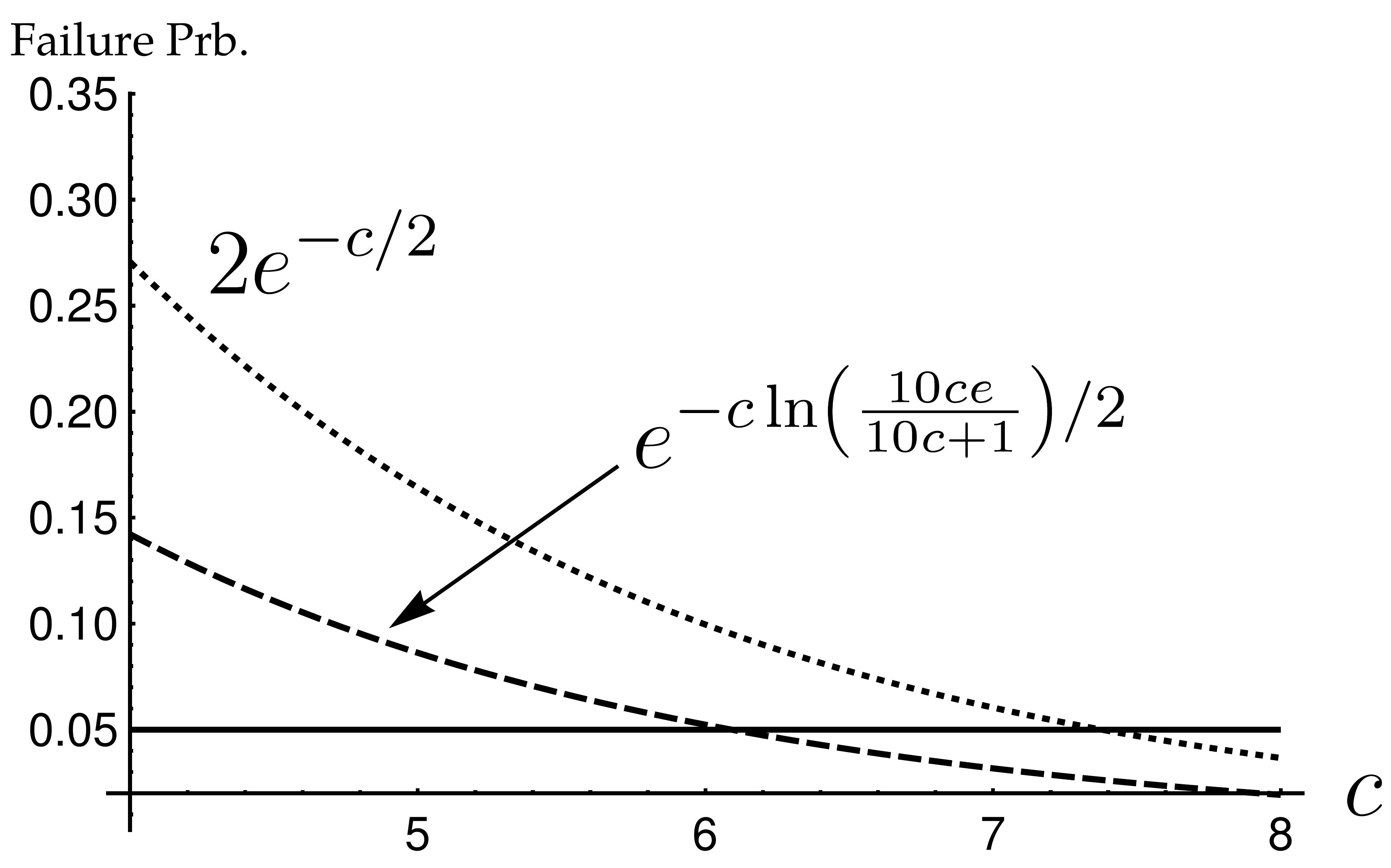}
  \caption{Bounds from lemma~\ref{lm:inftyx} and~\ref{lm:kinchine}}
  \label{fig:lemmacompare}
\end{wrapfigure}

Which evaluates to below $19/20$ for $c=7$, even if we only assume $\ln(n^2d)\geq1$.
Figure~\ref{fig:lemmacompare} shows a comparison between lemma \ref{lm:inftyx} and lemma~\ref{lm:kinchine} if we assume $\ln(n^2d)\geq10$.

In the following we will let $\Phi$ be a FJLT embedding constructed by setting $f=\min\left(c's^2,1\right)$ where $c'>0$ is some universal constant.
We will then show that if $\textbf{D}$ is $s$-smooth, this setting of $f$ makes $\Phi=\textbf{PDH} $ satisfy the distortion and shrinkage bounds.

\subsection{Distortion bound}
\label{sec:distortion-bound}

\begin{lemma}[Distortion bound] 
  For any $x,y\in S\cup \{0\}$ if $D$ is in an $s$-smooth setting, for $\epsilon>0$:
  \label{lm:distortion}
  \[\Pr\left[\|\Phi(x-y)\|_2\notin (1\pm\epsilon)\|x-y\|_2 \right]\leq e^{-\BOMx{k\epsilon ^2}}\]

  \begin{proof}
    The distortion bound is the main result in~\cite{Ailon09}.
    
\end{proof}
\end{lemma}

\subsection{Shrinkage bound}
\label{sec:shrinkage-bound}

The shrinkage bound is stronger than the distortion bound for large $\epsilon$.
We will need it later to confine the probability of any of an infinite series of events happening to a small constant.

\begin{lemma}[Shrinkage bound]
  For a fixed vector $x\in S$, if $D$ is in an $s$-smooth setting,
  for $\epsilon \in (0,1)$:
  \label{lm:fractionshrink}
  \[\Pr_{\textbf{P}}\left[\|\textbf{PHD} x\|_2\leq \epsilon \|x\|_2 \right]\leq \left({3}{\epsilon }\right)^k\]
\end{lemma}
Following \cite{Ailon09} we rewrite $\|\textbf{PHD}\|_2$ as $\sum_{i=1}^ky_i$.
Here $y_i\sim\mathcal{N}(0,f^-1)b_iu_i$ where $b_i$ is $1$ w.p $f$ and $0$ otherwise and $u=HDx$.
Define a random variable $Z_i=b_iu_i^2$ and we see that $y_i=\mathcal{N}(0,Z_i/f)$.
By the regular scaling of Gaussian with their standard deviation (See Lemma~\ref{lm:generalizednonstd}), it is clear that for an upper bound on:
\[\Pr[\sum_i^ky_i^2\leq t]\]  
we only need to lower bound the $Z_i$.
I.e.
\begin{lemma}
  \label{lm:non-to-std}
If $\forall i\in[k]$ $Z_i\geq f/2$ and $G$ is a full Gaussian matrix (entries sampled from $\mathcal{N}(0,1)$), then $\forall t\geq 0$:
\[\Pr[\|\textbf{PHD}x\|_2^2\leq t]\leq\Pr[\|Gx\|^2\leq 2t]\]

\begin{proof}
For $i\in\{1,\cdots,k\}$ assume $Z_i\geq \frac f2$ and let $X_i\sim\mathcal{N}(0,1)$, then:
\begin{align*}
  \Pr\left[\|\textbf{PHD}x\|_2^2\leq t\right]&=\Pr\left[\sum_{i=1}^ky_i^2\leq t\right]\\
              &\leq \Pr\left[\sum\left(\sqrt{\dfrac 1 2}X_i\right)^2\leq t\right]=\Pr\left[\sum X_i^2\leq 2t\right].
\end{align*}
Where the first equality follows the rewriting above and the inequality from the bound on the $Z_i$.
\end{proof}
\end{lemma}

In the $s$ smooth setting the most extreme concentration permitted still implies that $Z_i\sim \mathcal{B}(s^{-2},f)s^2$ (See \cite{Ailon09}). 
So $\Pr\left[ \forall i \in [k], Z_i\geq f/2 \right] \geq \frac{19}{20}$(Lemma 3 of~\cite{Ailon09}).
If we combine this bound with Lemma~\ref{lm:non-to-std} we are ready to prove Lemma~\ref{lm:fractionshrink}.

\begin{proof}
Let $z\in S$ and let $x=z\|z\|_2^{-1}$.
    \begin{align*}
      &\Pr\left[\|\textbf{PHD}z\|_2\leq \epsilon \|z\|_2\right]=\Pr\left[\|\textbf{PHD}x\|_2\leq {\epsilon }\right] =\\
      &\Pr\left[\|\textbf{PHD}x\|_2^2\leq {\epsilon ^2}\right] \leq \Pr\left[\sum_i^kX_i^2\leq 2 {\epsilon ^2}\right] \tag{by
        Lemma~\ref{lm:non-to-std} }
    \end{align*}
    Where $X_i\sim\mathcal{N}(0,1)$. In general for $s,t>0$ we know:

    \begin{align*}
      \Pr\left[\sum_i^kX_i^2\leq t\right]&=\Pr\left[e^{-s\sum X_i^2} \geq e^{-st}\right]\leq \frac {\E\left[e^{-s\sum X_i^2}\right]}{e^{-st}}\\
                                         &= e^{st}\prod_{i=1}^k\E\left[e^{-sX_i^2}\right]=e^{st}(1+2s)^{-k/2}
    \end{align*}

    Where the last step uses that $\E\left[e^{-sX_i^2}\right]=\frac 1 {\sqrt{1+2s}}$ for $-1/2 \leq s \leq \infty$.(See ~\cite{Dasgupta2003})

    Now to minimize we differentiate w.r.t s:

    \begin{align*}
      &t e^{st}(1+2s)^{-\frac{k}{2}}+2(-\frac k2)e^{st}(1+2s)^{-\frac k2 -1}=0 \Leftrightarrow t=k(1+2s)^{-1}\\
      &\Rightarrow s=(k/t-1)/2
    \end{align*}

    So $e^{st}(1+2s)^{-k/2} =   e^{(k-t)/2}(k/t)^{-k/2}= e^{-t/2}(\frac k {et})^{-k/2} \leq(et)^{k/2}$.
    Now plug in $t=2\epsilon^2$ and we have
   \[\Pr\left[\sum_i^kX_i^2\leq 2\epsilon^2\right]\leq (2e\epsilon^2)^{k/2}\leq(3\epsilon )^k\]
    
\end{proof}

\subsection{Embedding properties}

We have now seen how the Distortion and Shrinkage bounds follow from two events:

First $D$ must be in an $s$-smooth setting. Secondly all $Z_i$ must be within a constant factor of $f$.
By Lemma~\ref{lm:kinchine} the first event happens with probability at least $19/20$ when setting $s=\sqrt{7\frac{\lg(n^2d)}{d}}$, assuming $n^2d\ge 3.7$.
By choosing $f$ corresponding to $s$ as in\cite{Ailon09}, the second event occurs with probability at least $19/20$(See Lemma 3 of~\cite{Ailon09}).
For the chosen parameters $\Phi=\textbf{PHD}$ satisfies the Indyk-Naor properties with probability $\left(\frac{19}{20}\right)^2>9/10$.

We can then move on to prove Theorem~\ref{thm:fast-near-neighbor} by showing:

\begin{theorem}[Fast Nearest Neighbor Preserving Embeddings]
  \label{thm:fast-near-neighbor-by-FJLT}
  For any $S\subseteq \mathbb{R}^d, \epsilon,\delta \in (0,1)$ and some 
  \[k=\BO{\frac
      {\log{(2/\epsilon)}}{\epsilon^2}\log{(1/\delta)}\log{\lambda_S}}\enspace.\]
  Let $\Phi=\textbf{PHD}$ be a FJLT matrix with expected \[\BO{d\log(d)+\epsilon^{-2}\log^3(n)\log(2/\epsilon)}\] embedding time.
  For every $x\in S $ let $x'$ denotes the point closest  to $x$ in $S\setminus\{x\}$ under $\ell_2$.
  With probability at least $\delta$
    \begin{enumerate}
  \item $\min\limits_{z\in S\setminus\{x\}}\|\Phi x-\Phi z\|_2\leq(1+\epsilon)\|x-x'\|_2$, and 
  \item if $\|x-y\|_2> (1+2\epsilon)\|x-x'\|_2$ for some $y \in S$ then $\|\Phi x-\Phi y\|>(1+\epsilon)\|x-x'\|_2$.
  \end{enumerate}

  \begin{figure}
    \centering
 \includegraphics[width=0.3\paperwidth]{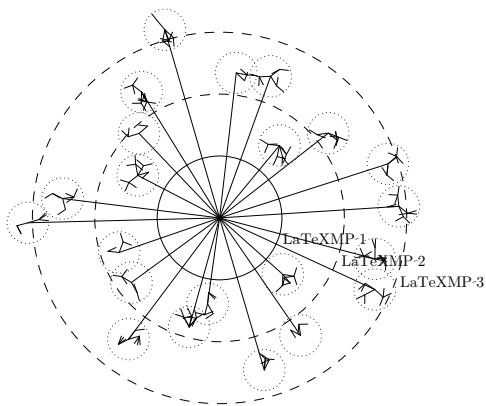}
  \caption{An illustration of the spanning tree construction used in the proof of Theorem~\ref{thm:fast-near-neighbor-by-FJLT}.}
    \label{fig:spantree}
\end{figure}

\begin{proof}
  Let $\Phi=\textbf{PHD}$ follow Definition~\ref{def:FJLT} with $f=\min\{\BOx{d^{-1}\ln(n^2d)},1\}$ so $\Phi$ satisfies the Indyk-Naor properties as pr. Definiton~\ref{def:IN-prop} with probability at least $9/10$.
  The proof then follows from~\cite[Theorem 4.1]{IN07}.
  For completeness we include an extended version of the proof here.
  For familiar readers, the only difference in this version is in making the spanning tree construction explicit.

  Without loss of generality let $x=0$ and $\|x'\|_2=1$.
  To show the first property let $y\in S$ satisfy $\|y\|_2=1$, then by the distortion bound, $\Pr[\|\Phi y\|\geq(1+\epsilon)]\leq e^{-\BOMx{k\epsilon^2}}$.
  So for some universal constant $C>0$, setting $k\geq C\ln(1/\delta)/\epsilon ^2$ we get:
  \[\Pr[\min\limits_{z\in S\setminus\{x\}}\|\Phi x-\Phi z\|_2>(1+\epsilon)\|x-x'\|_2]<\delta/2\]

  To show the second property we construct a spanning tree of $\left(S\setminus \B{x}{1+2\epsilon}\right)\cup \{0\}$ with $0$ at the root.
  Let $r_i=1+(i+2)\epsilon$. Consider the annuli:
  \[ A_i=S\cap\B{0}{r_{i+1}}\setminus\B{0}{r_{i}},\text{ for } i\geq1 \]
  By the definition of $\lambda_S$, for any $i$ we can construct a minimal set $S_i\subseteq S$ such that $A_i\subseteq \cup_{t\in S_i}\B{t}{\epsilon/4}$ and $|S_i|\leq \log_2(\frac {4r_i} \epsilon)$.
  The first level of the tree consists of an edge between $x$ and each $t\in S_i$ for all $i\geq0$.
  From each $t$ a spanning tree is build on the points in $\B{t}{\epsilon/4}$ with $t$ at the root, as described in lemma~\ref{lem:embedding-properties}.
  Figure~\ref{fig:spantree} illustrates the construction.
  Some ordering is imposed on the $t$ points so points in overlapping balls are only spanned once.

  We can then restate the second property as $\exists i\geq0,\exists x\in A_i . \|\Phi x\|\leq 1+\epsilon $, at least one of two events took place:
  \begin{enumerate}
  \item $\exists i\geq0,\exists t\in S_i. \|\Phi t\|_2\leq 1+\epsilon+\frac \epsilon 4 (1+\sqrt{i})$
  \item $\exists i\geq0,\exists t\in S_i, \exists x \in \B{t}{\frac \epsilon 4}\cap S. \Phi x \notin \B{\Phi t,(1+\sqrt{i})\frac \epsilon 4}$
  \end{enumerate}
  Since $\|t\|_2\geq r_i - \frac \epsilon 4$ there is some constant $C$ such that:
  \[
    \frac {\|\Phi t\|}{\|t\|_2} = \frac {1 + (1+\sqrt{i})\epsilon/4  + \epsilon } { 1+ (2+i)\epsilon - \epsilon /4} \leq
    \begin{cases}
      1-\epsilon/8 \text{ for } i\leq 1/\epsilon^2\\
      C/\sqrt{i} \text{ for } i > 1/\epsilon^2
  \end{cases}
\]

Fix some $i$. Using the distortion and shrinkage bounds:
\begin{align*}
  &\Pr\left[\exists t\in S_i, \|\Phi t\|_2\leq1+ (1+\sqrt{i})\epsilon/4 + \epsilon\right]\\
           &\leq\begin{cases}
             \lambda_x^{\log_2(4r_i/\epsilon)}e^{-ck\epsilon^2}\text{ for } i\leq 1/\epsilon ^2\\
             \lambda_x^{\log_2(4r_i/\epsilon)}(3C/\sqrt{i})^k\text{ for } i> 1/\epsilon ^2
           \end{cases}\\
             &\leq\begin{cases}
             e^{-c'k\epsilon ^2}\text{ for } i\leq 1/\epsilon ^2\\
             i^{-c'k} \text{ for } i> 1/\epsilon ^2
           \end{cases}\\
\end{align*}

For $k\geq \frac {c''} {\epsilon^2}\log(2/\epsilon)\log(\lambda_S)$ where $c''$ is some universal constant.
For the second event lemma~\ref{lem:embedding-properties} gives:
\[\Pr[\exists t, \exists y \in \B{t}{\frac \epsilon 4}\cap S, \Phi y \notin \B{\Phi t}{(1+\sqrt{i})\frac \epsilon 4}]\leq\lambda_x^{\log_2(4r_i/4)}e^{-ck(1+i)}\leq e^{-c'k(1+i)}\]
So there is some $c'''$ where the first event is most likely. Hence:
\[
\Pr[\exists x\in A_i. \|\Phi x\|_2\leq 1+ \epsilon ]\leq  
\begin{cases}
  2e^{-c'''k\epsilon ^2}\text{ for } i\leq 1/\epsilon ^2\\
  2i^{-c'''k} \text{ for } i> 1/\epsilon ^2
\end{cases}\\
\]

Summing over all the $i$ we get:
\begin{align*}
  \Pr[\exists i&\geq 0, \exists x \in A_i . \|\Phi S\|_2 \leq 1+\epsilon ] = \sum_i^{\infty} \Pr[\exists x\in A_i. \|\Phi x\|_2\leq 1+ \epsilon ]\\
         &\leq \frac 2 {\epsilon^2} e^{-c'''k\epsilon ^2} + \sum_{i>1/{\epsilon ^2}}  2i^{-c'''k}\leq \delta/2
\end{align*}

for some $k\geq \log(1/\delta)\frac {\tilde{c}} {\epsilon^2}\log(2/\epsilon)\log(\lambda_S)$ where $\tilde{c}$ is some large enough constant.
The number of operations required for embedding $x$ is $\BOx{d}$ for the diagonal matrix $D$, $\BOx{d\log{d}}$ for $H$ using the Walsh-Hadamard transform~\cite{Fino:1976:UMT:1311952.1312575} and finally $\BO{|\textbf{P}|}$ where $|\textbf{P}|$ is the number of non-zero entries.
$|\textbf{P}|\sim\mathcal{B}(kd,f)$ so by our setting of $f$:
\begin{align*}
  \E\left[|\textbf{P}|\right]&=kdf\\
                    &=\BOx{\epsilon^{-2}\log(\lambda_S)\log(2/\epsilon)\log^2(n)}\\
                    &=\BOx{\epsilon^{-2}\log^3(n)\log(2/\epsilon)}
\end{align*}
\end{proof}
\end{theorem}

\begin{lemma}
  \label{lem:embedding-properties}
  Let $S$ be a subset of the unit ball in $\mathbb{R}^d$, including $0$.
  Then there exists universal constants $c,C>0$ such that for $\epsilon>0$ and $k\geq C\log{\lambda_S}$:
  \[
    \Pr[\exists x\in S, \|\textbf{PHD}x\|_2\geq(1+\epsilon)]\leq e^{-ck(1+\epsilon)^2}.
  \]
  \begin{proof}
    The proof is given in ~\cite[Lemma 4.2]{IN07}.
    We include a spanning tree version here for completeness.
    We build a spanning tree $T$ on $S$ with root $0$ in the following way:
    Define sets for each possible level of the tree, $L_0,L_1,\ldots\subseteq S$.
    Let $L_0={0}$.
    To build $L_{j+1}$, for every point $t\in L_j$ let $S_{t}$ be the minimal size set such that $\cup_{s\in S_{t}}\B{t}{2^{-j-1}}\cap S$ covers all of $\B{t}{2^{-j}}\cap S$.
    By the definition of doubling constant we know that $|S_{t}|\leq \lambda_S$.
    Connect $t$ to every point in $S_t$, if some $S_t$ sets overlap only a single connection is made to avoid cycles.
    Let $L_{j+1} = \cup_{t\in L_j}S_t$. We observe that $0<|L_j|\leq \lambda_S^j$.

    Now let $E(T)$ denote the edges in $T$.
    Let $E_j$ be the subset of $E(T)$ with one node in $L_{j}$ and the other in $L_{j+1}$, by the construction of the tree $\forall e\in E_j$ we have $\|e\|_2\leq 2^{-j+1}$.
    For every $x\in S$ denote the unique path from $0$ to $x$ in $T$ by $p(x)\subseteq E(T)$.
    For $0\leq j\leq |p(x)|$ let $p_j(x)\in L_j$ be the vertex on the path at level $j$, for $j>|p(x)|$ let $p_j(x)=x$.
    We can then compose $x$ as $\sum_{j=0}^{\infty}\left(p_{j+1}(x)-p_{j}(x)\right)$, the first $|p(x)|$ steps corresponding to edges in $E(T)$, and the remaining steps having $0$ contribution.
    The argument then follows~\cite{IN07}:
    \begin{align*}
      \Pr[\exists x\in& S, \|\textbf{PHD}x\|_2\geq(1+\epsilon)]\\
              &\leq\Pr\left[\exists x \in S, \exists j\geq 0, \|\textbf{PHD}(p_{j+1}(x)-p_{j}(x))\|_2\geq\frac {(1+\epsilon)} 3 \left(\frac 3 2 \right)^{-j}\right] \\
              &=\sum_{j=0}^\infty \Pr\left[\exists e\in E_j, \|\textbf{PHD}e\|_2\geq \frac {1+\epsilon} 3 \left(\frac 3 2 \right)^{-j}\right]\\
              &\leq\sum_{j=0}^\infty \Pr\left[\exists e\in E_j, \|\textbf{PHD}e\|_2\geq \frac {1+\epsilon}6 \left( \frac 43 \right)^j\|e\|_2\right]\\
              &\leq\sum_{j=0}^\infty \lambda_S^{2j}\Pr\left[\|\textbf{PHD}x\|_2\geq1+\frac {1+\epsilon}6 \left( \frac 43 \right)^j-1\right]\text{,for any unit vector $x$}\\
              &\leq\sum_{j=0}^\infty \lambda_S^{2j}e^{-ck(1+\epsilon)^2(4/3)^{2j}/100}\leq e^{-ck(1+\epsilon)^2}
    \end{align*}
    For $k\geq C\log\lambda_S+1$.
    Crucially the second last step uses that$|E(T)|=|S|-1$.
    We can then use Lemma~\ref{lm:kinchine} to see that $D$ is in a smooth setting with constant probability, for our setting of $s$ at least $\frac {19}{20}$.
    The last step then follows from Lemma~\ref{lm:distortion}.
  \end{proof}
  
\end{lemma}

\section{Conclusion}

In this chapter we present embeddings that combine the low-dimensional embedding space achieved by Nearest Neighbor Preserving Embeddings~\cite{IN07} with a speedup of the embedding runtime achieved by a Fast-JL construction~\cite{Ailon09}.
This results in embeddings that are faster than fully Gaussian Nearest Neighbor Preserving Embeddings and use fewer dimensions than any Johnson-Lindenstrauss type embedding.

The benefit of Nearest Neighbor Preserving Embeddings generally depends on the difference between $n=|S|$ and $\lambda_S$.
While $\lambda_S$ is always upper bounded by $n$ it can often be much smaller, this helps to explain why some datasets can be successfully embedded into much fewer dimensions, and much faster, than theoretical results looking only on $|S|$ can explain.
For datasets with low doubling dimension we can expect to find fast embeddings into a low number of dimensions, even if the dataset is very large.

While the number of rows in the embedding matrix is independent of $n$, the sparsity of the matrix is not.
This happens because we must ensure that all $\BOx{n^2}$ possible edges in the constructed spanning trees used in lemma~\ref{lem:embedding-properties} are smooth.
Future work could focus on alternative constructions to increase the sparsity.


\chapter{Set Similarity Join}
\label{cha:set-similarity}
Set similarity join is a fundamental and well-studied database operator.
It is usually studied in the \emph{exact} setting where the goal is to compute all pairs of sets that exceed a given level of similarity (measured e.g.~as Jaccard similarity).
But set similarity join is often used in settings where 100\% recall may not be important --- indeed, where the exact set similarity join is itself only an approximation of the desired result set.

We present a new randomized algorithm for set similarity join that can achieve any desired recall up to 100\%, and show theoretically and empirically that it significantly outperforms state-of-the-art implementations of exact methods, and improves on existing approximate methods.
Our experiments on benchmark data sets show the method is several times faster than comparable approximate methods, at 90\% recall the algorithm is often more than $2$ orders of magnitude faster than exact methods.
Our algorithm makes use of recent theoretical advances in high-dimensional sketching and indexing that we believe to be of wider relevance to the database community.

\section{Introduction}

It is increasingly important for data processing and analysis systems to be able to work with data that is imprecise, incomplete, or noisy.
\emph{Similarity join} has emerged as a fundamental primitive in data cleaning and entity resolution over the last decade~\cite{augsten2013similarity,Chaudhuri_ICDE06, sarawagi2004efficient}.
In this chapter we focus on \emph{set similarity join}:
Given collections $R$ and $S$ of sets the task is to compute 
$$ R \simjoin S = \{(x,y)\in R\times S \; | \; \simil(x,y)\geq \lambda\}$$
where $\simil(\cdot, \cdot)$ is a similarity measure and $\lambda$ is a threshold parameter.
We deal with sets $x,y \subseteq \{1,\dots,d\}$, where the number $d$ of distinct tokens can be naturally thought of as the dimensionality of the data.

Many measures of set similarity exist~\cite{choi2010survey}, but perhaps the most well-known such measure is the \emph{Jaccard similarity},
$$J(x,y) = |x\cap y|/|x\cup y| \enspace .$$
For example, the sets $x = \{$\texttt{IT, University, Copenhagen}$\}$ and $y = \{$\texttt{University, Copenhagen, Denmark}$\}$  have Jaccard similarity $J(x,y) = 1/2$ which could suggest that they both correspond to the same entity.
In the context of entity resolution we want to find a set~$T$ that contains $(x,y)\in R\times S$ if and only if $x$ and $y$ correspond to the same entity.
The quality of the result can be measured in terms of \emph{precision} $|(R \simjoin S) \cap T|/|T|$ and \emph{recall} $|(R \simjoin S) \cap T|/|R \simjoin S|$ (both of which should be as high as possible).
We will be interested in methods that achieve 100\% precision, but that might not have 100\% recall.
We sometimes referring to methods with 100\% recall as exact, and others as approximate.
Note that this is in view of the output size, not the similarity as in our other approximate similarity problems.
Considering similarity join methods that are not exact allow for new randomized algorithmic techniques. 
It has been known from a theoretical point of view that this can lead to algorithms that are more scalable and robust (against hard inputs), compared to exact set similarity join methods for high-dimensional data.
However, these methods have not seen widespread use in practical join algorithms, arguably because they have not been sufficiently mature, e.g.~having large overheads that make asymptotic gains disappear and being unable to take advantage of features of real-life data sets that make similarity join computation easier.

\medskip

\textbf{Our contributions.}
We present the Chosen Path Set Similarity Join (\textsc{CPSJoin}) algorithm, its theoretical underpinnings, 
and show experimentally that it achieves substantial speedup in practice compared to state-of-the-art exact techniques by allowing less than 100\% recall.
The two key ideas behind \cpsj are:
\begin{itemize}
\item A new recursive filtering technique inspired by the recently proposed {\textsc ChosenPath} index for set similarity search~\cite{christiani2017set}, adding new ideas to make the method parameter-free, near-linear space, and adaptive to a given data set.
\item Apply efficient sketches for estimating set similarity~\cite{li2011theory} that take advantage of modern hardware. 
\end{itemize}

We compare \textsc{CPSJoin} to the exact set similarity join algorithms in the comprehensive empirical evaluation of Mann et al.~\cite{Mann2016}, using the same data sets, and to other approximate set similarity join methods suggested in the literature.
The probabilistic approach scales much better on input instances where prefix filtering does not cut down the search space significantly.
We see speedups of more than 1 order of magnitude at 90\% recall, especially for set similarity joins where the sets are relatively large (100 tokens or more) and the similarity threshold is low (e.g.~Jaccard similarity~0.5).


\subsection{Related work}\label{sec:related}

\textbf{Exact similarity join.}
For space reasons we present just a sample of the most related previous work, and refer to the book of Augsten and B{\"o}hlen~\cite{augsten2013similarity} for a survey of algorithms for exact similarity join in relational databases, covering set similarity joins as well as joins based on string similarity.

Early work on similarity join focused on the important special case of detecting \emph{near-duplicates} with similarity close to~1, see e.g.~\cite{broder2000identifying,sarawagi2004efficient}.
A sequence of results starting with the seminal paper of Bayardo et al.~\cite{Bayardo_WWW07} studied the range of thresholds that could be handled.
Recently, Mann et al.~\cite{Mann2016} conducted a comprehensive study of 7 state-of-the-art algorithms for exact set similarity join for Jaccard similarity threshold $\lambda \in \{0.5,0.6,0.7,0.8,0.9\}$.
These algorithms all use the idea of \emph{prefix filtering}~\cite{Bayardo_WWW07}, which generates a sequence of candidate pairs of sets that includes all pairs of similarity above the threshold.
The methods differ in how much additional filtering is carried out.
For example,~\cite{xiao2011efficient} applies additional \emph{length} and \emph{suffix} filters to prune the candidate pairs.
The main finding by Mann et al.~is that while advanced filtering techniques do yield speedups on some data sets, an optimized version of the basic prefix filtering method (referred to as ``ALL'') is always competitive, and often the fastest of the algorithms.
For this reason we will be comparing our results against ALL.


\medskip

\textbf{Locality-sensitive hashing.}
Locality-sensitive hashing (LSH) is a theoretically well-founded randomized method for creating candidate pairs~\cite{Gionis99}.
Though some LSH methods guaranteeing 100\% recall exist~\cite{arasu2006efficient,Pagh2016}, LSH is usually associated with having less than 100\% recall probability for each output pair.
We know only of a few papers using LSH techniques to solve similarity join.
Cohen et al.~\cite{cohen2001finding} used LSH techniques for set similarity join in a knowledge discovery context before the advent of prefix filtering.
They sketch a way of choosing parameters suitable for a given data set, but we are not aware of existing implementations of this approach.
Chakrabarti et al.~\cite{chakrabarti2015bayesian} improved plain LSH with an adaptive similarity estimation technique, \emph{BayesLSH}, that reduces the cost of checking candidate pairs and typically improves upon an implementation of the basic prefix filtering method by $2$--$20\times$.
Our experiments include comparison to both methods~\cite{chakrabarti2015bayesian,cohen2001finding}.


We refer to the recent survey paper~\cite{pagh2015large} for an overview of theoretical developments, but point out that these developments have not matured sufficiently to yield practical improvements to similarity join methods.

\medskip

\textbf{Locality-sensitive mappings.}
Several recent theoretical advances in high-dimensional indexing~\cite{andoni2017optimal,christiani2017framework,christiani2017set} have used an approach that can be seen as a generalization of LSH.
We refer to this approach as locality-sensitive \emph{mappings} (also known as locality-sensitive \emph{filters} in certain settings).
The idea is to construct a function $F$, mapping a set $x$ into a set of machine words, such that:
\begin{itemize}
\item If $\simil(x,y)\geq \lambda$ then $F(x)\cap F(y)$ is nonempty with some fixed probability~$\varphi > 0$.
\item If $\simil(x,y) < \lambda$, then the expected intersection size $\E[|F(x)\cap F(y)|]$ is ``small''.
\end{itemize}
Here the exact meaning of ``small'' depends on the difference~\mbox{$\lambda - \simil(x,y)$}, 
but in a nutshell, if it is the case that almost all pairs have similarity significantly below $\lambda$ then we can expect \mbox{$|F(x)\cap F(y)| = 0$} for almost all pairs.
Performing the similarity join amounts to identifying all candidate pairs $x,y$ for which $F(x)\cap F(y) \ne \varnothing$ (for example by creating an inverted index), 
and computing the similarity of each candidate pair.
To our knowledge these indexing methods have not been tried out in practice, probably because they are rather complicated.
An exception is the recent paper~\cite{christiani2017set}, which is relatively simple, and indeed our join algorithm is inspired by the index described in that paper.

\medskip

\textbf{Distance estimation.}
Similar to BayesLSH~\cite{chakrabarti2015bayesian} we make use of algorithms for similarity \emph{estimation}, but in contrast to BayesLSH we use algorithms that make use of bit-level parallelism.
This approach works when there exists a way of picking a random hash function $h$ such that
\begin{equation}\label{eq:lsh-able}
\Pr[h(x)=h(y)] = \simil(x,y)
\end{equation}
for every choice of sets $x$ and $y$.
Broder et al.~\cite{Broder_NETWORK97} presented such a hash function for Jaccard similarity, now known as \mh or ``minwise hashing'', as discussed in Section~\ref{sec:hashing}.
In the context of distance estimation, 1-bit minwise hashing of Li and K{\"o}nig~\cite{li2011theory} maps $t$ \mh values to a compact sketch, using just $t$ bits.
Still, this is sufficient information to be able to estimate the Jaccard similarity of two sets $x$ and $y$ just based on the Hamming distance of their sketches. 
(In fact, the approach of~\cite{li2011theory} is known to be close to optimal~\cite{pagh2014min}.)
Like in~\cite{chakrabarti2015bayesian} we will use distance estimation to perform an additional filtering of the set of candidate pairs, avoiding expensive exact similarity computations for candidate pairs of low similarity.

\section{Preliminaries}

The \textsc{CPSJoin} algorithm solves the set similarity join problem with a probabilistic guarantee on recall, formalized in Definition~\ref{def:simjoin}.
It returns a set $L\subseteq S \bowtie_{\lambda} R$ in a way that for every $(x, y) \in S \bowtie_{\lambda} R$ we are guaranteed $\Pr[(x, y) \in L] \geq \varphi$. 
It is important to note that the probability is over the random choices made by the algorithm, and \emph{not} over a random choice of $(x,y)$.
This means that the probability $(x, y) \in S \bowtie_{\lambda} R$ is \emph{not} reported in $i$ independent repetitions of the algorithm is bounded by $(1-\varphi)^i$.
A recall probability of, $\varphi = 0.9$ can be boosted to recall probability close to~1, e.g.~$99.9\%$ using $t=3$ repetitions.
Finally, note that recall probability $\varphi$ implies that we expect recall \emph{at least} $\varphi$, but the actual recall may be higher.

\subsection{Similarity measures}\label{sec:reduction}
Our algorithm can be used with a broad range of similarity measures through randomized \emph{embeddings}.
This allows our algorithms to be used with, for example, Jaccard and cosine similarity thresholds.

Embeddings map data from one space to another while approximately preserving distance information, with accuracy that can be tuned.
In our case we are interested in embeddings that map data to sets of tokens.
We can transform any so-called \emph{LSHable} similarity measure $\simil$, where we can choose $h$ to make (\ref{eq:lsh-able}) hold, into a set similarity measure by the following randomized embedding:
For a parameter $t$ pick hash functions $h_1,\dots,h_t$ independently from a family satisfying~(\ref{eq:lsh-able}).
The embedding of $x$ is the following set of size~$t$:
$$ f(x) = \{ (i,h_i(x)) \; | \; i=1,\dots,t \} \enspace .$$
It follows from~(\ref{eq:lsh-able}) that the expected size of the intersection $f(x)\cap f(y)$ is $t \cdot \simil(x,y)$.
We can use a Chernoff bound to bound the number of functions necessary.
\[\Pr\left[\left|\frac{|f(x)\cap f(y)|}t-\simil(x,y)\right|\geq \sqrt{\frac{6\ln{t}}{t}}\simil(x,y)\right]\leq2t^{-\simil(x,y)}\]
(See e.g. Equation~\ref{eq:minest}).
For our experiments with Jaccard similarity thresholds $\geq0.5$,  we found that $t=64$ gave sufficient precision for $>90\%$ recall.

In summary we can perform the similarity join $R\simjoin S$ for any LSHable similarity measure by creating two corresponding relations $R'=\{f(x) \; | \; x\in R\}$ and $S'=\{f(y) \; | \; y\in S\}$, 
and computing $R'\simjoin S'$ with respect to the similarity measure
\begin{equation}\label{eq:bb}
		BB(f(x),f(y)) = |f(x)\cap f(y)|/t \enspace .
\end{equation}
This measure is the special case of \emph{Braun-Blanquet} similarity where the sets are known to have size~$t$.
Our implementation will take advantage of the set size $t$ being fixed, though it is easy to extend to general Braun-Blanquet similarity.

The class of LSHable similarity measures is large, as discussed in~\cite{chierichetti2015lsh}.
It includes the Jaccard similarity, cosine similarity and other commonly used similarity measures.
If approximation errors are tolerable, even \emph{edit distance} can be embedded into Hamming space and handled by our algorithm~\cite{chakraborty2016streaming,zhang2017embedjoin}.

\subsection{Notation}

We are interested in sets $S$ where an element, $x\in S$ is a set with elements from some universe $[d]=\{1,2,3,\cdots,d\}$.
To avoid confusion we sometimes use ``record'' for $x\in S$ and ``token'' for the elements of $x$.
Throughout this chapter we will think of a record $x$ both as a set of tokens from $[d]$, as well as a vector from $\{0,1\}^d$, where:
\[
  x_i=
  \begin{cases}
    1\text { if } i\in x\\
    0\text { if } i\notin x
  \end{cases}
\]
It is clear that 
these representations are equivalent.
The set $\{1,4,5\}$ is equivalent to $(1,0,0,1,1,0,\cdots,0)$, $\{1,d\}$ is equivalent to $(1,0,\cdots,0,1)$, etc.

\section{Overview of approach}\label{sec:overview}
Our high-level approach is recursive 
and works as follows.
To compute $R\simjoin S$ we consider each $x\in R$ and either:
\begin{enumerate}
\item Compare $x$ to each record in $S$ (referred to as ``brute forcing'' $x$), or
\item create several subproblems $S_i \simjoin R_i$ with $x\in R_i \subseteq R$, $S_i \subseteq S$, and solve them recursively.
\end{enumerate}
The approach of~\cite{christiani2017set} corresponds to choosing option 2 until reaching a certain level $k$ of the recursion, where we finish the recursion by choosing option~1.
This makes sense for certain worst-case data sets, but we propose an improved parameter-free method that is better at adapting to the given data distribution.
In our method the decision on which option to choose depends on the size of $S$ and the average similarity of $x$ to the records of $S$.
We choose option~1 if $S$ has size below some (constant) threshold, or if the average Braun-Blanquet similarity of $x$ and $S$, $\tfrac{1}{|S|}\sum_{y\in S} BB(x,y)$, is close to the threshold~$\lambda$.
In the former case it is cheap to finish the recursion.
In the latter case many records $y\in S$ will have $BB(x,y)$ larger than or close to $\lambda$, so we do not expect to be able to produce output pairs with $x$ in less than linear time in $|S|$.

If none of the pruning conditions apply we choose option~2 and include $x$ in recursive sub problems as described below.
But first we note that the decision of which option to use can be made efficiently for each $x$, since the average Braun-Blanquet similarity of pairs from $R\times S$ can be computed from token frequencies in time $\BOx{|R|+|S|}$.

\begin{itemize}
\item \textbf{Comparing $x$ to each record of $S$}. 
We speed up the computation by using distance estimation (in our case using 1-bit minwise hashing) to efficiently avoid exact computation of similarities $BB(x,y)$ for $y\in S$ where $B(x,y)$ is significantly below $\lambda$.
\item \textbf{Recursion}. 
We would like to ensure that for each pair $(x,y)\in R\simjoin S$ the pair is computed in one of the recursive subproblems, i.e., that $(x,y)\in R_i\simjoin S_i$ for some $i$.
In particular, we want the expected number of subproblems containing $(x,y)$ to be at least~1, i.e.,
\begin{equation}\label{eq:expect}
\E[|\{i \;|\; (x,y)\in R_i\simjoin S_i\}|] \geq 1.
\end{equation}
Let $R'$ and $S'$ be the subsets of $R$ and $S$ that do not satisfy any of the pruning conditions.
To achieve (\ref{eq:expect}) for each pair $(x,y)\in R\simjoin S$ we recurse with probability $1/(\lambda t)$, where $t$ is the size of records in $R$ and $S$, on the subproblem $R_i\simjoin S_i$ with sets
\begin{align*}
R_i &= \{ x\in R' \; | \; i\in x \}\\
S_i &= \{ y\in S' \; | \; i\in y \}
\end{align*}
for each $i\in \{1,\dots,d\}$. 
It is not hard to check that (\ref{eq:expect}) is satisfied for every pair $(x,y)$ with $BB(x,y)\geq\lambda$.
Of course, expecting one subproblem to contain $(x,y)$ does not \emph{directly} imply a good probability that $(x,y)$ is contained in at least one subproblem.
But it turns out that we can use results from the theory of branching processes to show such a bound;
details are provided in section~\ref{sec:algorithm}.
\end{itemize}
%

\section{Chosen Path Set Similarity Join} \label{sec:algorithm}
The \textsc{CPSJoin} algorithm solves the $(\lambda,\varphi)$-set similarity join problem~(Definition~\ref{def:simjoin}).
To simplify the exposition we focus on a self-join version where given $S$ we wish to report $L \subseteq S \simjoin S$.
Handling a general join $S \simjoin R$ follows the overview in section~\ref{sec:overview} and requires no new ideas: Essentially consider a self-join on $S\cup R$ but make sure to consider only pairs in $S\times R$ for output.
We also make the simplifying assumption that all sets in $S$ have a fixed size $t$ --- as argued in section~\ref{sec:reduction} the general case can be reduced to this one by embedding.

The \textsc{CPSJoin} algorithm solves the $(\lambda,\varphi)$-set similarity join for every choice of $\lambda \in (0,1)$ 
and with a guarantee on $\varphi$ that we will lower bound in the analysis. 
We provide theoretical guarantees on the expected running time of \textsc{CPSJoin} as well as experimental results showing large speedups compared to existing state-of-the-art exact and approximate similarity join techniques. 
In the experiments a single run of our algorithm typically only reports around one third of the similar points compared to the exact algorithms, 
but through independent repetitions we are able to obtain speedups in the range of $2-50\times$ for many real data sets and parameter settings while keeping the recall above $90\%$.
\subsection{Description}
The \text{CPSJoin} algorithm (see Algorithm \ref{alg:cpsjoin} for pseudocode) works by recursively splitting the data set on elements of~$[d]$ that are selected according to a random process, 
forming a recursion tree with $S$ at the root and subsets of $S$ that are non-increasing in size as we get further down the tree.
The randomized splitting has the property that the probability of a pair of points $(x,y)$ being in a given node is increasing as a function of $|x\cap y|$.

Before each splitting step we run the recursive \textsc{BruteForce} subprocedure (see Algorithm \ref{alg:bruteforce} for pseudocode) that identifies subproblems that are best solved by brute force.
It has two parts:

1. If $S$ is below some constant size, controlled by the parameter \texttt{limit}, we report $S\bowtie_{\lambda}S$ exactly using a simple loop with $O(|S|^2)$ distance computations (\textsc{BruteForcePairs}) and exit the recursion.
In our experiments we have set \texttt{limit} to $250$, with the precise choice seemingly not having a large effect as shown experimentally in Section~\ref{sec:parameters}. 

2. If $S$ is larger than \texttt{limit} the second part activates:
for every $x \in S$ we check whether the expected number of comparisons that $x$ is a part of is going to decrease after performing the splitting.
If this is not the case, we immediately compare $x$ against every point in $S$ (\textsc{BruteForcePoint}), reporting close pairs, and proceed by removing $x$ from $S$.
The \textsc{BruteForce} procedure is then run again on the reduced set.
The recursion exits if every point $x\in S$ has a decreasing number of expected comparisons. 

This recursive procedure where we choose to handle some points by brute force crucially separates our algorithm from many other approximate similarity join methods in the literature that typically are LSH-based~\cite{PaghSIMJOIN2015, cohen2001finding}.
By efficiently being able to remove points at the ``right'' time, before they generate too many expensive comparisons further down the tree,
we are able to beat the performance of other approximate similarity join techniques in both theory and practice.
Another benefit of this rule is that it reduces the number of parameters compared to the usual LSH setting where the depth of the tree has to be selected by the user.

\begin{algorithm}
\DontPrintSemicolon
\emph{For $j \in [d]$ initialize $S_j \leftarrow \varnothing$.}\;
$S \leftarrow \textsc{BruteForce}(S, \lambda)$\;
$r \leftarrow \textsc{SeedHashFunction}()$\; 
\For{$x \in S$} {
	\For{$j \in x$} {

		\lIf{$r(j) < \frac{1}{\lambda |x|}$}{$S_{j} \leftarrow S_{j} \cup \{ x \} $}
	}
}
\lFor{$S_{j} \neq \varnothing$} {$\textsc{CPSJoin}(S_{j}, \lambda)$}
\caption{\textsc{CPSJoin}$(S, \lambda)$} \label{alg:cpsjoin}
\end{algorithm}

\begin{algorithm}
\SetKwData{Limit}{limit}
\SetKwArray{Count}{count}
\SetKwInOut{Global}{Global parameters}
\DontPrintSemicolon
\Global{\Limit $\geq 1$, $\varepsilon \geq 0$.}
\emph{Initialize empty map \Count{\,} with default value $0$.}\;
\If{$|S| \leq$ \Limit} {
	$\textsc{BruteForcePairs}(S, \lambda)$\;
	\Return $\varnothing$\;
}
\For{$x \in S$} {
	\For{$j \in x$} {
		\Count{j} $\leftarrow$ \Count{j} + 1\;
	}
}
\For{$x \in S$} {
	\If{$\frac{1}{|S| - 1}\sum_{j \in x}($\Count{j}$ - 1)/t > (1-\varepsilon)\lambda$} {
		$\textsc{BruteForcePoint}(S, x, \lambda)$\;
		\Return $\textsc{BruteForce}(S \setminus \{ x \}, \lambda)$\;
	} 
}
\Return $S$\;
	
\caption{\textsc{BruteForce}$(S, \lambda)$} \label{alg:bruteforce}
\end{algorithm}
\subsection{Comparison to Chosen Path}
\label{sec:comp-chos-path}
The \textsc{CPSJoin} algorithm is inspired by the \textsc{Chosen Path} algorithm~\cite{christiani2017set} for the approximate near neighbor problem 
and uses the same underlying random splitting tree that we will refer to as the Chosen Path Tree.
In the approximate near neighbor problem, the task is to construct a data structure that takes a query point and correctly reports an approximate near neighbor, if such a point exists in the data set.
Using the \textsc{Chosen Path} data structure directly to solve the $(\lambda,\varphi)$-set similarity join problem has several drawbacks that we avoid in the \textsc{CPSJoin} algorithm.
First, the \textsc{Chosen Path} data structure is parameterized in a non-adaptive way to provide guarantees for worst-case data, 
vastly increasing the amount of work done compared to the optimal parameterization when data is not worst-case.
Our recursion rule avoids this and instead continuously adapts to the distribution of distances as we traverse down the tree. 
Second, the data structure uses space $\BOx{n^{1+\rho}}$ where $\rho > 0$, storing the Chosen Path Tree of size $\BOx{n^\rho}$ for every data point. 
The \textsc{CPSJoin} algorithm, instead of storing the whole tree, essentially performs a depth-first traversal, allowing us to bound the space usage by $O(n + m)$ where $m$ is the output size.
Finally, the \textsc{Chosen Path} data structure only has to report a single point that is approximately similar to a query point, and can report points with similarity $< \lambda$.
To solve the approximate similarity join problem the \textsc{CPSJoin} algorithm has to satisfy reporting guarantees for \emph{every} pair of points $(x, y)$ in the exact join.
\subsection{Analysis} \label{sec:analysis}
The Chosen Path Tree for a data point $x \subseteq [d]$ is defined by a random process: 
at each node, starting from the root, we sample a random hash function $r \colon [d] \to [0,1]$ and construct children for every element $j \in x$ such that $r(j) < \frac{1}{\lambda |x|}$.
Nodes at depth $k$ in the tree are identified by their path $p = (j_1, \dots, j_k)$. 
Formally, the set of nodes at depth $k > 0$ in the Chosen Path Tree for $x$ is given by
\begin{equation}
	F_{k}(x) = \left\{ p \circ j \mid p \in F_{k-1}(x) \land r_{p}(j) < \frac{x_j}{\lambda |x|} \right\}
\end{equation}
where $p \circ j$ denotes vector concatenation and $F_{0}(x) = \{()\}$ is the set containing only the empty vector.
The subset of the data set $S$ that survives to a node with path $p = (j_1, \dots, j_{k})$ is given by
\begin{equation}
	S_{p} = \{ x \in S \mid x_{j_1} = 1 \land \dots \land x_{j_{k}} = 1 \}.
\end{equation}
The random process underlying the Chosen Path Tree belongs to the well studied class of Galton-Watson branching processes.
Originally these where devised to answer questions about the growth and decline of family names in a model of population growth assuming i.i.d.\ offspring for every member of the population across generations~\cite{watson1875}.
In order to make statements about the properties of the \textsc{CPSJoin} algorithm we study in turn the branching processes 
of the Chosen Path Tree associated with a point $x$, a pair of points $(x, y)$, and a set of points $S$.
Note that we use the same random hash functions for different points in $S$.

\paragraph{Brute forcing.}
The \textsc{BruteForce} subprocedure described by Algorithm \ref{alg:bruteforce} takes two global parameters: $\mathtt{limit} \geq 1$ and $\varepsilon \geq 0$.
The parameter $\mathtt{limit}$ controls the minimum size of $S$ before we discard the \cpsj algorithm for a simple exact similarity join by brute force pairwise distance computations.
The second parameter, $\varepsilon > 0$, controls the sensitivity of the \textsc{BruteForce} step to the expected number of comparisons that a point $x \in S$ will generate if allowed to continue in the branching process.
The larger $\varepsilon$ the more aggressively we will resort to the brute force procedure.
In practice we typically think of $\varepsilon$ as a small constant, say $\varepsilon = 0.05$, 
but for some of our theoretical results we will need a sub-constant setting of $\varepsilon \approx 1/\log(n)$ to show certain running time guarantees. 
The \textsc{BruteForce} step removes a point $x$ from the Chosen Path branching process, 
instead opting to compare it against every other point $y \in S$, if it satisfies the condition
\begin{equation}
	\frac{1}{|S| - 1}\sum_{y \in S \setminus \{ x \}}|x \cap y|/t > (1-\varepsilon)\lambda. \label{eq:bruteforce}
\end{equation}
In the pseudocode of Algorithm \ref{alg:bruteforce} we let \texttt{count} denote a hash table that keeps track of the number of times each element $j \in [d]$ appears in $S$.
This allows us to evaluate the condition in equation \eqref{eq:bruteforce} for an element $x \in S$ in time $O(|x|)$ by rewriting it as
\begin{equation}
\frac{1}{|S|-1}\sum_{j \in x} (\mathtt{count}[j] - 1)/t > (1-\varepsilon)\lambda.
\end{equation}
We claim that this condition minimizes the expected number of comparisons performed by the algorithm:
Consider a node in the Chosen Path Tree associated with a set of points $S$ while running the \cpsj algorithm.
For a point ${x\in S}$, we can either remove it from $S$ immediately at a cost of $|S|-1$ comparisons, or we can choose to let continue in the branching process (possibly into several nodes) and remove it later.
The expected number of comparisons if we let it continue $k$ levels before removing it from every node that it is contained in, is given by
\begin{equation}
	\sum_{y \in S \setminus \{ x \}} \left(\frac{1}{\lambda}\frac{|x \cap y|}{t}\right)^{k}.
\end{equation}
This expression is convex and increasing in the similarity $|x \cap y|/t$ between $x$ and other points $y \in S$, allowing us to state the following observation: 
\begin{observation}[Recursion]\label{obs:bruteforce}
	Let $\varepsilon = 0$ and consider a set $S$ containing a point $x \in S$ such that $x$ satisfies the recursion condition in equation~\eqref{eq:bruteforce}.
	Then the expected number of comparisons involving $x$ if we continue branching exceeds $|S|-1$ at every depth $k \geq 1$.
        If $x$ does not satisfy the condition, the opposite is observed.
\end{observation}
\paragraph{Tree depth.}
We proceed by bounding the maximal depth of the set of paths in the Chosen Path Tree that are explored by the \cpsj algorithm.
Having this information will allow us to bound the space usage of the algorithm and will also form part of the argument for the correctness guarantee.
Assume that the parameter \texttt{limit} in the \textsc{BruteForce} step is set to some constant value, say $\mathtt{limit} = 10$.
Consider a point $x \in S$ and let $S' = \{ y \in S \mid |x \cap y|/ t \leq (1-\varepsilon)\lambda \}$ be the subset of points in $S$ that are not too similar to $x$.
For every $y \in S'$ the expected number of vertices in the Chosen Path Tree at depth $k$ that contain both $x$ and $y$ is upper bounded by
\begin{equation}
	\E[|F_{k}(x \cap y)|] = \left(\frac{1}{\lambda}\frac{|x \cap y|}{t}\right)^{k} \leq (1-\varepsilon)^{k} \leq e^{-\varepsilon k}.
\end{equation}
Since $|S'| \leq n$ we use Markov's inequality to show the following bound: 
\begin{lemma}
	Let $x, y \in S$ satisfy that $|x \cap y|/ t \leq (1-\varepsilon)\lambda$ then the probability that there exists a vertex at depth $k$ 
	in the Chosen Path Tree that contains $x$ and $y$ is at most $e^{-\varepsilon k}$.
\end{lemma}
If $x$ does not share any paths with points that have similarity that falls below the threshold for brute forcing, then the only points that remain are ones that will cause $x$ to be brute forced.
This observation leads to the following probabilistic bound on the tree depth:
\begin{lemma}\label{lem:depth}
	With high probability the maximal depth of paths explored by the \cpsj algorithm is $O(\log(n) / \varepsilon)$.
\end{lemma}

\paragraph{Correctness.}
Let $x$ and $y$ be two sets of equal size $t$ such that $BB(x, y) = |x \cap y|/t \geq \lambda$. 
We are interested in lower bounding the probability that there exists a path of length~$k$ in the Chosen Path Tree that has been chosen by both $x$ and~$y$, 
i.e. $\Pr\left[F_{k}(x \cap y)\neq\varnothing\right]$.
Agresti~\cite{agresti1974} showed an upper bound on the probability that a branching process becomes extinct after at most $k$ steps. 
We use it to show the following lower bound on the probability of a close pair of points colliding at depth $k$ in the Chosen Path Tree.
\begin{lemma}[{Agresti~\cite{agresti1974}}]
	If $\simil(x, y) \geq \lambda$ then for every $k > 0$ we have that \mbox{$\Pr[F_{k}(x \cap y) \neq \varnothing] \geq \frac{1}{k+1}$}.
\end{lemma}
The bound on the depth of the Chosen Path Tree for $x$ explored by the \cpsj algorithm in Lemma~\ref{lem:depth} then implies a lower bound on $\varphi$.


\begin{lemma}\label{lem:correctness}
	Let $0 < \lambda < 1$ be constant. Then for every set $S$ of $|S| = n$ points the \cpsj algorithm solves the set similarity join problem with $\varphi = \Omega(\varepsilon / \log(n))$.  
\end{lemma}

\begin{remark}
  This analysis is very conservative: if either $x$ or $y$ is removed by the \textsc{BruteForce} step prior to reaching the maximum depth then it only increases the probability of collision.
  We note that similar guarantees can be obtained when using fast pseudorandom hash functions as shown in the paper introducing the \cp algorithm~\cite{christiani2017set}.
\end{remark}

\textbf{Space usage.}
We can obtain a trivial bound on the space usage of the \cpsj algorithm by combining Lemma \ref{lem:depth} with the observation that every call to $\cpsj$ on the stack uses additional space at most $O(n)$. 
The result is stated in terms of working space: the total space usage when not accounting for the space required to store the data set itself 
(our algorithms use references to data points and only reads the data when performing comparisons) as well as disregarding the space used to write down the list of results. 
\begin{lemma} \label{lem:space}
	With high probability the working space of the \cpsj algorithm is at most $O(n \log (n) /\varepsilon)$.
\end{lemma}
\begin{remark}
We conjecture that the expected working space is $O(n)$ due to the size of $S$ being geometrically decreasing in expectation as we proceed down the Chosen Path Tree.
\end{remark}
\paragraph{Running time.}
We will bound the running time of a solution to the general set similarity self-join problem that uses several calls to the 
\cpsj algorithm in order to piece together a list of results $L \subseteq S \simjoin S$.
In most of the previous related work, inspired by Locality-Sensitive Hashing, the fine-grainedness of the randomized partition of space, 
here represented by the Chosen Path Tree in the \cpsj algorithm, has been controlled by a single global parameter~$k$~\cite{Gionis99, PaghSIMJOIN2015}.
In the Chosen Path setting this rule would imply that we run the splitting step without performing any brute force comparison until reaching depth $k$ 
where we proceed by comparing $x$ against every other point in nodes containing $x$, reporting close pairs. 
In recent work by Ahle et al.~\cite{ahle2017} it was shown how to obtain additional performance improvements by setting an individual depth $k_{x}$ for every $x \in S$.
We refer to these stopping strategies as global and individual, respectively.
Together with our recursion strategy, this gives rise to the following stopping criteria for when to compare a point $x$ against everything else contained in a node: 
\begin{itemize}
	\item Global: Fix a single depth $k$ for every $x \in S$.
	\item Individual: For every $x \in S$ fix a depth $k_{x}$. 
	\item Adaptive: Remove $x$ when the expected number of comparisons is non-decreasing in the tree-depth.
\end{itemize}
Let $T$ denote the running time of our similarity join algorithm. 
We aim to show the following relation between the running time between the different stopping criteria when applied to the Chosen Path Tree:
\begin{equation}
	\E[T_{\text{Adaptive}}] \leq \E[T_{\text{Individual}}] \leq \E[T_{\text{Global}}]. 
\end{equation}
First consider the global strategy. 
We set $k$ to balance the contribution to the running time from the expected number of vertices containing a point, given by $(1/\lambda)^{k}$,
and the expected number of comparisons between pairs of points at depth $k$, resulting in the following expected running time for the global strategy:
\begin{equation*}
	O\left(\min_{k} n (1/\lambda)^{k} +  \sum_{\substack{x,y \in S \\ x \neq y }}  (\simil(x, y) / \lambda)^{k}  \right).
\end{equation*}
The global strategy is a special case of the individual case, and it must therefore hold that $\E[T_{\text{Individual}}] \leq \E[T_{\text{Global}}]$.
The expected running time for the individual strategy is upper bounded by:
\begin{equation*}
	O\left(\sum_{x \in S}\min_{k_{x}} \left( (1/\lambda)^{k_{x}} + \sum_{y \in S \setminus \{ x \}} (\simil(x, y) / \lambda)^{k_{x}} \right) \right).
\end{equation*}
We will now argue that the expected running time of the \cpsj algorithm under the adaptive stopping criteria is no more than a constant factor greater than 
$\E[T_{\text{Individual}}]$ when we set the global parameters of the \textsc{BruteForce} subroutine as follows:
\begin{align*}
	\mathtt{limit} = \Theta(1), \\
	\varepsilon = \frac{\log(1/\lambda)}{\log n}.
\end{align*}
Let $x \in S$ and consider a path $p$ where $x$ is removed in from $S_{p}$ by the \textsc{BruteForce} step. 
Let $k_{x}'$ denote the depth of the node (length of $p$) at which $x$ is removed.
Compared to the individual strategy that removes $x$ at depth $k_{x}$ we are in one of three cases, also displayed in Figure \ref{fig:paths}.
\begin{enumerate}
	\item The point $x$ is removed from path $p$ at depth $k_{x}' = k_{x}$.
	\item The point $x$ is removed from path $p$ at depth $k_{x}' < k_{x}$.
	\item The point $x$ is removed from path $p$ at depth $k_{x}' > k_{x}$.
\end{enumerate}
\begin{figure}[htpb]
  \centering
  \includegraphics[width=0.45\textwidth]{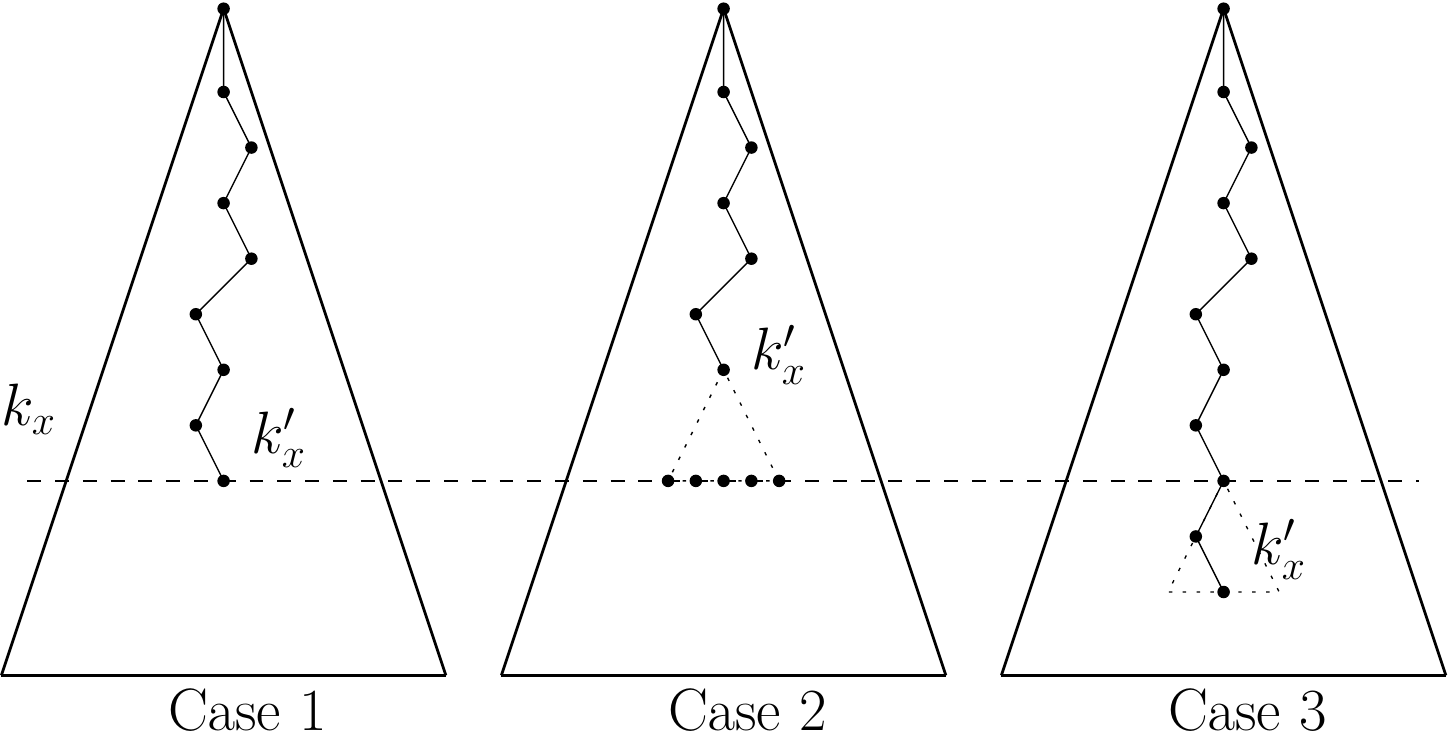}
  \caption{Path termination depth in the Chosen Path Tree}
    \label{fig:paths}
\end{figure}
The underlying random process behind the Chosen Path Tree is not affected by our choice of termination strategy.
In the first case we therefore have that the expected running time is upper bounded by the same (conservative) expression as the one used by the individual strategy.
In the second case we remove $x$ earlier than we would have under the individual strategy. 
For every $x \in S$ we have that $k_{x} \leq 1/\varepsilon$ since for larger values of $k_{x}$ the expected number of nodes containing $x$ exceeds $n$.
We therefore have that $k_{x} - k_{x}' \leq 1/\varepsilon$.
Let $S'$ denote the set of points in the node where $x$ was removed by the \textsc{BruteForce} subprocedure.
There are two rules that could have triggered the removal of $x$: Either $|S'| = O(1)$ or the condition in equation \eqref{eq:bruteforce} was satisfied.
In the first case, the expected cost of following the individual strategy would have been $\Omega(1)$ simply from the $1/\lambda$ children containing $x$ in the next step.
This is no more than a constant factor smaller than the adaptive strategy.
In the second case, when the condition in equation \eqref{eq:bruteforce} is activated we have that the expected number of comparisons 
involving $x$ resulting from $S'$ if we had continued under the individual strategy is at least
\begin{equation*}
	(1-\varepsilon)^{1/\varepsilon}|S'| = \Omega(|S'|)
\end{equation*}
which is no better than what we get with the adaptive strategy.
In the third case where we terminate at depth $k_{x}' > k_{x}$, if we retrace the path to depth $k_{x}$ we know that $x$ was not removed in this node, 
implying that the expected number of comparisons when continuing the branching process on~$x$ is decreasing compared to removing $x$ at depth $k_{x}$.
We have shown that the expected running time of the adaptive strategy is no greater than a constant times the expected running time of the individual strategy.

We are now ready to state our main theoretical contribution, stated below as Theorem \ref{thm:main}. 
The theorem combines the above argument that compares the adaptive strategy against the individual strategy together with Lemma \ref{lem:depth} and Lemma \ref{lem:correctness}, 
and uses $O(\log^{2} n)$ runs of the \cpsj algorithm to solve the set similarity join problem for every choice of constant parameters $\lambda,\varphi$.
\begin{theorem}\label{thm:main}
	For every LSHable similarity measure and every choice of constant threshold $\lambda \in (0,1)$ and probability of recall $\varphi \in (0,1)$ 
	we can solve the $(\lambda,\varphi)$-set similarity join problem on every set $S$ of $n$ points using working space $\tilde{O}(n)$ and with expected running time
	\begin{equation*}
		\tilde{O}\left(\sum_{x \in S}\min_{k_{x}} \left( \sum_{y \in S \setminus \{ x \}} (\simil(x, y) / \lambda)^{k_{x}} + (1/\lambda)^{k_{x}} \right) \right).
	\end{equation*}
\end{theorem}


\section{Implementation} \label{sec:implementation}
We implement an optimized version of the \cpsj algorithm for solving the Jaccard similarity self-join problem.
In our experiments (described in Section \ref{sec:experiments}) we compare the \cpsj algorithm against 
the approximate methods of MinHash LSH~\cite{Gionis99, Broder_NETWORK97} and BayesLSH~\cite{chakrabarti2015bayesian},
as well as the AllPairs~\cite{Bayardo_WWW07} exact similarity join algorithm.
The code for our experiments is written in C\texttt{++} and uses the benchmarking framework and data sets of the recent experimental survey on exact similarity join algorithms by Mann et al.~\cite{Mann2016}.  
For our implementation we assume that each set $x$ is represented as a list of 32-bit unsigned integers. 
We proceed by describing the details of each implementation in turn. 

\subsection{Chosen Path Similarity Join} \label{sec:implementation_cpsj}
The implementation of the \cpsj algorithm follows the structure of the pseudocode in Algorithm \ref{alg:cpsjoin} and Algorithm \ref{alg:bruteforce},
but makes use of a few heuristics, primarily sampling and sketching, in order to speed things up. 
The parameter setting is discussed and investigated experimentally in section \ref{sec:parameters}.

\medskip

\textbf{Preprocessing.}
Before running the algorithm we use the embedding described in section~\ref{sec:reduction}.
Specifically $t$ independent MinHash functions $h_1, \dots, h_t$ are used to map each set $x \in S$ to a list of $t$ hash values $(h_1(x), \dots, h_t(x))$.
The MinHash function is implemented using Zobrist hashing~\cite{zobrist1970new} from 32 bits to 64 bits with 8-bit characters.
We sample a MinHash function $h$ by sampling a random Zobrist hash function $g$ and let $h(x) = \arg\!\min_{j \in x} g(j)$. 
Zobrist hashing (also known as simple tabulation hashing) has been shown theoretically to have strong MinHash properties and is very fast in practice~\cite{Patrascu2012,Thorup2015}.  
We set $t = 128$ in our experiments, see discussion later. 

During preprocessing we also prepare sketches using the 1-bit minwise hashing scheme of Li and K{\"o}nig~\cite{li2011theory}. 
Let $\ell$ denote the length in 64-bit words of a sketch $\hat{x}$ of a set $x \in S$. 
We construct sketches for a data set $S$ by independently sampling $64 \times \ell$ MinHash functions $h_i$ and Zobrist hash functions $g_i$ that map from 32 bits to 1 bit. 
The $i$th bit of the sketch $\hat{x}$ is then given by $g_i(h_i(x))$.
In the experiments we set $\ell = 8$.

\medskip

\textbf{Similarity estimation using sketches.}
We use 1-bit minwise hashing sketches for fast similarity estimation in the \textsc{BruteForcePairs} and \textsc{BruteForcePoint} subroutines of the \textsc{BruteForce} step of the \cpsj algorithm.
Given two sketches, $\hat{x}$ and $\hat{y}$, we compute the number of bits in which they differ by going through the sketches word for word, computing the popcount of their XOR using the \texttt{gcc} builtin \texttt{\_mm\_popcnt\_u64} that translates into a single instruction on modern hardware. 
Let $\hat{J}(x, y)$ denote the estimated similarity of a pair of sets $(x, y)$. If $\hat{J}(x, y)$ is below a threshold $\hat{\lambda} \approx \lambda$, we exclude the pair from further consideration. If the estimated similarity is greater than $\hat{\lambda}$ we compute the exact similarity and report the pair if $J(x, y) \geq \lambda$. 

The speedup from using sketches comes at the cost of introducing false negatives:
A pair of sets $(x, y)$ with $J(x, y) \geq \lambda$ may have an estimated similarity less than $\hat{\lambda}$, causing us to miss it. 
We let $\delta$ denote a parameter for controlling the false negative probability of our sketches and set $\hat{\lambda}$ such that for sets $(x, y)$ with $J(x, y) \geq \lambda$ we have that $\Pr[\hat{J}(x, y) < \hat{\lambda}] < \delta$. 
In our experiments we set the sketch false negative probability to be $\delta = 0.05$.

\medskip

\textbf{Splitting step.}
The ``splitting step'' of the \cpsj algorithm as described in Algorithm \ref{alg:cpsjoin} where the set $S$ is split into buckets $S_j$ is implemented using the following heuristic: 
Instead of sampling a random hash function and evaluating it on each element $j \in x$, 
we sample an expected $1/\lambda$ elements from $[t]$ and split $S$ according to the corresponding minhash values from the preprocessing step.
This saves the linear overhead in the size of our sets $t$, reducing the time spent placing each set into buckets to $O(1)$.
Internally, a collection of sets $S$ is represented as a C\texttt{++} \texttt{std::vector<uint32\_t>} of set ids.
The collection of buckets $S_j$ is implemented using Google's \texttt{dense\_hash} hash map implementation from the \texttt{sparse\_hash} package~\cite{sparsehash}.

\medskip

\textbf{BruteForce step.}
Having reduced the overhead for each set $x \in S$ to $O(1)$ in the splitting step, we wish to do the same for the \textsc{BruteForce} step (described in Algorithm \ref{alg:bruteforce}), 
at least in the case where we do not call the \textsc{BruteForcePairs} or \textsc{BruteForcePoint} subroutines.
The main problem is that we spend time $O(t)$ for each set when constructing the \texttt{count} hash map and estimating the average similarity of $x$ to sets in $S \setminus \{x\}$.
To get around this we construct a 1-bit minwise hashing sketch $\hat{s}$ of length $64 \times \ell$ for the set $S$ using sampling and our precomputed 1-bit minwise hashing sketches.
The sketch $\hat{s}$ is constructed as follows: Randomly sample $64 \times \ell$ elements of $S$ and set the $i$th bit of $\hat{s}$ to be the $i$th bit of the $i$th sample from $S$.
This allows us to estimate the average similarity of a set $x$ to sets in $S$ in time $O(\ell)$ using word-level parallelism. 
A set $x$ is removed from $S$ if its estimated average similarity is greater than $(1 - \varepsilon)\lambda$. 
To further speed up the running time we only call the \textsc{BruteForce} subroutine once for each call to \cpsj, 
calling \textsc{BruteForcePoint} on all points that pass the check rather than recomputing $\hat{s}$ each time a point is removed.
Pairs of sets that pass the sketching check are verified using the same verification procedure as the \all implementation by Mann et al.~\cite{Mann2016}.
In our experiments we set the parameter $\varepsilon = 0.1$.
Duplicates are removed by sorting and performing a single linear scan.

\subsection{MinHash LSH}
We implement a locality-sensitive hashing similarity join using MinHash according to the pseudocode in Algorithm~\ref{alg:minhash}.
A single run of the \mh algorithm can be divided into two steps: 
First we split the sets into buckets according to the hash values of $k$ concatenated MinHash functions $h(x) = (h_1(x), \dots, h_k(x))$.
Next we iterate over all non-empty buckets and run \textsc{BruteForcePairs} to report all pairs of points with similarity above the threshold $\lambda$.
The \textsc{BruteForcePairs} subroutine is shared between the \mh and \cpsj implementation.
\mh therefore uses 1-bit minwise sketches for similarity estimation in the same way as in the implementation of the \cpsj algorithm described above. 

The parameter $k$ can be set for each dataset and similarity threshold $\lambda$ to minimize the combined cost of lookups and similarity estimations performed by algorithm.
This approach was mentioned by Cohen et al.~\cite{cohen2001finding} but we were unable to find an existing implementation.
In practice we set $k$ to the value that results in the minimum estimated running time when running the first part (splitting step) of the algorithm for values of $k$ in the range $\{2, 3, \dots, 10\}$ and estimating the running time by looking at the number of buckets and their sizes. 
Once $k$ is fixed we know that each repetition of the algorithm has probability at least $\lambda^k$ of reporting a pair $(x, y)$ with $J(x, y) \geq \lambda$. 
For a desired recall $\varphi$ we can therefore set $L = \lceil \ln(1/(1-\varphi)) / \lambda^k \rceil$.
In our experiments we report the actual number of repetitions required to obtain a desired recall rather than using the setting of $L$ required for worst-case guarantees.
\begin{algorithm}
\SetKwInOut{Params}{Parameters}
\SetKwArray{Buckets}{buckets}
\DontPrintSemicolon
\Params{$k \geq 1, L \geq 1$.}
\For{$i \leftarrow 1$ \KwTo $L$} {
\emph{Initialize hash map \Buckets{\,}.}\;
\emph{Sample $k$ MinHash fcts.} $h \leftarrow (h_1, \dots, h_k)$\;
\For{$x \in S$} {
	\Buckets{$h(x)$} $\leftarrow$ \Buckets{$h(x)$} $\cup$ $\{x\}$\; 
}
\For{$S' \in$ \Buckets} {
	$\textsc{BruteForcePairs}(S', \lambda)$\;
}
}
\caption{\textsc{MinHash}$(S, \lambda)$} \label{alg:minhash}
\end{algorithm}

\subsection{AllPairs}
To compare our approximate methods against a state-of-the-art exact similarity join we use Bayardo et al.'s \all algorithm~\cite{Bayardo_WWW07} as recently implemented in the set similarity join study by Mann et al.~\cite{Mann2016}. 
The study by Mann et al. compares implementations of several different exact similarity join methods and finds that the simple \all algorithm is most often the fastest choice. 
Furthermore, for Jaccard similarity, the \all algorithm was at most $2.16$ times slower than the best out of six different competing algorithm across all the data sets and similarity thresholds used, 
and for most runs \all is at most $11\%$ slower than the best exact algorithm (see Table 7 in Mann et al.~\cite{Mann2016}). 
Since our experiments run in the same framework and using the same datasets and with the same thresholds as Mann et al.'s study, 
we consider their \all implementation to be a good representative of exact similarity join methods for Jaccard similarity.  

\subsection{BayesLSH}
For a comparison against previous experimental work on approximate similarity joins we use an implementation of \blsh in C as provided by the \blsh authors~\cite{chakrabarti2015bayesian, bayeslsh}.
The BayesLSH package features a choice between \all and LSH as candidate generation method. 
For the verification step there is a choice between \blsh and \blsh-lite.
Both verification methods use sketching to estimate similarities between candidate pairs.
The difference between BayesLSH and BayesLSH-lite is that the former uses sketching to estimate the similarity of pairs that pass the sketching check, 
whereas the latter uses an exact similarity computation if a pair passes the sketching check.
Since the approximate methods in our \cpsj and \mh implementations correspond to the approach of BayesLSH-lite we restrict our experiments to this choice of verification algorithm.
In our experiments we will use \blsh to represent the fastest of the two candidate generation methods, combined with BayesLSH-lite for the verification step.


\section{Experiments} \label{sec:experiments}
We run experiments using the implementations of \cpsj, \mh, \blsh, and \all described in the previous section.
In the experiments we perform self-joins under Jaccard similarity for similarity thresholds $\lambda \in \{0.5, 0.6, 0.7, 0.8, 0.9 \}$.
We are primarily interested in measuring the join time of the algorithms, but we also look at the number of candidate pairs $(x,y)$ considered by the algorithms during the join as a measure of performance.  
Note that the preprocessing step of the approximate methods only has to be performed once for each set and similarity measure, 
and can be re-used for different similarity joins, we therefore do not count it towards our reported join times.
In practice the preprocessing time is at most a few minutes for the largest data sets.

\paragraph{Data sets.}
The performance is measured across $10$ real~world data sets along with $4$ synthetic data sets described in Table \ref{tab:datasets}. 
All datasets except for the TOKENS datasets were provided by the authors of~\cite{Mann2016} where descriptions and sources for each data set can also be found. 
Note that we have excluded a synthetic ZIPF dataset used in the study by Mann et al.\cite{Mann2016} due to it having no results for our similarity thresholds of interest.
Experiments are run on versions of the datasets where duplicate records are removed and any records containing only a single token are ignored.
\begin{table}
	\centering
	\scriptsize
	\begin{tabular}{lrrr} \toprule
		Dataset & \# sets / $10^6$ & avg. set size & sets / tokens \\\midrule
		AOL        & $7.35$ &   $3.8$ & $18.9$ \\
		BMS-POS    & $0.32$ &   $9.3$ & $1797.9$ \\
		DBLP       & $0.10$ &  $82.7$ & $1204.4$ \\
		ENRON      & $0.25$ & $135.3$ & $29.8$ \\
		FLICKR     & $1.14$ &  $10.8$ & $16.3$ \\
		LIVEJ      & $0.30$ &  $37.5$ & $15.0$ \\
		KOSARAK    & $0.59$ &  $12.2$ & $176.3$ \\
		NETFLIX    & $0.48$ & $209.8$ & $5654.4$ \\
		ORKUT      & $2.68$ & $122.2$ & $37.5$ \\
		SPOTIFY    & $0.36$ &  $15.3$ & $7.4$ \\
		UNIFORM    & $0.10$ &  $10.0$ & $4783.7$ \\
		TOKENS10K  & $0.03$ & $339.4$ & $10000.0$ \\
		TOKENS15K  & $0.04$ & $337.5$ & $15000.0$ \\
		TOKENS20K  & $0.06$ & $335.7$ & $20000.0$ \\ \bottomrule
	\end{tabular}
	\caption{Dataset size, average set size, and average number of sets that a token is contained in.}
	\label{tab:datasets}
\end{table}

In addition to the datasets from the study of Mann et al. we add three synthetic datasets TOKENS10K, TOKENS15K, and TOKENS20K, designed to showcase the robustness of the approximate methods.
These datasets have relatively few unique tokens, but each token appears in many sets. 
Each of the TOKENS datasets were generated from a universe of $1000$ tokens ($d = 1000$) and each token is contained in respectively, $10,000$, $15,000$, and $20,000$ different sets as denoted by the name.
The sets in the TOKENS datasets were generated by sampling a random subset of the set of possible tokens,
rejecting tokens that had already been used in more than the maximum number of sets ($10,000$ for TOKENS10K).
To sample sets with expected Jaccard similarity $\lambda'$ the size of our sampled sets should be set to $(2\lambda'/(1+\lambda'))d$. 
For $\lambda' \in \{0.95, 0.85, 0.75, 0.65, 0.55\}$ the TOKENS datasets each have $100$ random sets planted with expected Jaccard similarity $\lambda'$.
This ensures an increasing number of results for our experiments where we use thresholds $\lambda \in \{0.5, 0.6, 0.7, 0.8, 0.9 \}$. 
The remaining sets have expected Jaccard similarity $0.2$.
We believe that the TOKENS datasets give a good indication of the performance on real-world data that has the property that most tokens appear in a large number of sets.

\paragraph{Recall.}
In our experiments we aim for a recall of at least $90\%$ for the approximate methods. 
In order to achieve this for the \cpsj and \mh algorithms we perform a number of repetitions after the preprocessing step, stopping when the desired recall has been achieved.
This is done by measuring the recall against the recall of \all and stopping when reaching $90\%$.
In situations where the size of the true result set is not known it can be efficiently estimated using sampling if it is not too small. 
Alternatively, the algorithms can be stopped once the rate of new results drops below some threshold, indicating that most results have been found. 
For \blsh using LSH as the candidate generation method, the recall probability with the default parameter setting is $95\%$, although we experience a recall closer to $90\%$ in our experiments.

\paragraph{Hardware.}
All experiments were run on an Intel Xeon E5-2690v4 CPU at 2.60GHz with $35$MB L$3$,$256$kB L$2$ and $32$kB L$1$ cache and $512$GB of RAM.
Since a single experiment is always confined to a single CPU core we ran several experiments in parallel~\cite{Tange2011a} to better utilize our hardware.
\subsection{Results}
\paragraph{Join time.}
Table \ref{tab:jointimes} shows the average join time in seconds over five independent runs, when approximate methods are required to have at least $90\%$ recall.
We have omitted timings for \blsh since it was always slower than all other methods, and in most cases it timed out after 20 minutes when using LSH as candidate generation method.
The join time for \mh is always greater than the corresponding join time for \cpsj except in a single setting: the dataset KOSARAK with threshold $\lambda = 0.5$.
Since \cpsj is typically $2-4\times$ faster than \mh we can restrict our attention to comparing \all and \cpsj where the picture becomes more interesting.

\begin{table*}
  \begin{changemargin}{-3.5cm}{-3cm} 
	\scriptsize
\begin{filecontents*}{simfilter/results_table.csv}
File,AA,ACP,AMH,BA,BCP,BMH,CA,CCP,CMH,DA,DCP,DMH,EA,ECP,EMH
AOL, 483.5,\textbf{ 362.1},1329.9, 117.8,\textbf{ 113.4}, 444.2,\textbf{  13.7},  42.2, 152.9,\textbf{   4.2},  34.6, 100.6,\textbf{   1.6},  21.0,  43.8
BMS-POS,  62.5,\textbf{  27.0},  40.0,  20.9,\textbf{   7.1},  13.7,   5.6,\textbf{   2.7},   5.6,\textbf{   1.3},   2.0,   3.9,\textbf{   0.2},   0.9,   1.4
DBLP, 127.9,\textbf{   9.2},  22.1,  63.8,\textbf{   2.5},  10.1,  27.4,\textbf{   1.1},   3.7,   7.8,\textbf{   0.6},   1.8,   0.8,\textbf{   0.3},   0.7
ENRON,  78.0,\textbf{   6.9},  16.4,  23.2,\textbf{   4.4},   9.9,   6.0,\textbf{   2.4},   6.3,\textbf{   1.6},   1.6,   2.7,\textbf{   0.4},   0.7,   1.7
FLICKR,\textbf{  17.2},  48.6,  68.0,\textbf{   6.0},  30.9,  37.2,\textbf{   2.5},  13.8,  21.3,\textbf{   1.0},   6.3,  11.3,\textbf{   0.3},   3.4,   5.2
KOSARAK,\textbf{  73.1}, 377.9, 311.1,\textbf{  14.4},  62.7,  89.2,\textbf{   1.6},   7.2,  16.1,\textbf{   0.5},   3.9,   9.9,\textbf{   0.1},   1.2,   2.6
LIVEJ, 571.7,\textbf{ 131.3}, 279.4, 145.3,\textbf{  48.7}, 129.6,  30.6,\textbf{  28.2},  52.9,\textbf{   7.1},  16.2,  41.0,\textbf{   1.5},   9.2,  12.6
NETFLIX,1354.7,\textbf{  25.3}, 121.8, 520.4,\textbf{   8.2},  60.0, 177.3,\textbf{   4.8},  22.6,  46.2,\textbf{   2.4},  14.1,   5.4,\textbf{   1.6},   5.8
ORKUT, 359.7,\textbf{  26.5}, 115.7, 106.4,\textbf{  15.4},  60.1,  36.3,\textbf{   8.0},  25.1,  12.2,\textbf{   7.4},  19.7,\textbf{   3.7},   4.8,  13.3
SPOTIFY,\textbf{   0.5},   2.5,   9.3,\textbf{   0.3},   1.5,   3.4,\textbf{   0.2},   1.0,   2.6,\textbf{   0.1},   1.0,   1.9,\textbf{   0.1},   0.5,   0.6
TOKENS10K, 312.1,\textbf{   3.4},   4.8, 236.8,\textbf{   2.9},   3.9, 164.0,\textbf{   1.5},   1.7, 114.9,\textbf{   0.6},   1.2,  63.2,\textbf{   0.2},   0.4
TOKENS15K, 688.4,\textbf{   4.4},   6.2, 535.3,\textbf{   4.0},   7.1, 390.4,\textbf{   1.8},   3.7, 258.2,\textbf{   0.7},   1.7, 140.0,\textbf{   0.2},   0.7
TOKENS20K,1264.1,\textbf{   5.7},  12.0, 927.0,\textbf{   4.0},  11.4, 698.4,\textbf{   2.1},   4.5, 494.3,\textbf{   0.8},   2.2, 273.4,\textbf{   0.3},   0.8
UNIFORM005,  54.1,\textbf{   3.9},   6.6,  27.6,\textbf{   1.6},   3.0,  10.5,\textbf{   0.9},   1.4,   3.6,\textbf{   0.5},   1.0,   0.4,\textbf{   0.1},   0.3
\end{filecontents*}
\csvreader[tabular=l|rrr|rrr|rrr|rrr|rrr
, table head=\toprule&\multicolumn{3}{c}{Threshold $0.5$}&\multicolumn{3}{c}{Threshold $0.6$}&\multicolumn{3}{c}{Threshold $0.7$}&\multicolumn{3}{c}{Threshold $0.8$}&\multicolumn{3}{c}{Threshold $0.9$}\\
	\midrule Dataset & CP & MH & ALL & CP & MH & ALL & CP & MH & ALL & CP & MH & ALL & CP & MH & ALL \\ \midrule
, table foot=, head to column names
, late after line=\\
, late after last line =\\\bottomrule]{simfilter/results_table.csv}{}
{\File& \ACP & \AMH & \AA & \BCP & \BMH & \BA & \CCP & \CMH & \CA &\DCP & \DMH & \DA &\ECP & \EMH & \EA}
\caption{Join time in seconds for \cpsj (CP), \mh (MH) and \all (ALL) with at least $\ge90\%$ recall.}
  \label{tab:jointimes}
\end{changemargin}

\end{table*}

Figure \ref{fig:speed} shows the join time speedup that \cpsj achieves over \all. 
We achieve speedups of between $2-50\times$ for most of the datasets, with greater speedups at low similarity thresholds.
For a number of the datasets the \cpsj algorithm is slower than \all for the thresholds considered here.
Looking at Table \ref{tab:datasets} it seems that \cpsj generally performs well on most datasets where tokens are contained in a large number of sets on average (NETFLIX, UNIFORM, DBLP) and less well on datasets that have a lot of ``rare'' tokens (SPOTIFY, LIVEJOURNAL, AOL), although the picture is not completely consistent as shown by the poor performance of \cpsj on KOSARAK.
\begin{wrapfigure}{O}[1cm]{6cm}
  \includegraphics[width=0.7\textwidth]{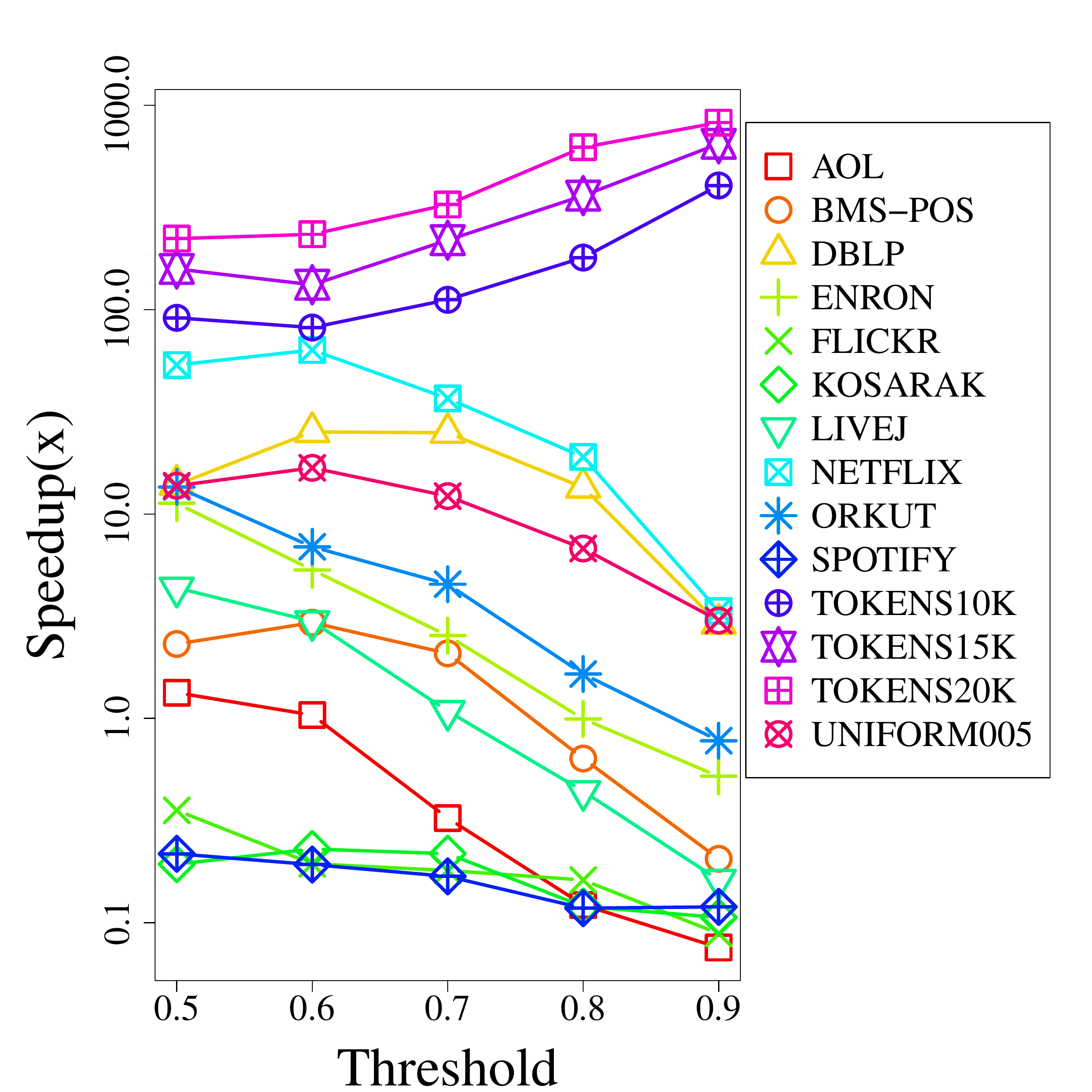}
  \caption{Join time of \cpsj with at least $90\%$ recall.}
  \label{fig:speed}
\end{wrapfigure}

\paragraph{BayesLSH.}
The poor performance of \blsh compared to the other algorithms (\blsh was always slower) can most likely be tracked down to differences in the implementation of the candidate generation methods of \blsh.
The \blsh implementation uses an older implementation of \all compared to the implementation by Mann et al.~\cite{Mann2016} which was shown to yield performance improvements by using a more efficient verification procedure. 
The LSH candidate generation method used by \blsh corresponds to the \mh splitting step, but with $k$ (the number of hash functions) fixed to one. 
Our technique for choosing $k$ in the \mh algorithm, aimed at minimizing the total join time, typically selects $k \in \{3,4,5,6\}$ in the experiments. 
It is therefore likely that \blsh can be competitive with the other techniques by combining it with other candidate generation procedures.
Further experiments to compare the performance of BayesLSH sketching to 1-bit minwise sketching for different parameter settings and similarity thresholds would also be instructive.

\paragraph{TOKEN datasets.}
The TOKENS datasets clearly favor the approximate join algorithms where \cpsj is two to three orders of magnitude faster than \all.
By increasing the number of times each token appears in a set we can make the speedup of \cpsj compared to \all arbitrarily large as shown by the progression from TOKENS10 to TOKENS20.
The \all algorithm generates candidates by searching through the lists of sets that contain a particular token, starting with rare tokens.
Since every token appears in a large number of sets every list will be long.

Interestingly, the speedup of \cpsj is even greater for higher similarity thresholds.
We believe that this is due to an increase in the gap between the similarity of sets to be reported and the remaining sets that have an average Jaccard similarity of $0.2$.
This is in line with our theoretical analysis of \cpsj and most theoretical work on approximate similarity search 
where the running time guarantees usually depend on the approximation factor.

\paragraph{Candidates and verification.}
Table \ref{tab:candidates} compares the number of pre-candidates, candidates, and results generated by the \all and \cpsj algorithms where the desired recall for \cpsj is set to be greater than $90\%$.
For \all the number of pre-candidates denotes all pairs $(x, y)$ investigated by the algorithm that pass checks on their size so that it is possible that $J(x, y) \geq \lambda$.
The number of candidates is simply the number of unique pre-candidates as duplicate pairs are removed explicitly by the \all algorithm.

For \cpsj we define the number of pre-candidates to be all pairs $(x, y)$ considered by the \textsc{BruteForcePairs} and \textsc{BruteForcePoint} subroutines of Algorithm \ref{alg:bruteforce}.
The number of candidates are pre-candidate pairs that pass size checks (similar to \all) and the 1-bit minwise sketching check as described in Section \ref{sec:implementation_cpsj}.
Note that for \cpsj the number of candidates may still contain duplicates as this is inherent to the approximate method for candidate generation. 
Removing duplicates though the use of a hash table would drastically increase the space usage of \cpsj.
For both \all and \cpsj the number of candidates denotes the number of points that are passed to the exact similarity verification step of the \all implementation of Mann et al.~\cite{Mann2016}.

Table \ref{tab:candidates} shows that for \all there is not a great difference between the number of pre-candidates and number of candidates, 
while for \cpsj the number of candidates is usually reduced by one or two orders of magnitude for datasets where \cpsj performs well.
For datasets where \cpsj performs poorly such as AOL, FLICKR, and KOSARAK there is less of a decrease when going from pre-candidates to candidates.
It would appear that this is due to many duplicate pairs from the candidate generation step and not a failure of the sketching technique.

\begin{figure*}
  \centering
  \begin{subfigure}{0.31\textwidth}
    \includegraphics[width =\textwidth]{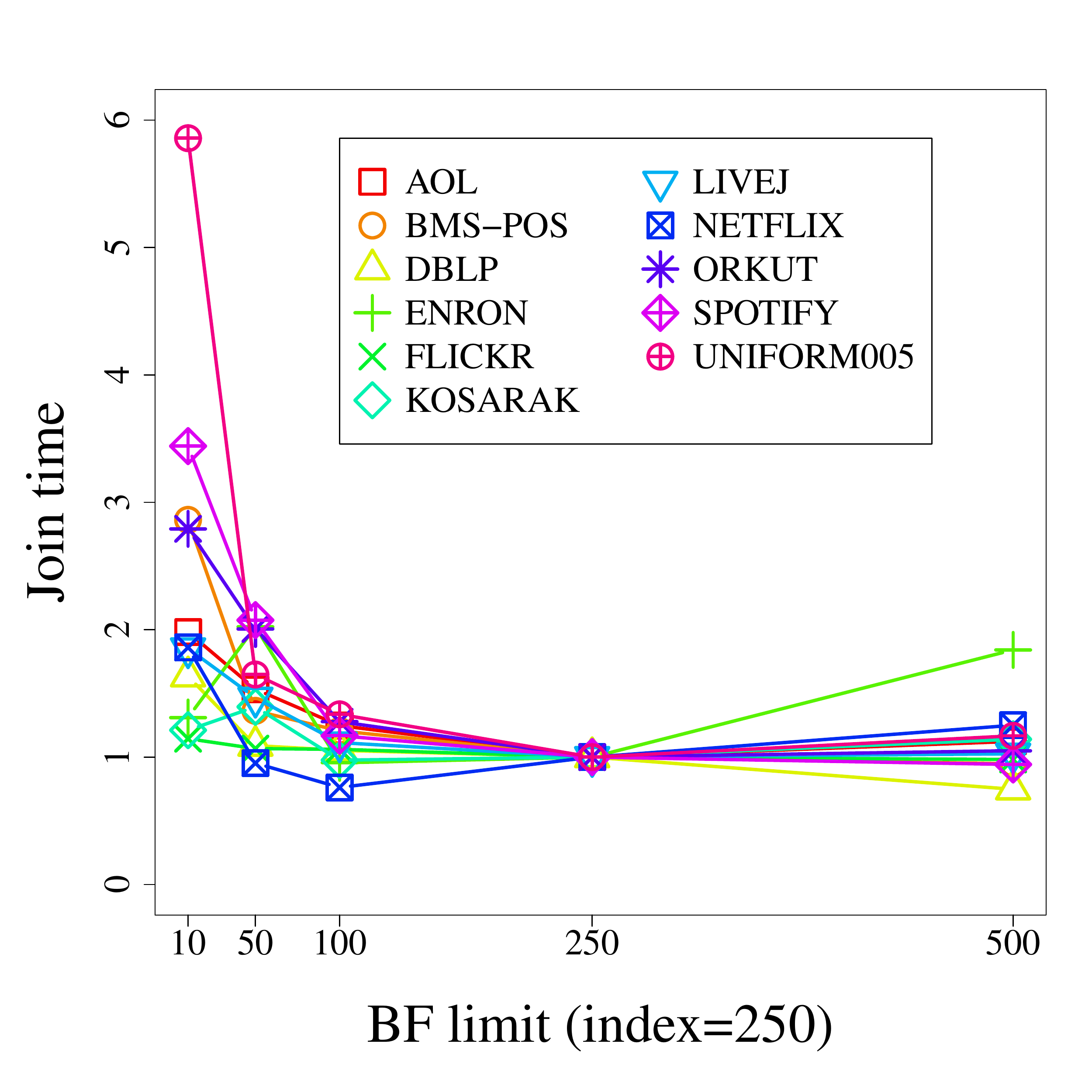}
    \caption{limit$\in\{10,\cdots,500\}$}
  \end{subfigure}
  \begin{subfigure}{0.31\textwidth}
    \includegraphics[width =\textwidth]{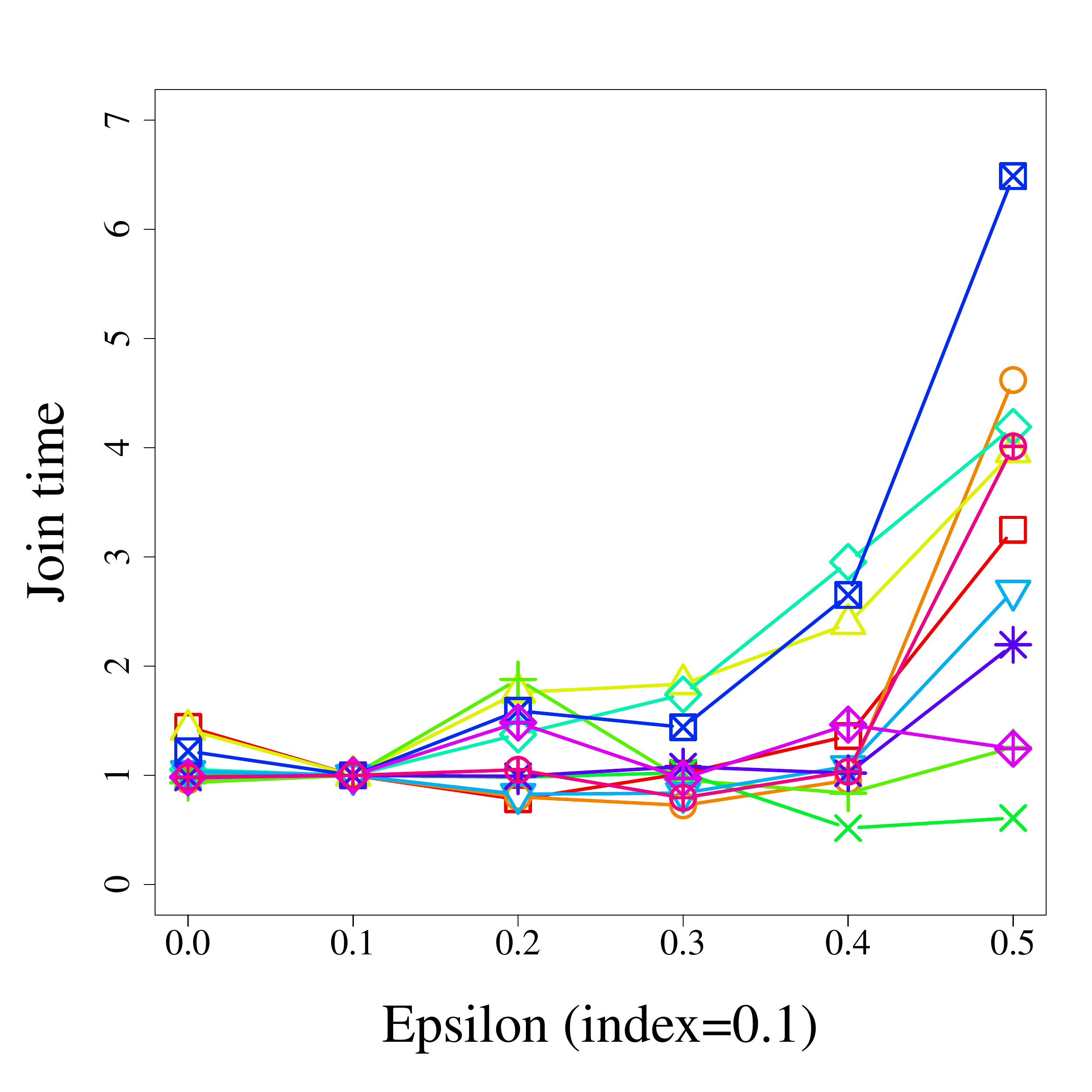}
    \caption{$\varepsilon \in \{.0, .1, .2, .3, .4, .5 \}$}
  \end{subfigure}
  \begin{subfigure}{0.31\textwidth}
    \includegraphics[width = \textwidth]{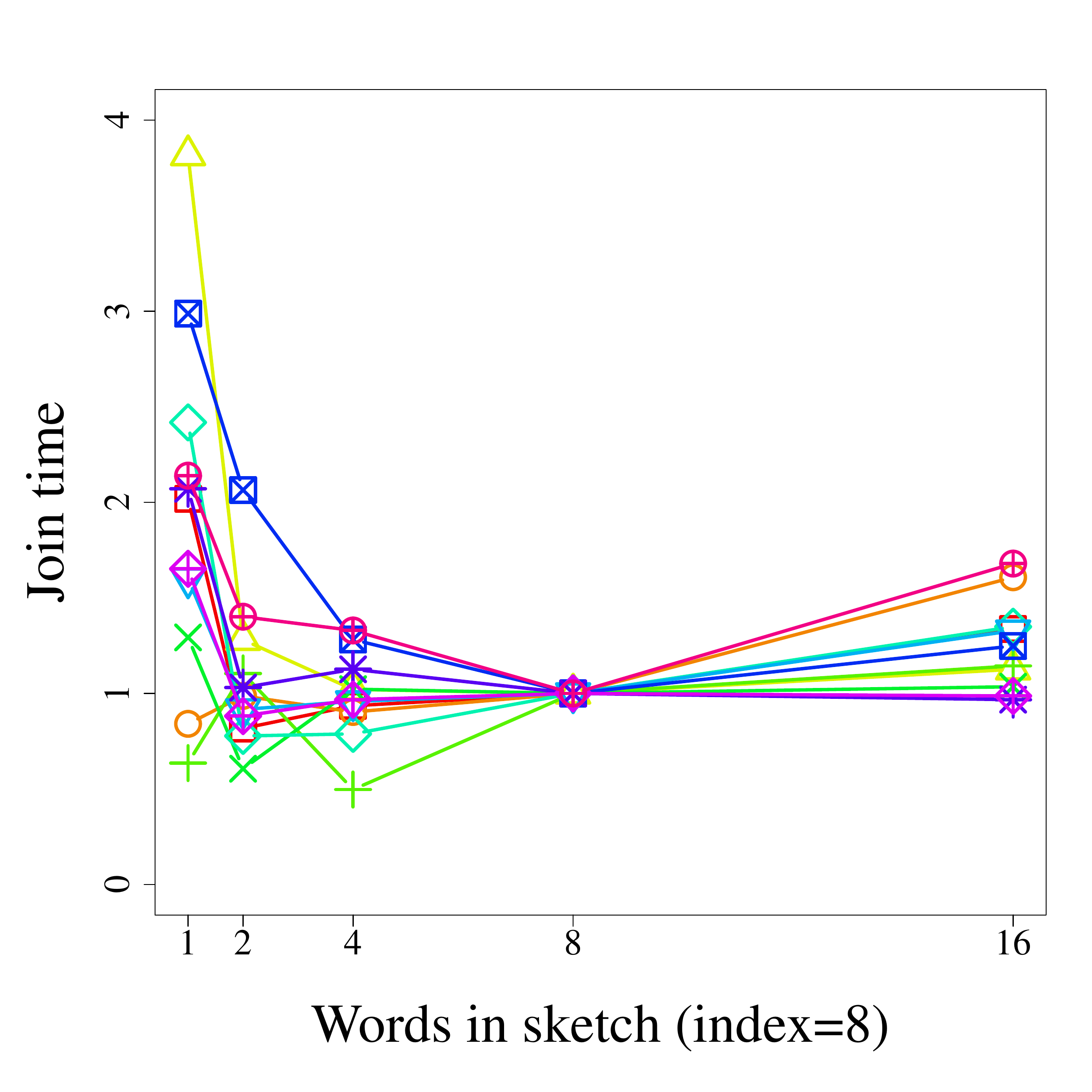}
    \caption{$w \in \{1, 2, 4, 8, 16 \}$}
  \end{subfigure}
  \caption{Relative join time for \cpsj with at least $80\%$ recall and similarity threshold $\lambda = 0.5$ for different parameter settings of \texttt{limit}, $\varepsilon$, and $w$.} 
\label{fig:parameters}
\end{figure*}

\subsection{Parameters} \label{sec:parameters}
In order to investigate how the parameter settings affect the performance of the \cpsj algorithm we run experiments where we vary the brute force parameter \texttt{limit},
the brute force aggressiveness parameter $\varepsilon$, and the sketch length in words $\ell$.
Table \ref{tab:parameters} gives an overview of the different parameters and shows how they were set during the parameter experiments and the final setting used for our join time experiments.
\begin{table}[b]
\small
\centering
	\begin{tabular}{llrr} \toprule
		Parameter & Description & Test & Final \\ \midrule
		\texttt{limit} & Brute force limit & $100$ & $250$ \\
		$\ell$ & Sketch word length & $4$ & $8$ \\
		$t$ & Size of MinHash set & $128$ & $128$ \\
		$\varepsilon$ & Brute force aggressiveness & $0.0$ & $0.1$ \\ 
		$\delta$ & Sketch false negative prob. & $0.1$ & $0.05$ \\ \bottomrule 
	\end{tabular}
	\caption{Parameters of the \cpsj algorithm, their setting during parameter experiments, and their setting for the final join time experiments}
	\label{tab:parameters}
\end{table}
Figure \ref{fig:parameters} shows the \cpsj join time for different settings of the parameters, relative to a certain parameter choice.
We argue that the join times are relatively stable around our setting of parameters,
leading us to believe that our technique of setting one parameter at a time is not too far away from the optimal setting, 
although changing one parameter probably changes the effect of other parameters to some extent.

Figure \ref{fig:parameters} (a) shows the effect of varying the brute force limit on the join time. 
Lowering \texttt{limit} to $10$ or $50$ causes the join time to increase due to a combination of spending more time splitting sets into buckets and the lower probability of recall that comes when randomly splitting the data further during candidate generation. The join time is relatively stable for \texttt{limit} $\in \{100, 250, 500\}$.  

Figure \ref{fig:parameters} (b) shows the effect of varying the brute force aggressiveness on the join time. 
As we increase $\varepsilon$, sets that are close to the other elements in their buckets are more likely to be removed by brute force comparing them to all other points.
The tradeoff here is between the loss of probability of recall by letting a point continue in the $\cp$ branching process versus the cost of brute forcing the point.
The join time is generally increasing as we increase $\varepsilon$ due to the cost of performing more brute force comparisons.
Nevertheless, it turns out that $\varepsilon = 0.1$ is a slightly better setting than $\varepsilon = 0.0$ for almost all data sets.

Figure \ref{fig:parameters} (c) shows the effect of varying the sketch length on the join time.
There is a tradeoff between the sketch similarity estimation time and the precision of the estimate, leading to fewer false positives.
For a similarity threshold of $\lambda = 0.5$  using only a single word negatively impacts the performance on most datasets compared to using two or more words. 
The cost of using longer sketches seems neglible as it is only a few extra instructions per similarity estimation so we opted to use $\ell = 8$ words in our sketches.
\begin{table}
\scriptsize
\centering
\renewcommand*{\arraystretch}{1.05}
\begin{tabular}{l|rr|rr}
\toprule
Dataset    &   \multicolumn{2}{c}{Threshold $0.5$}&\multicolumn{2}{c}{Threshold $0.7$} \\ 
		   &\multicolumn{1}{c}{ALL}&\multicolumn{1}{c}{CP} &	\multicolumn{1}{c}{ALL}&\multicolumn{1}{c}{CP} 	  \\ \midrule
           & 8.5E+09 & 7.4E+09 & 6.2E+08 & 2.9E+09 \\
AOL        & 8.5E+09 & 1.4E+09 & 6.2E+08 & 3.1E+07 \\
           & 1.3E+08 & 1.2E+08 & 1.6E+06 & 1.5E+06 \\ \hline
           & 2.0E+09 & 9.2E+08 & 2.7E+08 & 3.3E+08 \\
BMS-POS    & 1.8E+09 & 1.7E+08 & 2.6E+08 & 4.9E+06 \\
           & 1.1E+07 & 1.0E+07 & 2.0E+05 & 1.8E+05 \\ \hline
           & 6.6E+09 & 4.6E+08 & 1.2E+09 & 1.3E+08 \\
DBLP       & 1.9E+09 & 4.6E+07 & 7.2E+08 & 4.3E+05 \\
           & 1.7E+06 & 1.6E+06 & 9.1E+03 & 8.5E+03 \\ \hline
           & 2.8E+09 & 3.7E+08 & 2.0E+08 & 1.5E+08 \\
ENRON      & 1.8E+09 & 6.7E+07 & 1.3E+08 & 2.1E+07 \\
           & 3.1E+06 & 2.9E+06 & 1.2E+06 & 1.2E+06 \\ \hline
           & 5.7E+08 & 2.1E+09 & 9.3E+07 & 9.0E+08 \\
FLICKR     & 4.1E+08 & 1.1E+09 & 6.3E+07 & 3.8E+08 \\
           & 6.6E+07 & 6.1E+07 & 2.5E+07 & 2.3E+07 \\ \hline
           & 2.6E+09 & 4.7E+09 & 7.4E+07 & 4.2E+08 \\
KOSARAK    & 2.5E+09 & 2.1E+09 & 6.8E+07 & 2.1E+07 \\
           & 2.3E+08 & 2.1E+08 & 4.4E+05 & 4.1E+05 \\  \hline
           & 9.0E+09 & 2.8E+09 & 5.8E+08 & 1.2E+09 \\
LIVEJ      & 8.3E+09 & 3.6E+08 & 5.6E+08 & 1.8E+07 \\
           & 2.4E+07 & 2.2E+07 & 8.1E+05 & 7.6E+05 \\ \hline
           & 8.6E+10 & 1.3E+09 & 1.0E+10 & 4.3E+08 \\
NETFLIX    & 1.3E+10 & 3.1E+07 & 3.4E+09 & 6.4E+05 \\
           & 1.0E+06 & 9.5E+05 & 2.4E+04 & 2.2E+04 \\ \hline
           & 5.1E+09 & 1.1E+09 & 3.0E+08 & 7.2E+08 \\
ORKUT      & 3.9E+09 & 1.3E+06 & 2.6E+08 & 8.1E+04 \\
           & 9.0E+04 & 8.4E+04 & 5.6E+03 & 5.3E+03 \\ \hline
           & 5.0E+06 & 1.2E+08 & 4.7E+05 & 8.5E+07 \\
SPOTIFY    & 4.8E+06 & 3.1E+05 & 4.6E+05 & 2.7E+03 \\
           & 2.0E+04 & 1.8E+04 & 2.0E+02 & 1.9E+02 \\ \hline
           & 1.5E+10 & 1.7E+08 & 8.1E+09 & 4.9E+07 \\
TOKENS10K  & 4.1E+08 & 5.7E+06 & 4.1E+08 & 1.9E+06 \\
           & 1.3E+05 & 1.3E+05 & 7.4E+04 & 6.9E+04 \\ \hline
           & 3.6E+10 & 3.0E+08 & 1.9E+10 & 8.1E+07 \\
TOKENS15K  & 9.6E+08 & 7.2E+06 & 9.6E+08 & 1.9E+06 \\
           & 1.4E+05 & 1.3E+05 & 7.5E+04 & 6.9E+04 \\ \hline
           & 6.4E+10 & 4.4E+08 & 3.4E+10 & 1.0E+08 \\
TOKENS20K  & 1.7E+09 & 8.8E+06 & 1.7E+09 & 1.9E+06 \\
           & 1.4E+05 & 1.4E+05 & 7.9E+04 & 7.4E+04 \\ \hline
           & 2.5E+09 & 3.7E+08 & 6.5E+08 & 1.3E+08 \\
UNIFORM005 & 2.0E+09 & 9.5E+06 & 6.1E+08 & 3.9E+04 \\
           & 2.6E+05 & 2.4E+05 & 1.4E+03 & 1.3E+03 \\ \bottomrule 
\end{tabular}
\caption{Number of pre-candidates, candidates and results for ALL and CP with at least $90\%$ recall.}
\label{tab:candidates}
\end{table}


\section{Conclusion}
In this chapter we provide experimental and theoretical results on a new randomized set similarity join algorithm, \cpsj.
We compare \cpsj experimentally to state-of-the-art exact and approximate set similarity join algorithms.
\cpsj is typically $2-4$ times faster than previous approximate methods.
Compared to exact methods it obtains speedups of more than an order of magnitude on real-world datasets, while keeping the recall above $90\%$.

Among the datasets used in these experiments we note that NETFLIX and FLICKR represents two archetypes.
On average a token in the NETFLIX dataset appears in more than $5000$ sets while on average a token in the FLICKR dataset appears in less than $20$ sets. 
Our experiment indicate that \cpsj brings large speedups to the NETFLIX type datasets, while it is hard to improve upon the perfomance of \all on the FLICKR type.

A direction for future work could be to tighten and simplify the theoretical analysis to better explain the experimental results.
We conjecture that the running time of the algorithm can be bounded by a simpler function of the sum of similarities between pairs of points in $S$.

We note that recursive methods such as ours lend themselves well to parallel and distributed implementations since most of the computation happens in independent, recursive calls. Further investigating this is an interesting possibility.
\medskip

\textbf{Acknowledgement.} The authors would like to thank Willi Mann for making the source code and data sets of the study~\cite{Mann2016} available, and Aniket Chakrabarti for information about the implementation of BayesLSH.


\chapter{Summary and open problems}
\label{cha:direction}
In this chapter we revisit our results, but with a focus on future research directions and open problems.
We refer to Section~\ref{sec:problems} for a general overview of the results.

\medskip
\noindent
In Chapter~\ref{sec:furthest-neighbor} we presented a data structure for the approximate furthest neighbor problem~(Definition \ref{def:c-AFN}).
Our main contribution is the development of a new query procedure for the problem that eliminates the need for multiple $r$-far searches.
We showed that for iteration-based data structures is not possible to store less than $\min\{n, 2^{\BOMx{d}}\}-1$ points for $c$-AFN when $c<\sqrt{2}$.
However when $c=\sqrt{2}$ we need just $d+1$ points~\cite{Goel2001}~(See also Appendix~\ref{cha:2afn}).
It would be interesting to understand better why $\sqrt{2}$ is a special threshold, and to extend the lower bound beyond iteration-based data structures.
We show that the query-independent variation of our algorithm stores $\BOx{f(c)^d}$ points for some function $f$~(Section~\ref{sub:query-ind})
However our algorithm only works with high probability, and we do not have a closed form for $f$. 
An interesting open problem is to close this gap to the lower bound.

\begin{open problem}
Design a $\sqrt{2}(1-\epsilon )$-AFN data structure for $\epsilon \in(0,1)$ using space $\BOx{d\min\{n,2^{\BOx{d\epsilon^2}}\}}$ with query time $n^{1-\BOMx{1}}$.
\end{open problem}

\medskip
\noindent
In Chapter~\ref{sec:annulus-query} we used the $c$-AFN result in combination with LSH techniques to solve to approximate annulus query problem~(Definition~\ref{def:aaq}).
Our contribution here is the analysis of this combined data structure, achieving sub-linear query time.
An interesting direction of future research is in further combination of our data structure with LSH based data structures.
For example to improve the output sensitivity of near neighbor search based on LSH.
By replacing each hash bucket with an AFN data structure with suitable
approximation factors, it is possible to control the number of times each point in $S$ is reported.
Recent work on distance-sensitive hashing suggest a larger framework extending to ``anti-lsh'' functions~\cite{AumullerCP017}.
It would be interesting future work to place our results in that context.

\medskip
\noindent
The distance sensitive membership query investigated in Chapter~\ref{sec:dist-sens-appr} has not been the subject
of much prior research. In particular we have been unable to find any previous results without false negatives,
so there are many unanswered questions.
Our contributions are upper and lower bounds on the space usage for this problem in $(\{0,1\}^d,H)$.
Most pressing we do not show much in regards to query time.
Our method would use time $\BOx{n}$ to make a comparison to each of the stored signatures.
This could possibly be improved by using additional similarity search methods that avoid false negatives~(e.g.~\cite{Pagh2016}), but that would come with increased false positives.
In comparison a regular Bloom filter uses $\BOx{k}$ time independently of how many items are in the set.
However, a solution with constant time (or even polylog in $n$) could be used, say with $\varepsilon = 1/n$, to solve the $c$-approximate nearest neighbor problem.
The best currently known data structures for this problem use $n^{\Omega(1/c)}$ time~\cite{andoni2015optimal}.

\begin{open problem}
  Design a distance sensitive approximate membership filter for $(\{0,1\}^d,H)$ with query time $\BOx{n^{1/c}}$ and space $\BOx{n^{1+1/c}}$.
\end{open problem}

The signature vector method we introduced does not really extend well to other spaces.
This is another obvious area for future work.

\begin{open problem}
  Show non-trivial bounds for the $(r,c,w)$-DAMQ problem in $(\mathbb{R}^d,\ell_2)$.
\end{open problem}

Note that embeddings a la Johnson and Lindenstrauss can not be used here as they would introduce false negatives.

\medskip
\noindent
In Chapter \ref{cha:simil-pres-embedd} we look at nearest neighbor preserving embeddings.
The benefit of using this setting as opposed to normal distance preserving embeddings is that it is possible to embed into lower dimensional spaces.
Our contribution is showing that this benefit can be achieved while using sparse matrices and giving an analysis of the FJLT transform in this setting.
In the presented embedding, the embedding dimensionality is independent of $n$, but relies instead on $\lambda_X$.
Could the sparsity parameter $f$ be a similarly disconnected from the size of the embedded set?
In the spanning tree construction used in the proof of Theorem~\ref{thm:fast-near-neighbor-by-FJLT} this seems to be achievable if we can show results using only distances between the covering balls, and not the actual points inside them.
This would require a new way of bounding the probability that no point ``leaves'' a ball, independently of how many points are inside it. 
Currently we get $f=\BOx{d^{-1}\log^2{n}}$, but there are known $\ell_2$ embeddings with $\BOx{\epsilon}$ sparsity~\cite{DBLP:journals/jacm/KaneN14}.
Achieving similar results for nearest neighbor preserving embeddings would allow for much faster embeddings.

\begin{open problem}
  Construct a nearest neighbor preserving embedding with $k=\BOx{\epsilon ^{-2}\log{\lambda_s}\log{(2/\epsilon) }}$ and sparsity $\BOx{\epsilon }$.
\end{open problem}

\medskip
\noindent
Finally, in Chapter~\ref{cha:set-similarity} we looked at the set similarity join problem~(Definition~\ref{def:setsimjoin}).
We presented the \cpsj algorithm, based on the Chosen Path Tree.
Unlike previous LSH based methods we eliminate the setting of $k$ as a parameter by presenting an automatically adapting algorithm.
Our main theoretical contribution here is in analyzing the query time as well as giving probabilistic bounds for space and recall.
Empirically our methods are very fast on all data sets, but they can still be beaten by exact methods on data sets well suited for prefix filtering. 
It would be interesting future work to develop approximate set similarity methods that achieve high recall significantly faster than exact methods for all data sets.
Another direction would be to attempt to improve our recall guarantees, either through altering the algorithm or tightening the analysis.


\appendix
\chapter{Appendix}
\section{Properties of Gaussians}
\label{sec:append-i:pr-gauss}
\begin{lemma}
  \label{lm:nonstdgauss}
  Let $X\sim \mathcal{N}(0,x)$ and $Y\sim \mathcal{N}(0,y)$.
  Then $\forall t>0$:
  \begin{align}
    y\geq x &\Rightarrow \Pr[X^2\leq t]\ge \Pr[Y^2\leq t]\\
    y\leq x &\Rightarrow \Pr[X^2\leq t]\leq \Pr[y^2\leq t]
  \end{align}
With equality exactly when $x=y$.  

  \begin{proof}
    Let $y\geq x$:
    \begin{align}
      \Pr[X^2\leq t]=\Pr[X\leq\sqrt{t}]-\Pr[X\leq-\sqrt t]&\geq \\
      \Pr[Y\leq\sqrt t]-\Pr[X\leq-\sqrt t]&\geq \\
      \Pr[Y\leq\sqrt t]-\Pr[Y\leq-\sqrt t]&=\Pr[Y^2\leq t]
    \end{align}
    Similarly in the other direction when $y\leq x$.
    \end{proof}
\end{lemma}
    We can generalize to sums of such variables:
    
    \begin{lemma}
      \label{lm:generalizednonstd}
      For any integer $k\geq 1$. Let $X=\sum_{i=1}^k X_i^2$ where $X_i\sim\mathcal{N}(0,x_i)$ and $Y=\sum_{i=1}^kY_i^2$ where $Y_i\sim\mathcal{N}(0,y_i)$.
      Then if $y_i\geq x_i$ for all $i\in\{1,\cdots,k\}$ we have: 
      \[\Pr[Y\leq t]\leq\Pr[X\leq t].\]
      \begin{proof}
        We show a standard proof by induction.
        Define a new variable $S_l=\sum_{i=1}^{k-l}X_i^2+\sum_{j=k-l+1}^kY_{j}^2$.
        As a base case set $l=1$:
        \[\Pr[S_1\leq t]= \Pr[Y_k^2+\sum_{i=1}^{k-1}X_i^2\leq t]\leq\Pr[S_0\leq t].\]
        By fixing $X^2_i$ for $1\leq i\leq k-1$ and using lemma.~\ref{lm:nonstdgauss}.
        And generally for all integers $l>0$ up to $l=k$:

        \[\Pr[S_l\leq t]=\Pr[S_{l-1}\leq t ]\]
        By fixing everything but the $l$'th variable and using lemma.~\ref{lm:nonstdgauss}.
        We arrive at $\Pr[Y\leq t]=\Pr[S_k\leq t]\leq\Pr[S_0\leq t]=\Pr[X\leq t]$.        
      \end{proof}
      
    \end{lemma}

\section{$\sqrt{2}$-AFN in $d+1$ points}

\label{cha:2afn}
\begin{figure}[h]
  \centering
  \includegraphics{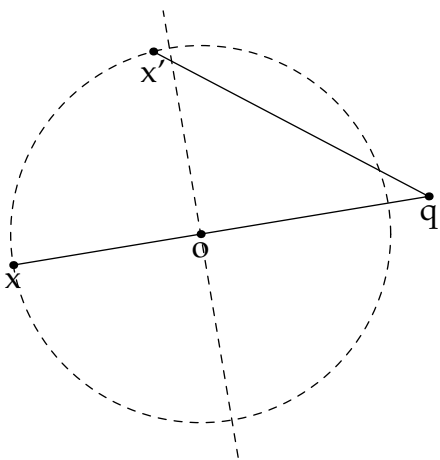}
  \caption{Illustration of the construction in the plane.}
  \label{fig:afn-sqrt2}
\end{figure}

\begin{theorem}
  For $c\geq \sqrt{2}$ there exists a data structure that computes the $c$-AFN of any set $S\subseteq\mathbb{R}^d$ by storing a size $d+1$ subset of $S$.

  \begin{proof}
    A proof outline is given in~\cite{Goel2001}, we fill in a few details.
    
    Given a set $S$ let $\B{o}{r}$ be the minimum enclosing ball of $S$.
    Assume without loss of generality that $r=1$.
    Let $P=\{x\in\B{o}{r}|\|o-x\|_2=r\}$.
    Pick a set $R$ of $d+1$ points from $P$ in a way that the convex hull of $R$ contains $o$.
    One (expensive) way of doing this is to iterate through the points in $P$ and remove all points that do not shrink the minimum enclosing ball of the remaining points on removal.
    The data structure stores $R$.

    Given any query point $q$, let $t=\|o-q\|_2$.
    Let $x\in S$ be the actual furthest neighbor.
    We see that $\|x-q\|_2\leq 1+t$.
    If $o=q$, any point in $R$ is an exact furthest neighbor.
    Otherwise, consider the hyperplane passing through $o$ and perpendicular to the line defined by $q$ and $o$.
    Since $o$ is inside the convex hull of $R$, $R$ must contain at least one point, $x'$, on the side of the hyperplane not containing $q$.
    Consider the triangle defined by $x',o$ and $q$.
    (See Figure~\ref{fig:afn-sqrt2}).
    It is clear that $\|x'-q\|_2\geq \sqrt{t^2+1}$.
    Hence $\frac{\|x'-q\|_2}{\|x-q\|_2}\geq\frac{\sqrt{t^2+1}}{1+t}$.
    This is minimized at $1/\sqrt{2}$ when $t=1$, so
    $\|x'-q\|_2\geq\frac{\|x-q\|_2}{\sqrt{2}}$.
  \end{proof}
\end{theorem}

\listoffigures
\listoftables


\bibliographystyle{abbrv}
\bibliography{biblio}

\end{document}